\documentclass[12pt]{amsart} 
\usepackage[english]{babel}
\usepackage[utf8]{inputenc}
\usepackage{amsmath}
\usepackage{graphicx}
\usepackage{subfigure}
\usepackage{amssymb}
\usepackage{amsthm}
\usepackage{bm}
\newcommand{\nm}{\noalign{\smallskip}}
\newcommand{\ds}{\displaystyle}
\usepackage{tikz-cd}
\usepackage{mathrsfs}
\usepackage[colorinlistoftodos]{todonotes}
\usepackage{enumitem}
\usepackage{yfonts}
\usepackage{ dsfont }
\usepackage{bbm}
\usepackage{geometry}
\usepackage{setspace}
\usepackage{tikz}
\usepackage{pgfplots}
\usepackage{cite}
\pgfplotsset{compat=newest}
\usepgfplotslibrary{fillbetween}
\usetikzlibrary{positioning}
\usetikzlibrary{decorations.pathreplacing}
\usetikzlibrary{decorations,arrows}
\usetikzlibrary{decorations.markings}
\usetikzlibrary{patterns}
\usetikzlibrary{quotes,angles}
\usetikzlibrary{arrows}

\numberwithin{equation}{section}
\newtheorem{theorem}{Theorem}[section]
\newtheorem{lemma}[theorem]{Lemma}

\newtheorem{definition}[theorem]{Definition}
\newtheorem{example}[theorem]{Example}

\newtheorem{corollary}[theorem]{Corollary}

\newtheorem{remark}[theorem]{Remark}
\newtheorem{proposition}[theorem]{Proposition}
\newtheorem{assumption}[theorem]{Assumption}
\allowdisplaybreaks[4]

\newcommand*{\rom}[1]{\expandafter\@slowromancap\romannumeral #1@}

\title[Symmetry-protected interface modes]{Symmetry-protected interface modes bifurcated from double Dirac cones}

\author{Habib Ammari} %
\address[H. Ammari]{Department of Mathematics, ETH Z\"{u}rich, R\"{a}mistrasse 101, CH-8092 Z\"{u}rich, Switzerland}
\email{habib.ammari@math.ethz.ch}

\author{Jiayu Qiu} %
\address[J. Qiu]{Department of Mathematics, ETH Z\"{u}rich, R\"{a}mistrasse 101, CH-8092 Z\"{u}rich, Switzerland}
\email{jiayu.qiu@sam.math.ethz.ch}

\date{}

\begin{document}

\begin{abstract}
We rigorously prove the existence of interface modes in a sharp interface model, which bifurcate from the double Dirac cone as a consequence of the band inversion induced by super-symmetry breaking. The exact number of interface modes are determined. The proof is based on a discrete version of the layer-potential framework. Moreover, we prove that such interface modes are symmetry-protected against perturbations that respect the reflection symmetry.
\end{abstract}

\maketitle

\smallskip

\noindent \textbf{Keywords.}  Interface mode, band inversion, double Dirac cone, symmetry protection, robustness, symmetry breaking, discrete layer potential method \par 

\smallskip

\noindent \textbf{AMS subject classifications.} 35Q40, 35C20, 35P99, 82D25 
\\

\tableofcontents

\section{Introduction}

\subsection{Background} \label{sec_background}

Recent developments in topological insulators have opened new avenues for guiding waves robustly. For example, it is known from the bulk-edge correspondence (BEC) principle that the interface between two insulating media with distinct topological invariants, such as the gap Chern number or the Kane-Mele index, supports unidirectional propagating modes \cite{qizhang11topo_insulator,bernivig13topo_insulator}. As a consequence of their topological origin, such interface modes exhibit high robustness against disorder and impurities present in the system as long as the topological indices remain unchanged; therefore, these interface modes are referred to as \textit{topologically protected}. Remarkably, this unidirectional guiding of waves is also achieved in classical wave systems, not only limited to condensed matter materials, such as the photonic system \cite{ozawa19topological_photonics}.

However, the application of BEC to achieve robust guiding of waves is limited by the scarcity of nontrivial topological indices in practical systems. For example, in photonic media, the emergence of a nonvanishing Chern number generally requires breaking time-reversal symmetry through magneto-optical effects, which in turn necessitates strong external magnetic fields \cite{ozawa19topological_photonics}. This difficulty has motivated the investigation of alternative mechanisms that may produce robust interface modes without relying on conventional topological indices.

Remarkably, a mechanism of this type arises in the study of the quantum valley Hall effect. Its basic structure may be summarized as follows:
\begin{itemize}
    \item We begin with a periodic medium described by a single-particle Hamiltonian $\mathcal{H}$ that is invariant under a symmetry group $\mathcal{S}$. The presence of symmetry may enforce special spectral features, including degeneracies between adjacent Bloch bands. For instance, when $\mathcal{S}=C_{3v}$ (hexagonal symmetry), the band structure may exhibit conical degeneracies, i.e. the so-called Dirac points, at high-symmetry quasi-momenta \cite{fefferman12dirac}.
    \item Next, if we perturb the system by a symmetry-breaking term by replacing $\mathcal{H}$ with
    $$
    \mathcal{H}+\delta \mathcal{H}_{per},
    $$ 
    where $\mathcal{H}_{per}$ does not commute with (part of) $\mathcal{S}$, the degeneracy is lifted and a local spectral gap opens in a neighborhood of the former touching point (see Figure \ref{fig_gap_open} for an illustration). Importantly, for suitably chosen perturbations, the Bloch eigenspaces associated with the upper and lower bands near the gap exhibit opposite symmetry character for the Hamiltonians $\mathcal{H}+\delta \mathcal{H}_{per}$ and $\mathcal{H}-\delta \mathcal{H}_{per}$, as measured by the residual symmetry in $\mathcal{S}$. This phenomenon is commonly referred to as \textit{band inversion} \cite{vanderbilt2018berry}.
    
    \item Then we consider a heterostructure formed by adjoining the two bulk media governed by $\mathcal{H}+\delta \mathcal{H}_{per}$ and $\mathcal{H}-\delta \mathcal{H}_{per}$ along an interface. It is conjectured and supported by both formal and rigorous analyzes in specific settings that a weaker form of bulk–edge correspondence holds in this context: namely, that interface modes arise when the adjoining media exhibit band inversion. In this regime, the interface states may be viewed as bifurcating from the degenerate spectral point at which the original band touching occurs.
\end{itemize}
Compared with the classical BEC based on with nontrivial topological indices, the band-inversion mechanism is often more accessible in concrete models. For example, in photonic systems it may be realized through suitable geometric deformations of a periodic structure, without requiring the breaking of time-reversal symmetry by external magnetic fields \cite{ozawa19topological_photonics}. Although the resulting interface modes do not arise from a global topological invariant in the usual sense, they nevertheless exhibit a form of robustness. Of particular interest is the \textit{symmetry protection}: suppose that the bulk Hamiltonians $\mathcal{H}\pm\delta \mathcal{H}_{per}$, as well as the interface configuration interpolating between them, are invariant under a symmetry $\mathcal{O}$ belonging to the original symmetry group $\mathcal{S}$. Then the associated interface modes are expected to be robust against perturbations that preserve the symmetry $\mathcal{O}$ \cite{fu2011crystalline}. The purpose of the present work is to provide a rigorous analysis of this mechanism. We establish the existence of interface modes induced by band inversion for a class of discrete Hamiltonians and prove their robustness under symmetry-preserving perturbations.

From a mathematical perspective, the analysis of interface modes induced by band inversion originates in the work of Fefferman and Weinstein \cite{fefferman2017topologically}. In that paper, the authors establish the existence of interface modes in a one-dimensional \textit{domain-wall model}, which smoothly interpolates between two periodic Schrödinger operators exhibiting band inversion near a Dirac point. Their multiscale analysis has subsequently been extended to various two-dimensional settings, proving the existence interface modes bifurcating from Dirac points in honeycomb-type structures; see, for example, \cite{fefferman2016honeycomb_edge,lee2019elliptic,drouot2020edge}. These works remain within the framework of domain-wall models. In the context of chains of subwavelength resonators, Ammari et al. have adapted the fictitious source method and Toeplitz theory to study the interface modes \cite{SWP3,SWP4}. More recently, Qiu et al. have developed a layer-potential framework which applies to studying interface modes in \textit{sharp interface models}, in which two bulk media are directly adjoined without adiabatic modulation. This approach applies to interface modes bifurcating from Dirac or quadratic degeneracies in one- and two-dimensional systems \cite{qiu2026waveguide_localized,li2024interface_mode_honeycomb,qiu2024square_lattice}. We emphasize that, in the above works, the interface modes are derived solely from the band-inversion mechanism, without assuming the presence of a nontrivial global topological invariant. In particular, a rigorous analysis of the robustness of two-dimensional interface modes in this non-topological setting appears to be absent from the literature. By contrast, in certain special situations, band inversion gives rise to a nontrivial topological invariant. In such cases, the existence and robustness of interface modes follow from the (strong) BEC principle. The mathematical theory of this principle has been extensively developed in a variety of settings, including Schrödinger operators \cite{kellendonk2004quantization,taarabt2014equality,Combes2005edge_impurity,Bourne2016ChernNL,cornean2021landau+functional,ludewig2020shortrange+coarse,gontier2023edge_channel}, Dirac Hamiltonians \cite{bal2019dirac+functional,bal2023dirac+microlocal}, tight-binding Hamiltonians \cite{graf2018shortrange+transfer,avila2013shortrange+transfer,graf2013shortrange+scattering,graf2005equality,tauber2022chiral_finite_chain,qiu2025generalized,bourne2017ktheory,kubota2017ktheory,qiu2025bec_disorder_finite,drouot2024bec_curvedinterfaces,bols2018quantization}, and classical wave operators \cite{lin2022transfer,thiang2023transfer,ammari2024toeplitz_1,ammari2024toeplitz_2,
ammari2024toeplitz_3,SWP1,SWP2,SWP3,SWP4,qiu2025bec_finite,craster}; we also refer to the excellent monograph \cite{prodan2016ktheory} and the recent review \cite{bal2024review} for comprehensive accounts.

In this paper, we establish the existence of interface modes bifurcating from a particular type of spectral degeneracy in which two Dirac cones coincide due to an additional symmetry of the system, forming a \textit{double Dirac cone}. We show that, upon breaking this symmetry, a band inversion occurs and gives rise to localized interface states. Our analysis is carried out for sharp interface models governed by discrete Hamiltonians. The proof relies on a discrete version of the layer-potential framework, which enables us not only to construct the interface modes but also to determine their precise multiplicities. A principal result of this work is a rigorous proof of symmetry protection: we show that the interface modes are robust against perturbations that preserve the reflection symmetry of the interface configuration. To our knowledge, this provides the first fully rigorous analysis of robustness for two-dimensional interface modes arising from band inversion in the absence of a nontrivial global topological invariant. The supersymmetric structure supporting a double Dirac cone was originally proposed in the physics literature as a photonic analogue of the quantum spin Hall effect \cite{WuHu15scheme,ozawa19topological_photonics,leykam2026limitations}. From a mathematical standpoint, interface modes bifurcating from a double Dirac cone were recently analyzed in the context of domain-wall models by Zhu et al. \cite{cao2025edge_double_cone}, using multiscale techniques in the spirit of Fefferman and Weinstein, and in \cite{miao2025zero} for nearest-neighbor-hopping Hamiltonians. In contrast, our approach based on a discrete layer-potential method applies to sharp interface configurations and any range of hopping. We expect that this framework, and in particular the mechanism used to establish symmetry-protected robustness, may provide new insight into this field.

\subsection{Outline}

This paper is organized as follows:

\begin{itemize}
    \item In Section \ref{sec_main_result}, we first present the detailed setup, including the lattice model and the governing discrete Hamiltonians. We will illustrate the emergence of double Dirac cone when the Hamiltonian is super-symmetric (Theorem \ref{thm_double_cone}). Then we present our main results,  including the existence and precise number of interface modes (Theorem \ref{thm_existence_interface_modes}) as a consequence of the band inversion induced by super-symmetry breaking, and the robustness of such interface modes against the symmetry-preserving perturbations (Theorem \ref{thm_robustness_interface_modes}).
    \item In Section \ref{sec_representation_analysis}, we prove the emergence of the double Dirac cone stated in Theorem \ref{thm_double_cone} using the representation theory of symmetric groups.
    \item In Section \ref{sec_gap_open}, we present a detailed asymptotic expansion of the Floquet-Bloch eigenpairs of the symmetry-broken bulk Hamiltonians $\mathcal{H}_{\pm\delta}$; see Theorem \ref{thm_asymptotics_eigenpairs}. The implication of the asymptotic expansion is two-fold. First, it follows immediately that the double Dirac cone is lifted upon the symmetry-breaking perturbation, and a common band gap is opened between $\mathcal{H}_{\pm\delta}$ (see Theorem \ref{thm_gap_open}). More importantly, from the asymptotic expansion of Floquet-Bloch eigenfunctions, it will be clear that $\mathcal{H}_{\delta}$ switches its eigenspace with that of $\mathcal{H}_{-\delta}$ at the end points of the common band gap, i.e., a band inversion occurs. This phenomenon lies in the center of the bifurcation of interface modes and will be illustrated in more detail in the next section.
    \item In Section \ref{sec_sec5}, we prove the first main result, Theorem \ref{thm_existence_interface_modes}, on the existence of interface modes by a discrete version of the layer-potential framework. We first present some preliminaries, mainly on the properties of the physical Green operator and discrete energy flux; see Section \ref{sec_prelim_interface_modes}. Then in Section \ref{sec_layer_potential_framework} we establish  the layer-potential formulation, which translates the study of the interface modes into solving the characteristic value problem of boundary matching operators; see Propositions \ref{prop_existence_equivalence} and \ref{prop_number_equivalence}. To do this in the small-gap regime ($\delta\ll 1$), the key step is to investigate the limit of boundary matching operators, as achieved in Section \ref{sec_boundary_operator_limit}. Remarkably, the limit of boundary matching operators clearly manifests the band inversion phenomenon; see Remark \ref{rmk_absence_band_inversion}. Based on the information obtained in Section \ref{sec_boundary_operator_limit}, we solve the characteristic value of the boundary matching operators in Sections \ref{sec_existence_interface_mode} and \ref{sec_number_interface_mode}, which completes the proof of Theorem \ref{thm_existence_interface_modes}.
    \item In Section \ref{sec_robustness}, we prove Theorem \ref{thm_robustness_interface_modes} on the robustness of interface modes against symmetry-preserving perturbations. The main idea of the proof, as outlined in Section \ref{sec_robustness_main_idea}, is based on the periodic approximation method. Specifically, we will first construct the perturbed interface modes on a sequence of strips with increasing width using the Lyapunov-Schmidt reduction argument, which is valid due to the symmetry of perturbation Hamiltonian. Then we prove their weak convergence at infinity. The key of this procedure is a uniform estimate of the perturbed interface modes on strips, as proved in Section \ref{sec_uniform_norm_estimate}. Based on that, the weak convergence follows from a standard compactness argument adapted to our strip geometry, as detailed in Section \ref{sec_weak_compact}.
\end{itemize}

\section{Setup and Main Results} \label{sec_main_result}

We consider single-particle Hamiltonians on a triangular lattice with internal degrees of freedom (DOF) equal to six. Specifically, the lattice is given by
$$
\Lambda:=\mathbb{Z}\bm{\ell}_1\oplus \mathbb{Z}\bm{\ell}_2,\quad \bm{\ell}_1:=(\frac{\sqrt{3}}{2},\frac{1}{2})^{\top},\, \bm{\ell}_2:=(0,1)^{\top} .
$$

\begin{figure}
\begin{center}
\begin{tikzpicture}[scale=1.2]
\draw[dashed] ({-1},{0})--({1},{0});
\draw[dashed] ({-1/2},{sqrt(3)/2})--({1/2},{sqrt(3)/2});
\draw[dashed] ({-1/2},{-sqrt(3)/2})--({1/2},{-sqrt(3)/2});
\draw[dashed] ({-1},{0})--({-1/2},{sqrt(3)/2});
\draw[dashed] ({-1/2},{-sqrt(3)/2})--({1/2},{sqrt(3)/2});
\draw[dashed] ({1/2},{-sqrt(3)/2})--({1},{0});
\draw[dashed] ({1},{0})--({1/2},{sqrt(3)/2});
\draw[dashed] ({1/2},{-sqrt(3)/2})--({-1/2},{sqrt(3)/2});
\draw[dashed] ({-1/2},{-sqrt(3)/2})--({-1},{0});
\draw[fill=black,opacity=0.1] ({1/2},{sqrt(3)/6}) ellipse(0.15 and 0.15);
\node[scale=1.5] at ({1/2},{sqrt(3)/6}) {3};

\draw[fill=black,opacity=0.1] ({1/2},{-sqrt(3)/6}) ellipse(0.15 and 0.15);
\node[scale=1.5] at ({1/2},{-sqrt(3)/6}) {5};

\draw[fill=black,opacity=0.1] ({-1/2},{sqrt(3)/6}) ellipse(0.15 and 0.15);
\node[scale=1.5] at ({-1/2},{sqrt(3)/6}) {2};

\draw[fill=black,opacity=0.1] ({-1/2},{-sqrt(3)/6}) ellipse(0.15 and 0.15);
\node[scale=1.5] at ({-1/2},{-sqrt(3)/6}) {4};

\draw[fill=black,opacity=0.1] ({0},{sqrt(3)/3}) ellipse(0.15 and 0.15);
\node[scale=1.5] at ({0},{sqrt(3)/3}) {1};

\draw[fill=black,opacity=0.1] ({0},{-sqrt(3)/3}) ellipse(0.15 and 0.15);
\node[scale=1.5] at ({0},{-sqrt(3)/3}) {6};

\draw[dashed] ({-1+3/2},{0+sqrt(3)/2})--({1+3/2},{0+sqrt(3)/2});
\draw[dashed] ({-1/2+3/2},{sqrt(3)/2+sqrt(3)/2})--({1/2+3/2},{sqrt(3)/2+sqrt(3)/2});
\draw[dashed] ({-1/2+3/2},{-sqrt(3)/2+sqrt(3)/2})--({1/2+3/2},{-sqrt(3)/2+sqrt(3)/2});
\draw[dashed] ({-1+3/2},{0+sqrt(3)/2})--({-1/2+3/2},{sqrt(3)/2+sqrt(3)/2});
\draw[dashed] ({-1/2+3/2},{-sqrt(3)/2+sqrt(3)/2})--({1/2+3/2},{sqrt(3)/2+sqrt(3)/2});
\draw[dashed] ({1/2+3/2},{-sqrt(3)/2+sqrt(3)/2})--({1+3/2},{0+sqrt(3)/2});
\draw[dashed] ({1+3/2},{0+sqrt(3)/2})--({1/2+3/2},{sqrt(3)/2+sqrt(3)/2});
\draw[dashed] ({1/2+3/2},{-sqrt(3)/2+sqrt(3)/2})--({-1/2+3/2},{sqrt(3)/2+sqrt(3)/2});
\draw[dashed] ({-1/2+3/2},{-sqrt(3)/2+sqrt(3)/2})--({-1+3/2},{0+sqrt(3)/2});
\draw[fill=black,opacity=1] ({1/2+3/2},{sqrt(3)/6+sqrt(3)/2}) ellipse(0.15 and 0.15);
\draw[fill=black,opacity=1] ({1/2+3/2},{-sqrt(3)/6+sqrt(3)/2}) ellipse(0.15 and 0.15);
\draw[fill=black,opacity=1] ({-1/2+3/2},{sqrt(3)/6+sqrt(3)/2}) ellipse(0.15 and 0.15);
\draw[fill=black,opacity=1] ({-1/2+3/2},{-sqrt(3)/6+sqrt(3)/2}) ellipse(0.15 and 0.15);
\draw[fill=black,opacity=1] ({0+3/2},{sqrt(3)/3+sqrt(3)/2}) ellipse(0.15 and 0.15);
\draw[fill=black,opacity=1] ({0+3/2},{-sqrt(3)/3+sqrt(3)/2}) ellipse(0.15 and 0.15);

\draw[dashed] ({-1-3/2},{0-sqrt(3)/2})--({1-3/2},{0-sqrt(3)/2});
\draw[dashed] ({-1/2-3/2},{sqrt(3)/2-sqrt(3)/2})--({1/2-3/2},{sqrt(3)/2-sqrt(3)/2});
\draw[dashed] ({-1/2-3/2},{-sqrt(3)/2-sqrt(3)/2})--({1/2-3/2},{-sqrt(3)/2-sqrt(3)/2});
\draw[dashed] ({-1-3/2},{0-sqrt(3)/2})--({-1/2-3/2},{sqrt(3)/2-sqrt(3)/2});
\draw[dashed] ({-1/2-3/2},{-sqrt(3)/2-sqrt(3)/2})--({1/2-3/2},{sqrt(3)/2-sqrt(3)/2});
\draw[dashed] ({1/2-3/2},{-sqrt(3)/2-sqrt(3)/2})--({1-3/2},{0-sqrt(3)/2});
\draw[dashed] ({1-3/2},{0-sqrt(3)/2})--({1/2-3/2},{sqrt(3)/2-sqrt(3)/2});
\draw[dashed] ({1/2-3/2},{-sqrt(3)/2-sqrt(3)/2})--({-1/2-3/2},{sqrt(3)/2-sqrt(3)/2});
\draw[dashed] ({-1/2-3/2},{-sqrt(3)/2-sqrt(3)/2})--({-1-3/2},{0-sqrt(3)/2});
\draw[fill=black,opacity=1] ({1/2-3/2},{sqrt(3)/6-sqrt(3)/2}) ellipse(0.15 and 0.15);
\draw[fill=black,opacity=1] ({1/2-3/2},{-sqrt(3)/6-sqrt(3)/2}) ellipse(0.15 and 0.15);
\draw[fill=black,opacity=1] ({-1/2-3/2},{sqrt(3)/6-sqrt(3)/2}) ellipse(0.15 and 0.15);
\draw[fill=black,opacity=1] ({-1/2-3/2},{-sqrt(3)/6-sqrt(3)/2}) ellipse(0.15 and 0.15);
\draw[fill=black,opacity=1] ({0-3/2},{sqrt(3)/3-sqrt(3)/2}) ellipse(0.15 and 0.15);
\draw[fill=black,opacity=1] ({0-3/2},{-sqrt(3)/3-sqrt(3)/2}) ellipse(0.15 and 0.15);

\draw[dashed] ({-1-3/2},{0+sqrt(3)/2})--({1-3/2},{0+sqrt(3)/2});
\draw[dashed] ({-1/2-3/2},{sqrt(3)/2+sqrt(3)/2})--({1/2-3/2},{sqrt(3)/2+sqrt(3)/2});
\draw[dashed] ({-1/2-3/2},{-sqrt(3)/2+sqrt(3)/2})--({1/2-3/2},{-sqrt(3)/2+sqrt(3)/2});
\draw[dashed] ({-1-3/2},{0+sqrt(3)/2})--({-1/2-3/2},{sqrt(3)/2+sqrt(3)/2});
\draw[dashed] ({-1/2-3/2},{-sqrt(3)/2+sqrt(3)/2})--({1/2-3/2},{sqrt(3)/2+sqrt(3)/2});
\draw[dashed] ({1/2-3/2},{-sqrt(3)/2+sqrt(3)/2})--({1-3/2},{0+sqrt(3)/2});
\draw[dashed] ({1-3/2},{0+sqrt(3)/2})--({1/2-3/2},{sqrt(3)/2+sqrt(3)/2});
\draw[dashed] ({1/2-3/2},{-sqrt(3)/2+sqrt(3)/2})--({-1/2-3/2},{sqrt(3)/2+sqrt(3)/2});
\draw[dashed] ({-1/2-3/2},{-sqrt(3)/2+sqrt(3)/2})--({-1-3/2},{0+sqrt(3)/2});
\draw[fill=black,opacity=1] ({1/2-3/2},{sqrt(3)/6+sqrt(3)/2}) ellipse(0.15 and 0.15);
\draw[fill=black,opacity=1] ({1/2-3/2},{-sqrt(3)/6+sqrt(3)/2}) ellipse(0.15 and 0.15);
\draw[fill=black,opacity=1] ({-1/2-3/2},{sqrt(3)/6+sqrt(3)/2}) ellipse(0.15 and 0.15);
\draw[fill=black,opacity=1] ({-1/2-3/2},{-sqrt(3)/6+sqrt(3)/2}) ellipse(0.15 and 0.15);
\draw[fill=black,opacity=1] ({0-3/2},{sqrt(3)/3+sqrt(3)/2}) ellipse(0.15 and 0.15);
\draw[fill=black,opacity=1] ({0-3/2},{-sqrt(3)/3+sqrt(3)/2}) ellipse(0.15 and 0.15);

\draw[dashed] ({-1+3/2},{0-sqrt(3)/2})--({1+3/2},{0-sqrt(3)/2});
\draw[dashed] ({-1/2+3/2},{sqrt(3)/2-sqrt(3)/2})--({1/2+3/2},{sqrt(3)/2-sqrt(3)/2});
\draw[dashed] ({-1/2+3/2},{-sqrt(3)/2-sqrt(3)/2})--({1/2+3/2},{-sqrt(3)/2-sqrt(3)/2});
\draw[dashed] ({-1+3/2},{0-sqrt(3)/2})--({-1/2+3/2},{sqrt(3)/2-sqrt(3)/2});
\draw[dashed] ({-1/2+3/2},{-sqrt(3)/2-sqrt(3)/2})--({1/2+3/2},{sqrt(3)/2-sqrt(3)/2});
\draw[dashed] ({1/2+3/2},{-sqrt(3)/2-sqrt(3)/2})--({1+3/2},{0-sqrt(3)/2});
\draw[dashed] ({1+3/2},{0-sqrt(3)/2})--({1/2+3/2},{sqrt(3)/2-sqrt(3)/2});
\draw[dashed] ({1/2+3/2},{-sqrt(3)/2-sqrt(3)/2})--({-1/2+3/2},{sqrt(3)/2-sqrt(3)/2});
\draw[dashed] ({-1/2+3/2},{-sqrt(3)/2-sqrt(3)/2})--({-1+3/2},{0-sqrt(3)/2});
\draw[fill=black,opacity=1] ({1/2+3/2},{sqrt(3)/6-sqrt(3)/2}) ellipse(0.15 and 0.15);
\draw[fill=black,opacity=1] ({1/2+3/2},{-sqrt(3)/6-sqrt(3)/2}) ellipse(0.15 and 0.15);
\draw[fill=black,opacity=1] ({-1/2+3/2},{sqrt(3)/6-sqrt(3)/2}) ellipse(0.15 and 0.15);
\draw[fill=black,opacity=1] ({-1/2+3/2},{-sqrt(3)/6-sqrt(3)/2}) ellipse(0.15 and 0.15);
\draw[fill=black,opacity=1] ({0+3/2},{sqrt(3)/3-sqrt(3)/2}) ellipse(0.15 and 0.15);
\draw[fill=black,opacity=1] ({0+3/2},{-sqrt(3)/3-sqrt(3)/2}) ellipse(0.15 and 0.15);

\draw[dashed] ({-1},{0+sqrt(3)})--({1},{0+sqrt(3)});
\draw[dashed] ({-1/2},{sqrt(3)/2+sqrt(3)})--({1/2},{sqrt(3)/2+sqrt(3)});
\draw[dashed] ({-1/2},{-sqrt(3)/2+sqrt(3)})--({1/2},{-sqrt(3)/2+sqrt(3)});
\draw[dashed] ({-1},{0+sqrt(3)})--({-1/2},{sqrt(3)/2+sqrt(3)});
\draw[dashed] ({-1/2},{-sqrt(3)/2+sqrt(3)})--({1/2},{sqrt(3)/2+sqrt(3)});
\draw[dashed] ({1/2},{-sqrt(3)/2+sqrt(3)})--({1},{0+sqrt(3)});
\draw[dashed] ({1},{0+sqrt(3)})--({1/2},{sqrt(3)/2+sqrt(3)});
\draw[dashed] ({1/2},{-sqrt(3)/2+sqrt(3)})--({-1/2},{sqrt(3)/2+sqrt(3)});
\draw[dashed] ({-1/2},{-sqrt(3)/2+sqrt(3)})--({-1},{0+sqrt(3)});
\draw[fill=black,opacity=1] ({1/2},{sqrt(3)/6+sqrt(3)}) ellipse(0.15 and 0.15);
\draw[fill=black,opacity=1] ({1/2},{-sqrt(3)/6+sqrt(3)}) ellipse(0.15 and 0.15);
\draw[fill=black,opacity=1] ({-1/2},{sqrt(3)/6+sqrt(3)}) ellipse(0.15 and 0.15);
\draw[fill=black,opacity=1] ({-1/2},{-sqrt(3)/6+sqrt(3)}) ellipse(0.15 and 0.15);
\draw[fill=black,opacity=1] ({0},{sqrt(3)/3+sqrt(3)}) ellipse(0.15 and 0.15);
\draw[fill=black,opacity=1] ({0},{-sqrt(3)/3+sqrt(3)}) ellipse(0.15 and 0.15);

\draw[dashed] ({-1},{0-sqrt(3)})--({1},{0-sqrt(3)});
\draw[dashed] ({-1/2},{sqrt(3)/2-sqrt(3)})--({1/2},{sqrt(3)/2-sqrt(3)});
\draw[dashed] ({-1/2},{-sqrt(3)/2-sqrt(3)})--({1/2},{-sqrt(3)/2-sqrt(3)});
\draw[dashed] ({-1},{0-sqrt(3)})--({-1/2},{sqrt(3)/2-sqrt(3)});
\draw[dashed] ({-1/2},{-sqrt(3)/2-sqrt(3)})--({1/2},{sqrt(3)/2-sqrt(3)});
\draw[dashed] ({1/2},{-sqrt(3)/2-sqrt(3)})--({1},{0-sqrt(3)});
\draw[dashed] ({1},{0-sqrt(3)})--({1/2},{sqrt(3)/2-sqrt(3)});
\draw[dashed] ({1/2},{-sqrt(3)/2-sqrt(3)})--({-1/2},{sqrt(3)/2-sqrt(3)});
\draw[dashed] ({-1/2},{-sqrt(3)/2-sqrt(3)})--({-1},{0-sqrt(3)});
\draw[fill=black,opacity=1] ({1/2},{sqrt(3)/6-sqrt(3)}) ellipse(0.15 and 0.15);
\draw[fill=black,opacity=1] ({1/2},{-sqrt(3)/6-sqrt(3)}) ellipse(0.15 and 0.15);
\draw[fill=black,opacity=1] ({-1/2},{sqrt(3)/6-sqrt(3)}) ellipse(0.15 and 0.15);
\draw[fill=black,opacity=1] ({-1/2},{-sqrt(3)/6-sqrt(3)}) ellipse(0.15 and 0.15);
\draw[fill=black,opacity=1] ({0},{sqrt(3)/3-sqrt(3)}) ellipse(0.15 and 0.15);
\draw[fill=black,opacity=1] ({0},{-sqrt(3)/3-sqrt(3)}) ellipse(0.15 and 0.15);

\draw[thick,blue] ({-3/4},{-3*sqrt(3)/4})--({-3/4},{sqrt(3)/4});
\draw[thick,blue] ({3/4},{3*sqrt(3)/4})--({3/4},{-sqrt(3)/4});
\draw[thick,blue] ({3/4},{3*sqrt(3)/4})--({-3/4},{sqrt(3)/4});
\draw[thick,blue] ({-3/4},{-3*sqrt(3)/4})--({3/4},{-sqrt(3)/4});

\draw[very thick,->,red] ({0},{0})--({0},{sqrt(3)});
\draw[very thick,->,red] ({0},{0})--({3/2},{sqrt(3)/2});
\draw[very thick,->,yellow] ({0},{0})--({1},{0});
\node[above,red,scale=0.9] at ({0},{sqrt(3)}) {$\ell_2$};
\node[above,red,scale=0.9] at ({3/2},{sqrt(3)/2}) {$\ell_1$};
\node[right,red,scale=0.9] at ({1},{0}) {$\tilde{\ell}$};
\end{tikzpicture}
\caption{A super-symmetric periodic structure, which contains six sublattices (illustrated by the six sites in a single unit cell). The numbers in the first cell label the ordering of the sublattices.}
\label{fig_periodic_structure}
\end{center}
\end{figure}
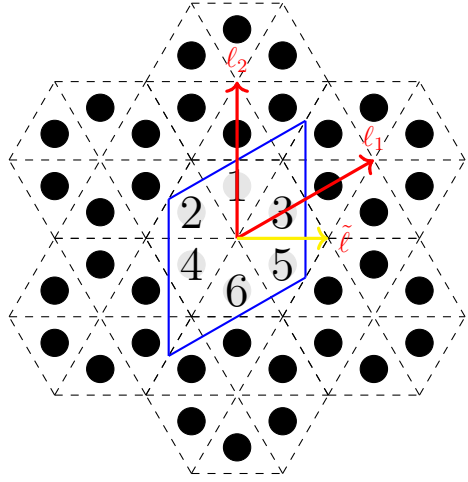

\noindent Naturally, the state space is
\begin{equation*}
\mathcal{X}:=\ell^2(\Lambda)\otimes \mathbb{C}^6.
\end{equation*}
Throughout this paper, we refer to the internal DOF as the sublattice DOF, as illustrated pictorially in Figure \ref{fig_periodic_structure}.\footnote{In fact, it can also stand for the spin DOF. The reason we prefer to interpret the internal DOF as the sublattices is that, in this case, a super-symmetry breaking perturbation is easily designed by geometrically deforming the sublattices; see Figure \ref{fig_unit_structure} for details.} The Hamilonians are bounded self-adjoint operators on $\mathcal{X}$. Throughout this paper, we assume that the Hamiltonians are short-range in the following sense. 
\begin{assumption} \label{asmp_short_range}
The Hamiltonian $\mathcal{H}$ is said to be \textit{of range $N$} for some $N>0$ if its kernel satisfies
\begin{equation*}
\mathcal{H}(\bm{n},\bm{m})=0\quad \text{for any $\bm{n},\bm{m}\in\Lambda$ with $\|\bm{n}-\bm{m}\|_{2}> N$},
\end{equation*}
where $\|\cdot\|_2$ denotes the $2-$norm on $\mathbb{R}^2$. All the Hamiltonians appearing in this paper are assumed to be of the same range $N$.
\end{assumption}
Note that the discussion in this paper is not confined to electronic systems: even though the operators are called ``Hamiltonians", they can also be referred to as the short-range governing operators for general physical systems, such as the capacitance operator in subwavelength phononic and photonic systems \cite{ammari2025from_conden_to_subwave}.

Now we consider a Hamiltonian, $\mathcal{H}_{b}$, that is symmetric with respect to the $C_{6v}$ point group. To be precise, we consider the symmetry group generated by two elements:
\begin{equation*}
\mathcal{S}=\langle \mathcal{R}_6,\mathcal{F}_x :\mathcal{R}_6^6=F_x^2=id,\,\mathcal{R}_6\mathcal{F}_x=\mathcal{F}_x\mathcal{R}_6^{-1} \rangle ,
\end{equation*}
where $\mathcal{R}_6$ and $\mathcal{F}_x$ denote the $\frac{\pi}{3}-$rotation and $x-$axis reflection operators (acting both on the lattice and internal DOFs). By introducing the following matrices, 
\begin{equation*}
\begin{aligned}
&R_6^{ext}:=
\begin{pmatrix}
\frac{1}{2} & -\frac{\sqrt{3}}{2} \\
\frac{\sqrt{3}}{2} & \frac{1}{2} 
\end{pmatrix}\in\mathbb{R}^{2\times 2},\quad
R_6^{int}:= e_{13}+e_{21}+e_{35}+e_{42}+e_{56}+e_{64} \in\mathbb{R}^{6\times 6}, \\
&F_x^{ext}:=
\begin{pmatrix}
1 & 0 \\
0 & -1 
\end{pmatrix}\in\mathbb{R}^{2\times 2},\quad
F_x^{int}:= e_{16}+e_{24}+e_{35}+e_{42}+e_{53}+e_{61} \in\mathbb{R}^{6\times 6},
\end{aligned}
\end{equation*}
where the element $(i,j)$ of $e_{i,j}$ is equal to one, while all the others are zero, $\mathcal{R}_6$ and $\mathcal{F}_x$ are defined as
\begin{equation*}
\begin{aligned}
(\mathcal{R}_6u)(\bm{n}):=R_6^{int}\big[u(R_6^{ext}\bm{n})\big],\quad
(\mathcal{F}_xu)(\bm{n}):=F_x^{int}\big[u(F_x^{ext}\bm{n})\big].
\end{aligned}
\end{equation*}
Then the symmetry of $\mathcal{H}_b$ is summarized as the following commutation relations.

\begin{assumption} \label{assum_C6v_symmetry}
$\mathcal{H}_b$ is translational invariant with respect to $\Lambda$, that is, $$\mathcal{H}_b(\bm{n}+\bm{e},\bm{m}+\bm{e}) = \mathcal{H}_b(\bm{n},\bm{m})$$ for all $\bm{e}\in\Lambda$. Moreover, for each $\mathcal{O}\in \mathcal{S}$, it holds that $[\mathcal{O},\mathcal{H}_b]=0$.
\end{assumption}

\begin{example} \label{examp_super_symmetric_Hamiltonian}
A naive construction of a Hamiltonian that satisfies Assumption \ref{assum_C6v_symmetry} is by taking the constant nearest-neighbor hopping model on the periodic structure shown in Figure \ref{fig_periodic_structure}. To be precise, the following map is introduced:
\begin{equation*}
\Xi:\Lambda\times \{1,2,3,4,5,6\} \to \mathbb{R}^2,\quad \Xi(\bm{n},i):= \bm{n}+\bm{d}_i
\end{equation*}
with
\begin{equation*}
\bm{d}_1:=\frac{1}{3}\bm{\ell}_2,\quad \bm{d}_2:=-\frac{1}{3}\bm{\ell}_1+\frac{1}{3}\bm{\ell}_2,\quad \bm{d}_3:=\frac{1}{3}\bm{\ell}_1,\quad
\bm{d}_4:=-\frac{1}{3}\bm{\ell}_1,\quad \bm{d}_5:=\frac{1}{3}\bm{\ell}_1-\frac{1}{3}\bm{\ell}_2 ,\quad
\bm{d}_6:=-\frac{1}{3}\bm{\ell}_2 .
\end{equation*}
Then the Hamiltonian with its kernel $\mathcal{H}(\bm{n},\bm{m})=\Big(\mathcal{H}(\bm{n},\bm{m})_{i,j}\Big)_{1\leq i,j\leq 6}$ specified as follows clearly satisfies Assumption \ref{assum_C6v_symmetry}:
\begin{equation*}
\mathcal{H}_b(\bm{n},\bm{m})_{i,j}=\left\{
\begin{aligned}
& 1,\quad \|\Xi(\bm{n},i)-\Xi(\bm{m},j)\|_{2}=\frac{1}{3}, \\
&0,\quad \text{otherwise}.
\end{aligned}
\right.
\end{equation*}
\end{example}

The symmetry property of $\mathcal{H}_b$ leads to its special spectral structure, which is discussed in detail in the following paragraphs. Let $\Lambda^*:=\mathbb{Z}\bm{\ell}_1^*\oplus \mathbb{Z}\bm{\ell}_2^*$ be the dual lattice of $\Lambda$, where $\bm{\ell}_1^*:=2\pi(\frac{2\sqrt{3}}{3},0)^{\top}$, $\bm{\ell}_2^*:=2\pi(-\frac{\sqrt{3}}{3},1)^{\top}$. The unit cell of $\Lambda^*$ (Brillouin zone) is taken as $$Y^*:=\{s\cdot \bm{\ell}_1^*+t\cdot \bm{\ell}_2^*: -\frac{1}{2}\leq s,t\leq \frac{1}{2}\}.$$ For each $\bm{\kappa}=\kappa_1\bm{\ell}_1^*+\kappa_2\bm{\ell}_2^*$, we denote the Floquet transform of $\mathcal{H}_b$ by $\mathcal{H}_b(\bm{\kappa})$, i.e., the restriction of $\mathcal{H}_b$ to the space
\begin{equation} \label{eq_quasi_periodic_space_unit_cell}
\mathcal{X}_{\bm{\kappa}}:=\{u\in \ell^2_{loc}(\Lambda)\otimes \mathbb{C}^6:\, u(\bm{n}+\bm{e})=e^{i\bm{e}\cdot\bm{\kappa}}u(\bm{n}),\, \forall \bm{e}\in\Lambda\}\simeq \mathbb{C}^6,
\end{equation}
equipped with the inner product
\begin{equation*}
(u,v)_{\mathcal{X}_{\bm{\kappa}}}:=(u(\bm{0}),v(\bm{0}))_{\mathbb{C}^6}.
\end{equation*}
Then the symmetry properties of $\mathcal{H}_b$ find their counterparts in the momentum space:

\begin{equation} \label{eq_symm_momen_sapce_1}
\mathcal{H}_b(\bm{\kappa}+n_1\cdot \bm{\ell}_1^*+n_2\cdot \bm{\ell}_2^*)=\mathcal{H}_b(\bm{\kappa}),\quad 
n_1,n_2\in\mathbb{Z},
\end{equation}

\begin{equation} \label{eq_symm_momen_sapce_2}
\mathcal{R}_6\mathcal{H}_b(\bm{\kappa})\mathcal{R}_6^{-1}=\mathcal{H}_b((R_6^{ext})^{-1}\bm{\kappa}),
\end{equation}
and
\begin{equation} \label{eq_symm_momen_sapce_3}
\mathcal{F}_x\mathcal{H}_b(\bm{\kappa})\mathcal{F}_x^{-1}=\mathcal{H}_b((F_x^{ext})^{-1}\bm{\kappa}).
\end{equation}

The fixed points of the symmetry operations \eqref{eq_symm_momen_sapce_1}--\eqref{eq_symm_momen_sapce_3} in momentum space are the so-called \textit{high symmetry points}, where interesting spectral properties arise. We focus on the high symmetry point $(0,0)$ (i.e., the $\Gamma$ point in physics literature) and the band structure nearby. Since the internal DOF equals six, the spectrum of $\mathcal{H}_b$ consists of six bands, i.e., $\text{Spec}(\mathcal{H}_b(\bm{\kappa}))=\{\lambda_j(\bm{\kappa})\}_{1\leq j\leq 6}$. Remarkably, since $\mathcal{X}_{\bm{0}}$ is invariant under the group $\mathcal{S}$ and due to Assumption \ref{assum_C6v_symmetry}, the eigenspaces of $\mathcal{H}_{b}(\bm{0})$ are decomposed corresponding to different representations of $\mathcal{S}$. In the physics literature, a degenerate eigenvalue $\lambda_*$ of $\mathcal{H}_{b}(\bm{0})$ is called \textit{generically degenerate} if its eigenfunctions form a basis of a multi-dimensional irreducible representations (irrep) of $\mathcal{S}$. It is well-known that $\mathcal{S}$ has two distinct two-dimensional irreps, which indicates the even degeneracy in $\text{Spec}(\mathcal{H}_b(\bm{0}))$. The following proposition, which is proved in Section \ref{sec_local_flat} following the standard perturbation theory (see, e.g., \cite{berkolaiko2018symmetry}), claims that the dispersion surface is locally flat near a doubly degenerate eigenvalue.
\begin{proposition} \label{prop_local_flat_double_eigenvalue}
Assume that $\lambda_*\in \text{Spec}(\mathcal{H}_b(\bm{0}))$ with multiplicity two. Let $\lambda_1(\bm{\kappa})$ and $\lambda_2(\bm{\kappa})$ be two branches of dispersion surfaces such that $\lambda_*=\lambda_1(\bm{\kappa}_*)=\lambda_2(\bm{\kappa}_*)$, and let $u_{n}(\cdot;\bm{\kappa})$ ($n=1,2$) be the corresponding Floquet-Bloch eigenfunctions. If $\{u_{n}(\cdot;\bm{0})\}_{n=1,2}$ form a basis of a two-dimensional irrep of $\mathcal{S}$, then we have
\begin{equation*}
\nabla \lambda_n\big|_{\bm{\kappa}=\bm{0}}=0,\quad n=1,2.
\end{equation*}
\end{proposition}

According to Proposition \ref{prop_local_flat_double_eigenvalue}, the possibility of some special band structure is ruled out: for example, a single Dirac cone (i.e., two adjacent dispersion surfaces touch conically) cannot exist in the bands of $\mathcal{H}_b$ at the $\Gamma$ point. Hence, a natural question is whether other types of interesting band structures could exist. In the sequel, we show that it is indeed possible: by imposing a further symmetry on the Hamiltonian $\mathcal{H}_b$, \textit{a double-cone structure exists when the eigenvalue degeneracy is greater than two.} To be more specific, we introduce the following \textit{super-symmetry operator $\tilde{\mathcal{T}}$} defined as the linear extension corresponding to the following rules on the basis of $\mathcal{X}$:
\begin{equation} \label{eq_super_symmetry}
\tilde{\mathcal{T}}\big(\mathbbm{1}_{\bm{n}}\otimes \bm{v}\big)
:=\mathbbm{1}_{\bm{n}}\otimes (e_{32}+e_{54})\bm{v}+\mathbbm{1}_{\bm{n}+\bm{\ell}_1}\otimes (e_{41}+e_{63})\bm{v}+\mathbbm{1}_{\bm{n}+\bm{\ell}_1-\bm{\ell}_2}\otimes (e_{15}+e_{26})\bm{v}.
\end{equation}
Intuitively, $\tilde{\mathcal{T}}$ can be visualized as a translation operator along an additional vector $\tilde{\bm{\ell}}=(\frac{\sqrt{3}}{3},0)^{\top}$, as shown in Figure \ref{fig_periodic_structure}. In particular, thanks to the $\mathbb{Z}\tilde{\bm{\ell}}$-translation invariance of the structure in Figure \ref{fig_periodic_structure}, one can check that the toy Hamiltonian constructed in Example \ref{examp_super_symmetric_Hamiltonian} is actually $\tilde{\mathcal{T}}$-symmetric.

Remarkably, as we will see in Section \ref{sec_double_cone}, \textit{the symmetry group generated by the $C_{6v}$ group $\mathcal{S}$ and the super-symmetry operator $\tilde{\mathcal{T}}$ possesses a four-dimensional irrep}, which is built by pinching the two-dimensional irreps of $\mathcal{S}$. This leads to the possibility of a degenerate eigenvalue $\lambda_*$ in $\text{Spec}(\mathcal{H}_b)$ with multiplicity equal to four if $\mathcal{H}_b$ is super-symmetric. Near such a degenerate point, a double Dirac cone structure can actually appear.
\begin{theorem} \label{thm_double_cone}
Assume that $[\mathcal{O},\mathcal{H}_b]=0$ for $\mathcal{O}\in\mathcal{S}\cup \{\tilde{\mathcal{T}}\}$ and $\lambda_*\in \text{Spec}(\mathcal{H}_b(\bm{0}))$ with multiplicity equal to four. Let $\lambda_n(\bm{\kappa})$ ($n=1,2,3,4$) be the four branches of dispersion surfaces such that $\lambda_*=\lambda_n(\bm{\kappa}_*)$ and $u_{n}(\cdot;\bm{\kappa})$ ($n=1,2$) be the corresponding Floquet-Bloch eigenfunctions. If the eigenfunctions $u_{n}(\cdot;\bm{0})$ form a basis of the four-dimensional irrep of $\tilde{\mathcal{S}}$ and the following condition holds:
\begin{equation} \label{eq_slope_cond}
\alpha_*:=\Big(u_{1}(\cdot;\bm{0}),\frac{\partial \mathcal{H}_b}{\partial \kappa_1}(\bm{0})u_{3}(\cdot;\bm{0}) \Big)_{\mathcal{X}_{\bm{0}}}\neq 0,
\end{equation}
then we have
\begin{equation*}
\begin{aligned}
&\lambda_1(\bm{\kappa})=-\frac{\sqrt{3}}{2}|\alpha_*||\bm{\kappa}|+\mathcal{O}(|\bm{\kappa}|^2), \quad \lambda_2(\bm{\kappa})=-\frac{\sqrt{3}}{2}|\alpha_*||\bm{\kappa}|+\mathcal{O}(|\bm{\kappa}|^2), \\
&\lambda_3(\bm{\kappa})=\frac{\sqrt{3}}{2}|\alpha_*||\bm{\kappa}|+\mathcal{O}(|\bm{\kappa}|^2), 
\quad\lambda_4(\bm{\kappa})=\frac{\sqrt{3}}{2}|\alpha_*||\bm{\kappa}|+\mathcal{O}(|\bm{\kappa}|^2).
\end{aligned}
\end{equation*}
\end{theorem}
Note that Theorem \ref{thm_double_cone}, especially the validity of condition \eqref{eq_slope_cond}, has been proved for generic continuous Schrödinger Hamiltonians in \cite{cao2023double_cone} based on an analytic extension argument and also for the capacitance operator in subwavelength physics \cite{miao2024subwave_double_cone}. Our main focus in Theorem \ref{thm_double_cone} is the role played by symmetries, without making any effort to verify \eqref{eq_slope_cond}; see the details in Section \ref{sec_double_cone}.

The main focus of this paper is \textit{the emergence of in-gap localized modes if we break the $\tilde{\mathcal{T}}$ symmetry of $\mathcal{H}_b$ at the two sides of an interface}. To start, we introduce the following perturbed Hamiltonians with $\delta>0$:
\begin{equation} \label{eq_perturbed_bulk_Hamiltonian}
\mathcal{H}_{\pm\delta}:=\mathcal{H}_{b}\pm\delta \mathcal{H}_{per} \in \mathcal{B}(\mathcal{X}),
\end{equation}
where the perturbation $\mathcal{H}_{per}$ satisfies the following property.
\begin{assumption} \label{assum_symmetry_break}
$\mathcal{H}_{per}$ is translational invariant with respect to $\Lambda$. Moreover, it holds that
\begin{equation*}
[\mathcal{O},\mathcal{H}_{per}]=0,\quad
[\tilde{\mathcal{T}},\mathcal{H}_{per}]\neq 0,\quad
\forall \mathcal{O}\in \mathcal{S}.
\end{equation*}
\end{assumption}
Breaking of the $\tilde{\mathcal{T}}$-symmetry can be achieved by manipulating the relative position between the sublattices, such as the expansion and shrinking operations shown in Figure \ref{fig_unit_structure}, and adjusting the hopping coefficient (of the toy Hamiltonian in Example \ref{examp_super_symmetric_Hamiltonian}) according to the change of relative spacing between lattice sites; we refer to \cite{miao2025zero} for a detailed construction.

\begin{figure}
\centering
\subfigure[]{
\label{fig_unit_structure}     
\begin{tikzpicture}[scale=1]
\draw[dashed] ({-1/2},{sqrt(3)/2})--({1/2},{sqrt(3)/2});
\draw[dashed] ({-1/2},{-sqrt(3)/2})--({1/2},{-sqrt(3)/2});
\draw[dashed] ({-1},{0})--({-1/2},{sqrt(3)/2});
\draw[dashed] ({1/2},{-sqrt(3)/2})--({1},{0});
\draw[dashed] ({1},{0})--({1/2},{sqrt(3)/2});
\draw[dashed] ({-1/2},{-sqrt(3)/2})--({-1},{0});
\draw[fill=black,opacity=0.1] ({1/2},{sqrt(3)/6}) ellipse(0.1 and 0.1);
\draw[fill=black,opacity=0.1] ({1/2},{-sqrt(3)/6}) ellipse(0.1 and 0.1);
\draw[fill=black,opacity=0.1] ({-1/2},{sqrt(3)/6}) ellipse(0.1 and 0.1);
\draw[fill=black,opacity=0.1] ({-1/2},{-sqrt(3)/6}) ellipse(0.1 and 0.1);
\draw[fill=black,opacity=0.1] ({0},{sqrt(3)/3}) ellipse(0.1 and 0.1);
\draw[fill=black,opacity=0.1] ({0},{-sqrt(3)/3}) ellipse(0.1 and 0.1);
\draw[fill=black,opacity=1] ({1.2*1/2},{1.2*sqrt(3)/6}) ellipse(0.1 and 0.1);
\draw[fill=black,opacity=1] ({1.2*1/2},{-1.2*sqrt(3)/6}) ellipse(0.1 and 0.1);
\draw[fill=black,opacity=1] ({-1.2*1/2},{1.2*sqrt(3)/6}) ellipse(0.1 and 0.1);
\draw[fill=black,opacity=1] ({-1.2*1/2},{-1.2*sqrt(3)/6}) ellipse(0.1 and 0.1);
\draw[fill=black,opacity=1] ({0},{1.2*sqrt(3)/3}) ellipse(0.1 and 0.1);
\draw[fill=black,opacity=1] ({0},{-1.2*sqrt(3)/3}) ellipse(0.1 and 0.1);
\draw[decorate,decoration={brace}] ({0.12},{1.2*sqrt(3)/3}) -- ({0.12},{sqrt(3)/3});
\node[right,scale=0.5] at ({0.15},{1.1*sqrt(3)/3}) {$d_+$};

\draw[dashed] ({-1/2},{sqrt(3)/2-3})--({1/2},{sqrt(3)/2-3});
\draw[dashed] ({-1/2},{-sqrt(3)/2-3})--({1/2},{-sqrt(3)/2-3});
\draw[dashed] ({-1},{0-3})--({-1/2},{sqrt(3)/2-3});
\draw[dashed] ({1/2},{-sqrt(3)/2-3})--({1},{0-3});
\draw[dashed] ({1},{0-3})--({1/2},{sqrt(3)/2-3});
\draw[dashed] ({-1/2},{-sqrt(3)/2-3})--({-1},{0-3});
\draw[fill=black,opacity=0.1] ({1/2},{sqrt(3)/6-3}) ellipse(0.1 and 0.1);
\draw[fill=black,opacity=0.1] ({1/2},{-sqrt(3)/6-3}) ellipse(0.1 and 0.1);
\draw[fill=black,opacity=0.1] ({-1/2},{sqrt(3)/6-3}) ellipse(0.1 and 0.1);
\draw[fill=black,opacity=0.1] ({-1/2},{-sqrt(3)/6-3}) ellipse(0.1 and 0.1);
\draw[fill=black,opacity=0.1] ({0},{sqrt(3)/3-3}) ellipse(0.1 and 0.1);
\draw[fill=black,opacity=0.1] ({0},{-sqrt(3)/3-3}) ellipse(0.1 and 0.1);
\draw[fill=black,opacity=1] ({0.6*1/2},{0.6*sqrt(3)/6-3}) ellipse(0.1 and 0.1);
\draw[fill=black,opacity=1] ({0.6*1/2},{-0.6*sqrt(3)/6-3}) ellipse(0.1 and 0.1);
\draw[fill=black,opacity=1] ({-0.6*1/2},{0.6*sqrt(3)/6-3}) ellipse(0.1 and 0.1);
\draw[fill=black,opacity=1] ({-0.6*1/2},{-0.6*sqrt(3)/6-3}) ellipse(0.1 and 0.1);
\draw[fill=black,opacity=1] ({0},{0.6*sqrt(3)/3-3}) ellipse(0.1 and 0.1);
\draw[fill=black,opacity=1] ({0},{-0.6*sqrt(3)/3-3}) ellipse(0.1 and 0.1);
\draw[decorate,decoration={brace,mirror}] ({0.1},{0.6*sqrt(3)/3-3}) -- ({0.1},{sqrt(3)/3-3});
\node[right,scale=0.5] at ({0.12},{0.8*sqrt(3)/3-3}) {$d_-$};
\end{tikzpicture}
}
\subfigure[] { 
\label{fig_unpertrubed_interface}
\begin{tikzpicture}[scale=0.6]

\foreach \x in {0,1,2} { 
        \foreach \y in {0,-1} {
            \draw[dashed] ({-1/2+\x*3/2},{sqrt(3)/2+\x*sqrt(3)/2+\y*sqrt(3)})--({1/2+\x*3/2},{sqrt(3)/2+\x*sqrt(3)/2+\y*sqrt(3)});
            \draw[dashed] ({-1/2+\x*3/2},{-sqrt(3)/2+\x*sqrt(3)/2+\y*sqrt(3)})--({1/2+\x*3/2},{-sqrt(3)/2+\x*sqrt(3)/2+\y*sqrt(3)});
            \draw[dashed] ({-1+\x*3/2},{0+\x*sqrt(3)/2+\y*sqrt(3)})--({-1/2+\x*3/2},{sqrt(3)/2+\x*sqrt(3)/2+\y*sqrt(3)});
            \draw[dashed] ({1/2+\x*3/2},{-sqrt(3)/2+\x*sqrt(3)/2+\y*sqrt(3)})--({1+\x*3/2},{0+\x*sqrt(3)/2+\y*sqrt(3)});
            \draw[dashed] ({1+\x*3/2},{0+\x*sqrt(3)/2+\y*sqrt(3)})--({1/2+\x*3/2},{sqrt(3)/2+\x*sqrt(3)/2+\y*sqrt(3)});
            \draw[dashed] ({-1/2+\x*3/2},{-sqrt(3)/2+\x*sqrt(3)/2+\y*sqrt(3)})--({-1+\x*3/2},{0+\x*sqrt(3)/2+\y*sqrt(3)});

            \node[scale=1] at ({0+\x*3/2},{0+\x*sqrt(3)/2+\y*sqrt(3)}) {$d_+$};
        }
    }

\foreach \x in {0,1} { 
        \foreach \y in {1} {
            \draw[dashed] ({-1/2+\x*3/2},{sqrt(3)/2+\x*sqrt(3)/2+\y*sqrt(3)})--({1/2+\x*3/2},{sqrt(3)/2+\x*sqrt(3)/2+\y*sqrt(3)});
            \draw[dashed] ({-1/2+\x*3/2},{-sqrt(3)/2+\x*sqrt(3)/2+\y*sqrt(3)})--({1/2+\x*3/2},{-sqrt(3)/2+\x*sqrt(3)/2+\y*sqrt(3)});
            \draw[dashed] ({-1+\x*3/2},{0+\x*sqrt(3)/2+\y*sqrt(3)})--({-1/2+\x*3/2},{sqrt(3)/2+\x*sqrt(3)/2+\y*sqrt(3)});
            \draw[dashed] ({1/2+\x*3/2},{-sqrt(3)/2+\x*sqrt(3)/2+\y*sqrt(3)})--({1+\x*3/2},{0+\x*sqrt(3)/2+\y*sqrt(3)});
            \draw[dashed] ({1+\x*3/2},{0+\x*sqrt(3)/2+\y*sqrt(3)})--({1/2+\x*3/2},{sqrt(3)/2+\x*sqrt(3)/2+\y*sqrt(3)});
            \draw[dashed] ({-1/2+\x*3/2},{-sqrt(3)/2+\x*sqrt(3)/2+\y*sqrt(3)})--({-1+\x*3/2},{0+\x*sqrt(3)/2+\y*sqrt(3)});

            \node[scale=1] at ({0+\x*3/2},{0+\x*sqrt(3)/2+\y*sqrt(3)}) {$d_+$};
        }
    }

\foreach \x in {1,2} { 
        \foreach \y in {-2} {
            \draw[dashed] ({-1/2+\x*3/2},{sqrt(3)/2+\x*sqrt(3)/2+\y*sqrt(3)})--({1/2+\x*3/2},{sqrt(3)/2+\x*sqrt(3)/2+\y*sqrt(3)});
            \draw[dashed] ({-1/2+\x*3/2},{-sqrt(3)/2+\x*sqrt(3)/2+\y*sqrt(3)})--({1/2+\x*3/2},{-sqrt(3)/2+\x*sqrt(3)/2+\y*sqrt(3)});
            \draw[dashed] ({-1+\x*3/2},{0+\x*sqrt(3)/2+\y*sqrt(3)})--({-1/2+\x*3/2},{sqrt(3)/2+\x*sqrt(3)/2+\y*sqrt(3)});
            \draw[dashed] ({1/2+\x*3/2},{-sqrt(3)/2+\x*sqrt(3)/2+\y*sqrt(3)})--({1+\x*3/2},{0+\x*sqrt(3)/2+\y*sqrt(3)});
            \draw[dashed] ({1+\x*3/2},{0+\x*sqrt(3)/2+\y*sqrt(3)})--({1/2+\x*3/2},{sqrt(3)/2+\x*sqrt(3)/2+\y*sqrt(3)});
            \draw[dashed] ({-1/2+\x*3/2},{-sqrt(3)/2+\x*sqrt(3)/2+\y*sqrt(3)})--({-1+\x*3/2},{0+\x*sqrt(3)/2+\y*sqrt(3)});

            \node[scale=1] at ({0+\x*3/2},{0+\x*sqrt(3)/2+\y*sqrt(3)}) {$d_+$};
        }
    }

\foreach \x in {3} { 
        \foreach \y in {-1,-2} {
            \draw[dashed] ({-1/2+\x*3/2},{sqrt(3)/2+\x*sqrt(3)/2+\y*sqrt(3)})--({1/2+\x*3/2},{sqrt(3)/2+\x*sqrt(3)/2+\y*sqrt(3)});
            \draw[dashed] ({-1/2+\x*3/2},{-sqrt(3)/2+\x*sqrt(3)/2+\y*sqrt(3)})--({1/2+\x*3/2},{-sqrt(3)/2+\x*sqrt(3)/2+\y*sqrt(3)});
            \draw[dashed] ({-1+\x*3/2},{0+\x*sqrt(3)/2+\y*sqrt(3)})--({-1/2+\x*3/2},{sqrt(3)/2+\x*sqrt(3)/2+\y*sqrt(3)});
            \draw[dashed] ({1/2+\x*3/2},{-sqrt(3)/2+\x*sqrt(3)/2+\y*sqrt(3)})--({1+\x*3/2},{0+\x*sqrt(3)/2+\y*sqrt(3)});
            \draw[dashed] ({1+\x*3/2},{0+\x*sqrt(3)/2+\y*sqrt(3)})--({1/2+\x*3/2},{sqrt(3)/2+\x*sqrt(3)/2+\y*sqrt(3)});
            \draw[dashed] ({-1/2+\x*3/2},{-sqrt(3)/2+\x*sqrt(3)/2+\y*sqrt(3)})--({-1+\x*3/2},{0+\x*sqrt(3)/2+\y*sqrt(3)});

            \node[scale=1] at ({0+\x*3/2},{0+\x*sqrt(3)/2+\y*sqrt(3)}) {$d_+$};
        }
    }

\foreach \x in {-1,-2} { 
        \foreach \y in {0,1} {
            \draw[dashed] ({-1/2+\x*3/2},{sqrt(3)/2+\x*sqrt(3)/2+\y*sqrt(3)})--({1/2+\x*3/2},{sqrt(3)/2+\x*sqrt(3)/2+\y*sqrt(3)});
            \draw[dashed] ({-1/2+\x*3/2},{-sqrt(3)/2+\x*sqrt(3)/2+\y*sqrt(3)})--({1/2+\x*3/2},{-sqrt(3)/2+\x*sqrt(3)/2+\y*sqrt(3)});
            \draw[dashed] ({-1+\x*3/2},{0+\x*sqrt(3)/2+\y*sqrt(3)})--({-1/2+\x*3/2},{sqrt(3)/2+\x*sqrt(3)/2+\y*sqrt(3)});
            \draw[dashed] ({1/2+\x*3/2},{-sqrt(3)/2+\x*sqrt(3)/2+\y*sqrt(3)})--({1+\x*3/2},{0+\x*sqrt(3)/2+\y*sqrt(3)});
            \draw[dashed] ({1+\x*3/2},{0+\x*sqrt(3)/2+\y*sqrt(3)})--({1/2+\x*3/2},{sqrt(3)/2+\x*sqrt(3)/2+\y*sqrt(3)});
            \draw[dashed] ({-1/2+\x*3/2},{-sqrt(3)/2+\x*sqrt(3)/2+\y*sqrt(3)})--({-1+\x*3/2},{0+\x*sqrt(3)/2+\y*sqrt(3)});

            \node[scale=1] at ({0+\x*3/2},{0+\x*sqrt(3)/2+\y*sqrt(3)}) {$d_-$};
        }
    }

\foreach \x in {-1} { 
        \foreach \y in {-1,2} {
            \draw[dashed] ({-1/2+\x*3/2},{sqrt(3)/2+\x*sqrt(3)/2+\y*sqrt(3)})--({1/2+\x*3/2},{sqrt(3)/2+\x*sqrt(3)/2+\y*sqrt(3)});
            \draw[dashed] ({-1/2+\x*3/2},{-sqrt(3)/2+\x*sqrt(3)/2+\y*sqrt(3)})--({1/2+\x*3/2},{-sqrt(3)/2+\x*sqrt(3)/2+\y*sqrt(3)});
            \draw[dashed] ({-1+\x*3/2},{0+\x*sqrt(3)/2+\y*sqrt(3)})--({-1/2+\x*3/2},{sqrt(3)/2+\x*sqrt(3)/2+\y*sqrt(3)});
            \draw[dashed] ({1/2+\x*3/2},{-sqrt(3)/2+\x*sqrt(3)/2+\y*sqrt(3)})--({1+\x*3/2},{0+\x*sqrt(3)/2+\y*sqrt(3)});
            \draw[dashed] ({1+\x*3/2},{0+\x*sqrt(3)/2+\y*sqrt(3)})--({1/2+\x*3/2},{sqrt(3)/2+\x*sqrt(3)/2+\y*sqrt(3)});
            \draw[dashed] ({-1/2+\x*3/2},{-sqrt(3)/2+\x*sqrt(3)/2+\y*sqrt(3)})--({-1+\x*3/2},{0+\x*sqrt(3)/2+\y*sqrt(3)});

            \node[scale=1] at ({0+\x*3/2},{0+\x*sqrt(3)/2+\y*sqrt(3)}) {$d_-$};
        }
    }

\foreach \x in {-2} { 
        \foreach \y in {2} {
            \draw[dashed] ({-1/2+\x*3/2},{sqrt(3)/2+\x*sqrt(3)/2+\y*sqrt(3)})--({1/2+\x*3/2},{sqrt(3)/2+\x*sqrt(3)/2+\y*sqrt(3)});
            \draw[dashed] ({-1/2+\x*3/2},{-sqrt(3)/2+\x*sqrt(3)/2+\y*sqrt(3)})--({1/2+\x*3/2},{-sqrt(3)/2+\x*sqrt(3)/2+\y*sqrt(3)});
            \draw[dashed] ({-1+\x*3/2},{0+\x*sqrt(3)/2+\y*sqrt(3)})--({-1/2+\x*3/2},{sqrt(3)/2+\x*sqrt(3)/2+\y*sqrt(3)});
            \draw[dashed] ({1/2+\x*3/2},{-sqrt(3)/2+\x*sqrt(3)/2+\y*sqrt(3)})--({1+\x*3/2},{0+\x*sqrt(3)/2+\y*sqrt(3)});
            \draw[dashed] ({1+\x*3/2},{0+\x*sqrt(3)/2+\y*sqrt(3)})--({1/2+\x*3/2},{sqrt(3)/2+\x*sqrt(3)/2+\y*sqrt(3)});
            \draw[dashed] ({-1/2+\x*3/2},{-sqrt(3)/2+\x*sqrt(3)/2+\y*sqrt(3)})--({-1+\x*3/2},{0+\x*sqrt(3)/2+\y*sqrt(3)});

            \node[scale=1] at ({0+\x*3/2},{0+\x*sqrt(3)/2+\y*sqrt(3)}) {$d_-$};
        }
    }

\foreach \x in {-3} { 
        \foreach \y in {1,2} {
            \draw[dashed] ({-1/2+\x*3/2},{sqrt(3)/2+\x*sqrt(3)/2+\y*sqrt(3)})--({1/2+\x*3/2},{sqrt(3)/2+\x*sqrt(3)/2+\y*sqrt(3)});
            \draw[dashed] ({-1/2+\x*3/2},{-sqrt(3)/2+\x*sqrt(3)/2+\y*sqrt(3)})--({1/2+\x*3/2},{-sqrt(3)/2+\x*sqrt(3)/2+\y*sqrt(3)});
            \draw[dashed] ({-1+\x*3/2},{0+\x*sqrt(3)/2+\y*sqrt(3)})--({-1/2+\x*3/2},{sqrt(3)/2+\x*sqrt(3)/2+\y*sqrt(3)});
            \draw[dashed] ({1/2+\x*3/2},{-sqrt(3)/2+\x*sqrt(3)/2+\y*sqrt(3)})--({1+\x*3/2},{0+\x*sqrt(3)/2+\y*sqrt(3)});
            \draw[dashed] ({1+\x*3/2},{0+\x*sqrt(3)/2+\y*sqrt(3)})--({1/2+\x*3/2},{sqrt(3)/2+\x*sqrt(3)/2+\y*sqrt(3)});
            \draw[dashed] ({-1/2+\x*3/2},{-sqrt(3)/2+\x*sqrt(3)/2+\y*sqrt(3)})--({-1+\x*3/2},{0+\x*sqrt(3)/2+\y*sqrt(3)});

            \node[scale=1] at ({0+\x*3/2},{0+\x*sqrt(3)/2+\y*sqrt(3)}) {$d_-$};
        }
    }

\draw[thick,red,dashed] ({-3/4},{-2.5*sqrt(3)})--({-3/4},{2.5*sqrt(3)});
\end{tikzpicture}
}
\caption{(a) $\tilde{\mathcal{T}}-$breaking structures by dislocating the sublattices (upper panel) outward with a distance $d_+>0$ or (lower panel) inward with a distance $d_->0$. In that case, the dislocated lattice is no longer translational invariant with respect to the $\mathbb{Z}\tilde{\bm{\ell}}$ translation. (b) a zigzag interface model: aside the interface, the $\tilde{\mathcal{T}}-$symmetry is broken by dislocating the sublattice oppositely, where the letter $d_+$ ($d_-$, resp.) in each cell indicates the first (second, resp.) type of deformation in Figure \ref{fig_unit_structure} is implemented. We note that this interface structure is translational invariant along $\mathbb{Z}\bm{\ell}_2$ (the zigzag interface) and reflectional symmetric about the $x-$axis.}
\label{fig_interface structure}
\end{figure}
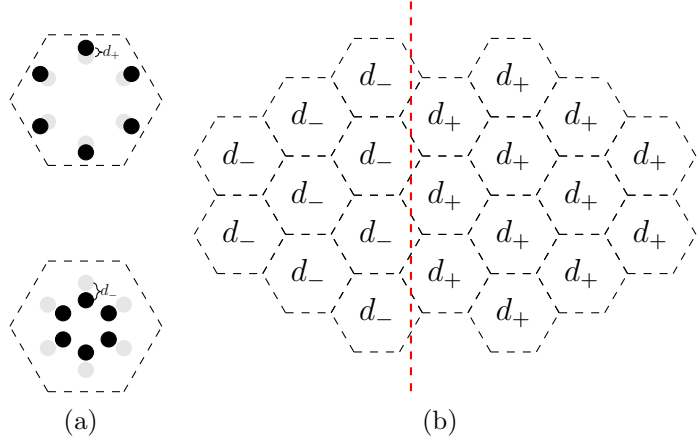

As a consequence of $\tilde{\mathcal{T}}$-symmetry breaking, the perturbed Hamiltonian $\mathcal{H}_{\pm\delta}$ does not possess a generic degeneracy of dimensional four. Thus, one can expect the double Dirac cone to vanish and a band gap to open. The following result is proved in Section \ref{sec_gap_open}:

\begin{theorem} \label{thm_gap_open}
Suppose that the conditions in Theorem \ref{thm_double_cone}, Assumption \ref{assum_symmetry_break}, and the following spectral no-fold condition hold
\begin{equation} \label{eq_no_fold}
\lambda_n(\bm{\kappa})=\lambda_*
\Longrightarrow n\in\{1,2,3,4\}\,\text{ and }\, \bm{\kappa}=\bm{0}.
\end{equation}
Assume also that the first four bands are isolated
\begin{equation} \label{eq_isolated_bands}
\max_{\bm{\kappa}\in Y^*}\lambda_{4}(\bm{\kappa})<\min_{\bm{\kappa}\in Y^*}\lambda_{5}(\bm{\kappa}).
\end{equation}
Define $\beta_{*,k}:=\big(u_{k}(\cdot;\bm{0}),\mathcal{H}_{per}(\bm{0}) u_{k}(\cdot;\bm{0}) \big)_{\mathcal{X}_{\bm{0}}}\in\mathbb{R}$ ($k=1,3$). We assume that
\begin{equation} \label{eq_gap_open_criterion}
\beta_{*}:=\beta_{*,1}=- \beta_{*,3}>0.
\end{equation}
Then, there exists $\Delta_0>0$ such that for any $c_*\in (0,1)$ and $\delta\in (0,\Delta_0)$, there is a band gap opened in $\text{Spec}(\mathcal{H}_{\pm\delta})$ in the sense that
\begin{equation} \label{eq_gap_open}
\Big(\lambda_*+ c_*\delta\cdot\beta_{*},\lambda_*-c_*\delta\cdot\beta_{*}
\Big)\bigcap \text{Spec}(\mathcal{H}_{\pm\delta})=\emptyset .
\end{equation}
\end{theorem}

In the sequel, we denote the band gap in \eqref{eq_gap_open} by $\mathcal{I}_{\delta}:=(\lambda_*+ c_*\delta\cdot\beta_{*},\lambda_*-c_*\delta\cdot\beta_{*})$.

\begin{remark} \label{rmk_isolated_bands}
The conditions \eqref{eq_no_fold} and \eqref{eq_gap_open_criterion} are necessary for the gap opening and hence the in-gap interface modes; we refer the reader to \cite{qiu2023no_gap} for results on the bifurcation of Dirac points when \eqref{eq_no_fold} fails. The assumption \eqref{eq_isolated_bands} is imposed only for technical reasons. In particular, \eqref{eq_gap_open_criterion} indicates that the perturbation $\mathcal{H}_{per}$ has \textit{nontrivial and opposite} effects on the Floquet-Bloch modes $u_{1}(\bm{n};\bm{0})$ and $u_{3}(\bm{n};\bm{0})$ (which possess different symmetries). This  is analogous to the effect of external magnetic fields on electrons with opposite spins. In one aspect, the interaction with external fields changes the energy of electrons, with the sign of energy change depending on the spin directions. This energy fluctuation is manifested by the lifting of double Dirac cone degeneracy and the appearance of a band gap (hence \eqref{eq_gap_open_criterion} is sometimes called \textit{the non-degeneracy condition}); see Section \ref{sec_gap_open} for details. Based on this analogy, the interface modes studied later in this paper demonstrate an analog of the quantum spin Hall effect \cite{WuHu15scheme}.

We also note that condition \eqref{eq_gap_open_criterion} can be generalized to the case $\beta_{*,1}\cdot\beta_{*,3}<0$, with Theorem \ref{thm_gap_open}, and more importantly, Theorems \ref{thm_existence_interface_modes} and \ref{thm_robustness_interface_modes} in the sequel still hold; nonetheless, we focus on the simple case \eqref{eq_gap_open_criterion} to minimize the technical details in this paper.
\end{remark}

\begin{remark}
In addition to the common band gap between $\mathcal{H}_{\delta}$ and $\mathcal{H}_{-\delta}$, a more important consequence of the symmetry-breaking perturbation is that it induces a band inversion between $\mathcal{H}_{\delta}$ and $\mathcal{H}_{-\delta}$, which are characterized by the switching of Bloch eigenspaces at the end points of the gap $\mathcal{I}_{\delta}$; see Remark \ref{rmk_phase_transition} and \ref{rmk_absence_band_inversion} for a detailed discussion.
\end{remark}

Finally, we are ready to introduce the interface model as the main focus of this paper. It is described by the tight-binding Hamiltonian $\mathcal{H}_{zig}$, specified as follows:
\begin{equation} \label{eq_interface_Hamiltonian}
\mathcal{H}_{zig}(\bm{n},\bm{m})=
\left\{
\begin{aligned}
&\mathcal{H}_{\delta}(\bm{n},\bm{m}),\quad n_1,m_1\geq 0, \\
&\mathcal{H}_{-\delta}(\bm{n},\bm{m}),\quad n_1,m_1< 0, \\
&(\mathcal{H}_{b}+\delta\mathcal{E})(\bm{n},\bm{m}),\quad \text{otherwise},
\end{aligned}
\right.
\end{equation}
where the additional Hamiltonian $\mathcal{E}$ brought by the interface satisfies the following condition.

\begin{assumption} \label{asmp_interface_hamiltonian}
$\mathcal{E}$ is translational invariant with respect to the zig-zag interface $\mathbb{Z}\bm{\ell}_2$, i.e.,
$$\mathcal{E}(\bm{n}+\bm{\ell}_2,\bm{m}+\bm{\ell}_2) = \mathcal{E}(\bm{n},\bm{m}) $$ for all $\bm{n},\bm{m}$. Moreover, $\mathcal{E}$ is $\mathcal{F}_x$ symmetric, i.e., $[\mathcal{E},\mathcal{F}_x]=0$.
\end{assumption}

The interface structure with which we are concerned is illustrated in Figure \ref{fig_unpertrubed_interface}. We emphasize that, under Assumption \ref{asmp_interface_hamiltonian}, the governing Hamiltonian $\mathcal{H}_{zig}$ is $\mathcal{F}_x$ (i.e., reflectional) symmetric. This property is critical for the robustness of interface modes, as we will show later.

The interface modes are generalized eigenfunctions of $\mathcal{H}_{zig}$ that are localized transversally, while extended along the interface $\mathbb{Z}\bm{\ell}_2$. To be precise, we introduce the following space imposed with a quasi-periodic boundary condition along $\mathbb{Z}\bm{\ell}_2$:
\begin{equation} \label{eq_quasi-periodic_strip_space_def}
\mathcal{X}_{\kappa_{\parallel}}:=\{u\in \ell^2_{loc}(\Lambda)\otimes \mathbb{C}^6:\, u(\bm{n}+\bm{\ell}_2)=e^{i\kappa_{\parallel}}u(\bm{n}),\, \sum_{n_1\in\mathbb{Z}}\|u(n_1\bm{e}_1)\|^2<\infty \},
\end{equation}
with $\kappa_{\parallel}\in [-\pi,\pi]$.\footnote{Caution of notation: the space $\mathcal{X}_{\bm{\kappa}}$ defined in \eqref{eq_quasi_periodic_space_unit_cell} (subscript with a vector $\bm{\kappa}$) is imposed with quasi-periodic boundary conditions in two directions; hence the unit cell is the parallelogram depicted in Figure \ref{fig_periodic_structure}. On the other hand, we only impose quasi-periodic boundary conditions along the single direction $\bm{\ell}_2$ in \eqref{eq_quasi-periodic_strip_space_def}; thus the unit structure of $\mathcal{X}_{\kappa_{\parallel}}$ (scalar in the subscript) is a strip.} When $\kappa_{\parallel}=0$ (the most frequently used case in this paper), we write $\mathcal{X}_{\sharp}:=\mathcal{X}_{\kappa_{\parallel}=0}$.

As the first main result of this paper, we prove the existence and determine the precise number of in-gap interface modes.

\begin{theorem}[Existence of interface modes] \label{thm_existence_interface_modes}
Suppose that the conditions in Theorem \ref{thm_gap_open} and Assumption \ref{asmp_interface_hamiltonian} hold true. On the other hand, we assume that the unperturbed bulk Hamiltonian has non-singular $N-$neighbor hopping when restricted on $\mathcal{X}_{\sharp}$ in the sense that
\begin{equation} \label{eq_non_singular_hopping}
\sum_{n_2}\mathcal{H}_{b}(\bm{0},N\bm{\ell}_1+n_2\bm{\ell}_2)\in \text{GL}(6,\mathbb{C}).
\end{equation}
Then, there exists $\Delta_1>0$ such that if $0<\delta<\Delta_1$, the Hamiltonian $\mathcal{H}_{zig}\Big|_{\mathcal{X}_{\sharp}}$ has exactly two in-gap eigenvalues (counted with multiplicity) $\lambda_{zig,n}\in\mathcal{I}_{\delta}$ ($n=1,2$) which are near the mid-gap energy, that is, 
\begin{equation} \label{eq_interface_eigenvalue_estimate}
\lambda_{zig,n}-\lambda_*=o(\mathcal{\delta}).
\end{equation}
The associated eigenmodes, denoted by $u_{zig,n}$, attain opposite parity with respect to the $x-$axis reflection in the sense that
\begin{equation} \label{eq_interface_mode_parity}
\mathcal{F}_x u_{zig,1}=u_{zig,1},\quad \mathcal{F}_x u_{zig,2}=-u_{zig,2}.
\end{equation}
\end{theorem} 

We note that, by combining Theorem \ref{thm_existence_interface_modes} and the analytic perturbation theory, it follows immediately that the in-gap interface modes also exist in $\mathcal{X}_{\kappa_{\parallel}}$ for $\kappa_{\parallel}$ being sufficiently close to zero:

\begin{corollary} \label{corol_interface_eigenvalue_curve}
Suppose that the conditions in Theorem \ref{thm_existence_interface_modes} hold true. Then, there exists a neighborhood of $\kappa_{\parallel}=0$ in which $\mathcal{H}_{zig}\Big|_{\mathcal{X}_{\kappa_{\parallel}}}$ has exactly two eigenvalues $\lambda_{zig,n}(\kappa_{\parallel})$ ($n=1,2$) inside $\mathcal{I}_{\delta}$.
\end{corollary}

\begin{remark} \label{rmk_absence_eigenvalue_pi}
Thanks to the spectral no-fold condition \eqref{eq_no_fold}, we know that the unperturbed Hamiltonian $\mathcal{H}_{b}$ is already ``locally gapped" near $\kappa_{\parallel}=\pi$. To be more precise, there exists a neighborhood of $\kappa_{\parallel}=\pi$ in which the spectrum of $\mathcal{H}_{b}\big|_{\mathcal{X}_{\kappa_{\parallel}}}$ is empty near $\lambda=\lambda_*$. As a consequence, for sufficiently small $\delta$, we can see by the standard perturbation theory that the interface Hamiltonian with parallel momentum $\kappa_{\parallel}=\pm \pi$, i.e., $\mathcal{H}_{zig}\big|_{\mathcal{X}_{\kappa_{\parallel}}}$, does not have eigenvalues within $\mathcal{I}_{\delta}$. In other words, the range of continuation of the interface eigenvalue stated in Corollary \ref{corol_interface_eigenvalue_curve} does not contain $\kappa_{\parallel}=\pm \pi$ if we further restrict the size of the parameter $\delta$.
\end{remark}

\begin{remark} \label{rmk_non_singular_hopping}
The assumption \eqref{eq_non_singular_hopping} is imposed for technical reasons; see the discussion following \eqref{eq_M_PV_kernel_proof_7} and \eqref{eq_M_PV_kernel_proof_8}. It is satisfied if we consider the nearest-neighborhood non-singular hopping, i.e., $N=1$ in Assumption \ref{asmp_short_range} and $\mathcal{H}_{b}(\bm{n},\bm{m}+\bm{e})\in \text{GL}(6,\mathbb{C})$ with $\bm{e}\in \{\pm \bm{\ell}_1,\pm \bm{\ell}_2\}$. For example, one can improve the toy Hamiltonian in Example \ref{examp_super_symmetric_Hamiltonian} by including hoppings between lattice sites with a distance smaller than one:
\begin{equation*}
\mathcal{H}_b(\bm{n},\bm{m})_{i,j}=\left\{
\begin{aligned}
& \frac{1}{3}\|\Xi(\bm{n},i)-\Xi(\bm{m},j)\|_{2}^{-1},\quad \|\Xi(\bm{n},i)-\Xi(\bm{m},j)\|_{2}\in [\frac{1}{3},1], \\
&0,\quad \text{otherwise}.
\end{aligned}
\right.
\end{equation*}
Then it can be checked that \eqref{eq_non_singular_hopping} is satisfied. We note that such a technical assumption is not needed for studying interface modes in a continuous system because the differential operator naturally encodes all long-range hoppings; see \cite{qiu2026waveguide_localized,qiu2024square_lattice,li2024interface_mode_honeycomb}.
\end{remark}

The proof of Theorem \ref{thm_existence_interface_modes} is based on a layer-potential framework, which is first established in previous works for continuous systems, as reviewed in the introduction, and generalized to discrete structures in this paper. We also note that, with this framework and following the lines of \cite[Section 9]{li2024interface_mode_honeycomb}, one can even obtain more information on the in-gap spectrum of $\mathcal{H}_{zig}$, such as the asymptotic expansion of in-gap eigenvalues in a neighborhood of $\kappa_{\parallel}=0$; this is left to the interested reader.

The second main result of this paper demonstrates the robustness of the interface modes, found in Theorem \ref{thm_existence_interface_modes}, against symmetry-preserving perturbations. To state the result precisely, we first introduce some notation. We denote the isolation distance from the unperturbed interface eigenvalue $\lambda_{zig,n}$ to the complement of the bulk spectral gap by
\begin{equation} \label{eq_isolation_distance_interface_eigenvalue}
d_{zig,n}:=\text{dist}(\lambda_{zig,n},\mathcal{I}_{\delta}^{c})>0 .
\end{equation}
Clearly, by estimate \eqref{eq_interface_eigenvalue_estimate}, $d_{zig,n}$ is on the same scale as the bulk spectral gap, i.e., $\mathcal{O}(\delta)$. For the remaining paragraphs of this section, we will fix the parameter $\delta$. Now, we introduce the following class of perturbation Hamiltonians (examples are provided later).
\begin{definition} \label{def_class_A_perturbation}
The Hamiltonian $\mathcal{W}$ is in class (A) if the following conditions hold:
\begin{itemize}
    \item[(i)] (reflectional symmetry) $[\mathcal{W},\mathcal{F}_{x}]=0$;
    \item[(ii)] (longitudinal localization) there exists $M_{\mathcal{W}}\in (0,\infty)$ such that
    \begin{equation}  \label{eq_longi_loc_cond}
    \sup_{n_1\in\mathbb{Z}}\sum_{\substack{n_2\in\mathbb{Z} ,\bm{m}\in \Lambda}} \|\mathcal{W}(n_1\bm{\ell}_1+n_2\bm{\ell}_2,\bm{m})\|_{\mathbb{C}^{6\times 6}}<M_{\mathcal{W}}.
    \end{equation}
\end{itemize}
\end{definition}
Applying the perturbation to the interface model \eqref{eq_interface_Hamiltonian} produces the following new Hamiltonian:
\begin{equation} \label{eq_perturbed_interface_Hamiltonian}
    \mathcal{H}_{zig,\mathcal{W}}:=\mathcal{H}_{zig}+\mathcal{W}.
\end{equation}
Then the robustness of interface modes is stated as follows.
\begin{theorem}[Symmetry-protection of interface modes] \label{thm_robustness_interface_modes}
Suppose that the conditions in Theorem \ref{thm_existence_interface_modes} hold true, and $\mathcal{W}$ is of class (A). Assume that the unperturbed interface Hamiltonian with parallel momentum $\kappa_{\parallel}=\pm \pi$ does not have eigenvalue within $\mathcal{I}_{\delta}$ (i.e., the conclusions of Remark \ref{rmk_absence_eigenvalue_pi} hold) in the sense that
\begin{equation} \label{eq_absence_eigenvalue_pi} 
\text{Spec}\big(\mathcal{H}_{zig}\big|_{\mathcal{X}_{\pm \pi}}\big)\cap \mathcal{I}_{\delta}=\emptyset .
\end{equation}
Then, there exists $c_{\mathcal{W}}\in (0,\frac{1}{2})$ such that whenever
\begin{equation} \label{eq_perturbation_potential_isolation_distance}
M_{\mathcal{W}}<c_{\mathcal{W}}\min\{d_{zig,1},d_{zig,2}\},
\end{equation}
the perturbed Hamiltonian $\mathcal{H}_{zig,\mathcal{W}}$ has two in-gap generalized eigenvalues $\lambda_{zig,\mathcal{W},n}\in\mathcal{I}_{\delta}$ ($n=1,2$) with the associated eigenmodes $u_{zig,\mathcal{W},n}\in \ell^2_{loc}(\Lambda)\otimes \mathbb{C}^6$. For each $n$, the eigenmode $u_{zig,\mathcal{W},n}$ attains the same parity as $u_{zig,n}$ in the sense of \eqref{eq_interface_mode_parity}. Moreover, they admit the following decomposition:
\begin{equation*}
u_{zig,\mathcal{W},n}=u_{zig,n}+u_{zig,\mathcal{W},n}^{(1)}\quad \text{for some }u_{zig,\mathcal{W},n}^{(1)}\in \ell^2(\Lambda)\otimes \mathbb{C}^6.
\end{equation*}
In other words, the far-field profile of the perturbed interface mode $u_{zig,\mathcal{W},n}$ is the same as the unperturbed mode $u_{zig,n}$, modulo an $\ell^2$ localized scattering part $u_{zig,\mathcal{W},n}^{(1)}$.
\end{theorem}

\begin{remark}
One can remove \eqref{eq_absence_eigenvalue_pi} by paying the price of including the isolation distance of the in-gap eigenvalue at $\kappa_{\parallel}=\pm \pi$, i.e., $\text{dist}\big(\text{Spec}\big(\mathcal{H}_{zig}\big|_{\mathcal{X}_{\pm \pi}}\big)\cap \mathcal{I}_{\delta},\mathcal{I}_{\delta}^{c}\big)$, to the right side of \eqref{eq_perturbation_potential_isolation_distance}.
\end{remark}

\begin{example} \label{examp_symmetric_perturbation}
The simplest class (A) Hamiltonian is the one that describes the compact perturbation, as shown in Figure \ref{fig_pertrubed_interface_compact}, with the kernel specified as
\begin{equation*}
\mathcal{W}(\bm{n},\bm{m})_{i,j}=\left\{
\begin{aligned}
&\delta^2,\quad \|\bm{n}-\bm{m}\|_2\leq 1,\, \|\bm{n}\|_2\leq 1 \text{ or }\|\bm{m}\|_2\leq 1, \\
&0,\quad \text{otherwise}.
\end{aligned}
\right.
\end{equation*}
Or, more generally, the one that describes a line defect as in Figure \ref{fig_pertrubed_interface_line}:
\begin{equation*}
\mathcal{W}(\bm{n},\bm{m})_{i,j}=\left\{
\begin{aligned}
&\delta^2,\quad \|\bm{n}-\bm{m}\|_2\leq 1,\, |\bm{n}\cdot\bm{\ell}_2|\leq 1 \text{ or }|\bm{m}\cdot\bm{\ell}_2|\leq 1, \\
&0,\quad \text{otherwise}.
\end{aligned}
\right.
\end{equation*}
In both cases, $\mathcal{W}$ are localized longitudinally along $\mathbb{Z}\bm{\ell}_1$, as indicated by the $n_2$ summation in \eqref{eq_longi_loc_cond} in Definition \ref{def_class_A_perturbation}.
\end{example}

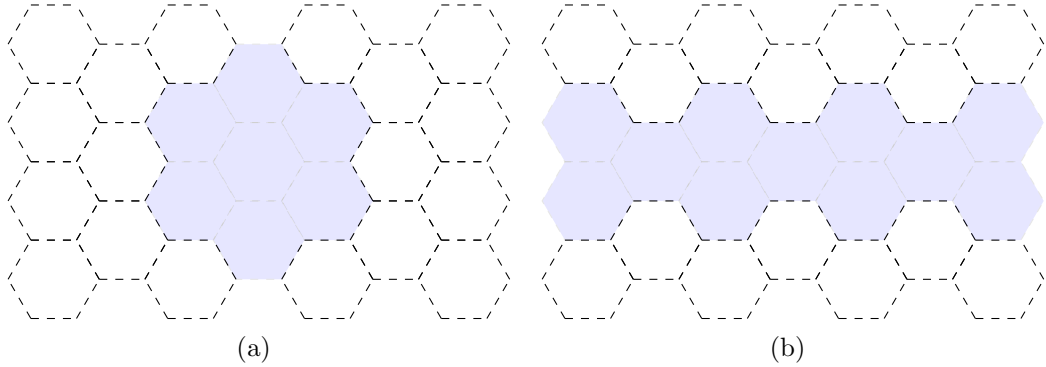
\begin{figure}
\centering
\subfigure[] { 
\label{fig_pertrubed_interface_compact}
\begin{tikzpicture}[scale=0.6]

\foreach \x in {1} { 
        \foreach \y in {1,-2} {
            \draw[dashed] ({-1/2+\x*3/2},{sqrt(3)/2+\x*sqrt(3)/2+\y*sqrt(3)})--({1/2+\x*3/2},{sqrt(3)/2+\x*sqrt(3)/2+\y*sqrt(3)});
            \draw[dashed] ({-1/2+\x*3/2},{-sqrt(3)/2+\x*sqrt(3)/2+\y*sqrt(3)})--({1/2+\x*3/2},{-sqrt(3)/2+\x*sqrt(3)/2+\y*sqrt(3)});
            \draw[dashed] ({-1+\x*3/2},{0+\x*sqrt(3)/2+\y*sqrt(3)})--({-1/2+\x*3/2},{sqrt(3)/2+\x*sqrt(3)/2+\y*sqrt(3)});
            \draw[dashed] ({1/2+\x*3/2},{-sqrt(3)/2+\x*sqrt(3)/2+\y*sqrt(3)})--({1+\x*3/2},{0+\x*sqrt(3)/2+\y*sqrt(3)});
            \draw[dashed] ({1+\x*3/2},{0+\x*sqrt(3)/2+\y*sqrt(3)})--({1/2+\x*3/2},{sqrt(3)/2+\x*sqrt(3)/2+\y*sqrt(3)});
            \draw[dashed] ({-1/2+\x*3/2},{-sqrt(3)/2+\x*sqrt(3)/2+\y*sqrt(3)})--({-1+\x*3/2},{0+\x*sqrt(3)/2+\y*sqrt(3)});

        }
    }

\foreach \x in {2} { 
        \foreach \y in {0,-1,-2} {
            \draw[dashed] ({-1/2+\x*3/2},{sqrt(3)/2+\x*sqrt(3)/2+\y*sqrt(3)})--({1/2+\x*3/2},{sqrt(3)/2+\x*sqrt(3)/2+\y*sqrt(3)});
            \draw[dashed] ({-1/2+\x*3/2},{-sqrt(3)/2+\x*sqrt(3)/2+\y*sqrt(3)})--({1/2+\x*3/2},{-sqrt(3)/2+\x*sqrt(3)/2+\y*sqrt(3)});
            \draw[dashed] ({-1+\x*3/2},{0+\x*sqrt(3)/2+\y*sqrt(3)})--({-1/2+\x*3/2},{sqrt(3)/2+\x*sqrt(3)/2+\y*sqrt(3)});
            \draw[dashed] ({1/2+\x*3/2},{-sqrt(3)/2+\x*sqrt(3)/2+\y*sqrt(3)})--({1+\x*3/2},{0+\x*sqrt(3)/2+\y*sqrt(3)});
            \draw[dashed] ({1+\x*3/2},{0+\x*sqrt(3)/2+\y*sqrt(3)})--({1/2+\x*3/2},{sqrt(3)/2+\x*sqrt(3)/2+\y*sqrt(3)});
            \draw[dashed] ({-1/2+\x*3/2},{-sqrt(3)/2+\x*sqrt(3)/2+\y*sqrt(3)})--({-1+\x*3/2},{0+\x*sqrt(3)/2+\y*sqrt(3)});

        }
    }

\foreach \x in {3} { 
        \foreach \y in {0,-1,-2,-3} {
            \draw[dashed] ({-1/2+\x*3/2},{sqrt(3)/2+\x*sqrt(3)/2+\y*sqrt(3)})--({1/2+\x*3/2},{sqrt(3)/2+\x*sqrt(3)/2+\y*sqrt(3)});
            \draw[dashed] ({-1/2+\x*3/2},{-sqrt(3)/2+\x*sqrt(3)/2+\y*sqrt(3)})--({1/2+\x*3/2},{-sqrt(3)/2+\x*sqrt(3)/2+\y*sqrt(3)});
            \draw[dashed] ({-1+\x*3/2},{0+\x*sqrt(3)/2+\y*sqrt(3)})--({-1/2+\x*3/2},{sqrt(3)/2+\x*sqrt(3)/2+\y*sqrt(3)});
            \draw[dashed] ({1/2+\x*3/2},{-sqrt(3)/2+\x*sqrt(3)/2+\y*sqrt(3)})--({1+\x*3/2},{0+\x*sqrt(3)/2+\y*sqrt(3)});
            \draw[dashed] ({1+\x*3/2},{0+\x*sqrt(3)/2+\y*sqrt(3)})--({1/2+\x*3/2},{sqrt(3)/2+\x*sqrt(3)/2+\y*sqrt(3)});
            \draw[dashed] ({-1/2+\x*3/2},{-sqrt(3)/2+\x*sqrt(3)/2+\y*sqrt(3)})--({-1+\x*3/2},{0+\x*sqrt(3)/2+\y*sqrt(3)});

        }
    }

\foreach \x in {0} { 
        \foreach \y in {0,-1,1} {
            \draw[dashed,fill=blue,opacity=0.1] ({-1/2+\x*3/2},{sqrt(3)/2+\x*sqrt(3)/2+\y*sqrt(3)})--({1/2+\x*3/2},{sqrt(3)/2+\x*sqrt(3)/2+\y*sqrt(3)})--({1+\x*3/2},{0+\x*sqrt(3)/2+\y*sqrt(3)})--({1/2+\x*3/2},{-sqrt(3)/2+\x*sqrt(3)/2+\y*sqrt(3)})--({-1/2+\x*3/2},{-sqrt(3)/2+\x*sqrt(3)/2+\y*sqrt(3)})--({-1+\x*3/2},{0+\x*sqrt(3)/2+\y*sqrt(3)})--({-1/2+\x*3/2},{sqrt(3)/2+\x*sqrt(3)/2+\y*sqrt(3)});

        }
    }

\foreach \x in {1} { 
        \foreach \y in {0,-1} {
            \draw[dashed,fill=blue,opacity=0.1] ({-1/2+\x*3/2},{sqrt(3)/2+\x*sqrt(3)/2+\y*sqrt(3)})--({1/2+\x*3/2},{sqrt(3)/2+\x*sqrt(3)/2+\y*sqrt(3)})--({1+\x*3/2},{0+\x*sqrt(3)/2+\y*sqrt(3)})--({1/2+\x*3/2},{-sqrt(3)/2+\x*sqrt(3)/2+\y*sqrt(3)})--({-1/2+\x*3/2},{-sqrt(3)/2+\x*sqrt(3)/2+\y*sqrt(3)})--({-1+\x*3/2},{0+\x*sqrt(3)/2+\y*sqrt(3)})--({-1/2+\x*3/2},{sqrt(3)/2+\x*sqrt(3)/2+\y*sqrt(3)});

        }
    }


\foreach \x in {-1} { 
        \foreach \y in {-1,2} {
            \draw[dashed] ({-1/2+\x*3/2},{sqrt(3)/2+\x*sqrt(3)/2+\y*sqrt(3)})--({1/2+\x*3/2},{sqrt(3)/2+\x*sqrt(3)/2+\y*sqrt(3)});
            \draw[dashed] ({-1/2+\x*3/2},{-sqrt(3)/2+\x*sqrt(3)/2+\y*sqrt(3)})--({1/2+\x*3/2},{-sqrt(3)/2+\x*sqrt(3)/2+\y*sqrt(3)});
            \draw[dashed] ({-1+\x*3/2},{0+\x*sqrt(3)/2+\y*sqrt(3)})--({-1/2+\x*3/2},{sqrt(3)/2+\x*sqrt(3)/2+\y*sqrt(3)});
            \draw[dashed] ({1/2+\x*3/2},{-sqrt(3)/2+\x*sqrt(3)/2+\y*sqrt(3)})--({1+\x*3/2},{0+\x*sqrt(3)/2+\y*sqrt(3)});
            \draw[dashed] ({1+\x*3/2},{0+\x*sqrt(3)/2+\y*sqrt(3)})--({1/2+\x*3/2},{sqrt(3)/2+\x*sqrt(3)/2+\y*sqrt(3)});
            \draw[dashed] ({-1/2+\x*3/2},{-sqrt(3)/2+\x*sqrt(3)/2+\y*sqrt(3)})--({-1+\x*3/2},{0+\x*sqrt(3)/2+\y*sqrt(3)});

        }
    }

\foreach \x in {-2} { 
        \foreach \y in {0,1,2} {
            \draw[dashed] ({-1/2+\x*3/2},{sqrt(3)/2+\x*sqrt(3)/2+\y*sqrt(3)})--({1/2+\x*3/2},{sqrt(3)/2+\x*sqrt(3)/2+\y*sqrt(3)});
            \draw[dashed] ({-1/2+\x*3/2},{-sqrt(3)/2+\x*sqrt(3)/2+\y*sqrt(3)})--({1/2+\x*3/2},{-sqrt(3)/2+\x*sqrt(3)/2+\y*sqrt(3)});
            \draw[dashed] ({-1+\x*3/2},{0+\x*sqrt(3)/2+\y*sqrt(3)})--({-1/2+\x*3/2},{sqrt(3)/2+\x*sqrt(3)/2+\y*sqrt(3)});
            \draw[dashed] ({1/2+\x*3/2},{-sqrt(3)/2+\x*sqrt(3)/2+\y*sqrt(3)})--({1+\x*3/2},{0+\x*sqrt(3)/2+\y*sqrt(3)});
            \draw[dashed] ({1+\x*3/2},{0+\x*sqrt(3)/2+\y*sqrt(3)})--({1/2+\x*3/2},{sqrt(3)/2+\x*sqrt(3)/2+\y*sqrt(3)});
            \draw[dashed] ({-1/2+\x*3/2},{-sqrt(3)/2+\x*sqrt(3)/2+\y*sqrt(3)})--({-1+\x*3/2},{0+\x*sqrt(3)/2+\y*sqrt(3)});

        }
    }

\foreach \x in {-3} { 
        \foreach \y in {0,1,2,3} {
            \draw[dashed] ({-1/2+\x*3/2},{sqrt(3)/2+\x*sqrt(3)/2+\y*sqrt(3)})--({1/2+\x*3/2},{sqrt(3)/2+\x*sqrt(3)/2+\y*sqrt(3)});
            \draw[dashed] ({-1/2+\x*3/2},{-sqrt(3)/2+\x*sqrt(3)/2+\y*sqrt(3)})--({1/2+\x*3/2},{-sqrt(3)/2+\x*sqrt(3)/2+\y*sqrt(3)});
            \draw[dashed] ({-1+\x*3/2},{0+\x*sqrt(3)/2+\y*sqrt(3)})--({-1/2+\x*3/2},{sqrt(3)/2+\x*sqrt(3)/2+\y*sqrt(3)});
            \draw[dashed] ({1/2+\x*3/2},{-sqrt(3)/2+\x*sqrt(3)/2+\y*sqrt(3)})--({1+\x*3/2},{0+\x*sqrt(3)/2+\y*sqrt(3)});
            \draw[dashed] ({1+\x*3/2},{0+\x*sqrt(3)/2+\y*sqrt(3)})--({1/2+\x*3/2},{sqrt(3)/2+\x*sqrt(3)/2+\y*sqrt(3)});
            \draw[dashed] ({-1/2+\x*3/2},{-sqrt(3)/2+\x*sqrt(3)/2+\y*sqrt(3)})--({-1+\x*3/2},{0+\x*sqrt(3)/2+\y*sqrt(3)});

        }
    }

\foreach \x in {-1} { 
        \foreach \y in {0,1} {
            \draw[dashed,fill=blue,opacity=0.1] ({-1/2+\x*3/2},{sqrt(3)/2+\x*sqrt(3)/2+\y*sqrt(3)})--({1/2+\x*3/2},{sqrt(3)/2+\x*sqrt(3)/2+\y*sqrt(3)})--({1+\x*3/2},{0+\x*sqrt(3)/2+\y*sqrt(3)})--({1/2+\x*3/2},{-sqrt(3)/2+\x*sqrt(3)/2+\y*sqrt(3)})--({-1/2+\x*3/2},{-sqrt(3)/2+\x*sqrt(3)/2+\y*sqrt(3)})--({-1+\x*3/2},{0+\x*sqrt(3)/2+\y*sqrt(3)})--({-1/2+\x*3/2},{sqrt(3)/2+\x*sqrt(3)/2+\y*sqrt(3)});

        }
    }
\end{tikzpicture}
}
\subfigure[] { 
\label{fig_pertrubed_interface_line}
\begin{tikzpicture}[scale=0.6]

\foreach \x in {0} { 
        \foreach \y in {1,-1} {
            \draw[dashed] ({-1/2+\x*3/2},{sqrt(3)/2+\x*sqrt(3)/2+\y*sqrt(3)})--({1/2+\x*3/2},{sqrt(3)/2+\x*sqrt(3)/2+\y*sqrt(3)});
            \draw[dashed] ({-1/2+\x*3/2},{-sqrt(3)/2+\x*sqrt(3)/2+\y*sqrt(3)})--({1/2+\x*3/2},{-sqrt(3)/2+\x*sqrt(3)/2+\y*sqrt(3)});
            \draw[dashed] ({-1+\x*3/2},{0+\x*sqrt(3)/2+\y*sqrt(3)})--({-1/2+\x*3/2},{sqrt(3)/2+\x*sqrt(3)/2+\y*sqrt(3)});
            \draw[dashed] ({1/2+\x*3/2},{-sqrt(3)/2+\x*sqrt(3)/2+\y*sqrt(3)})--({1+\x*3/2},{0+\x*sqrt(3)/2+\y*sqrt(3)});
            \draw[dashed] ({1+\x*3/2},{0+\x*sqrt(3)/2+\y*sqrt(3)})--({1/2+\x*3/2},{sqrt(3)/2+\x*sqrt(3)/2+\y*sqrt(3)});
            \draw[dashed] ({-1/2+\x*3/2},{-sqrt(3)/2+\x*sqrt(3)/2+\y*sqrt(3)})--({-1+\x*3/2},{0+\x*sqrt(3)/2+\y*sqrt(3)});

        }
    }
\foreach \x in {0} { 
        \foreach \y in {0} {
            \draw[dashed,fill=blue,opacity=0.1] ({-1/2+\x*3/2},{sqrt(3)/2+\x*sqrt(3)/2+\y*sqrt(3)})--({1/2+\x*3/2},{sqrt(3)/2+\x*sqrt(3)/2+\y*sqrt(3)})--({1+\x*3/2},{0+\x*sqrt(3)/2+\y*sqrt(3)})--({1/2+\x*3/2},{-sqrt(3)/2+\x*sqrt(3)/2+\y*sqrt(3)})--({-1/2+\x*3/2},{-sqrt(3)/2+\x*sqrt(3)/2+\y*sqrt(3)})--({-1+\x*3/2},{0+\x*sqrt(3)/2+\y*sqrt(3)})--({-1/2+\x*3/2},{sqrt(3)/2+\x*sqrt(3)/2+\y*sqrt(3)});

        }
    }

\foreach \x in {1} { 
        \foreach \y in {1,-2} {
            \draw[dashed] ({-1/2+\x*3/2},{sqrt(3)/2+\x*sqrt(3)/2+\y*sqrt(3)})--({1/2+\x*3/2},{sqrt(3)/2+\x*sqrt(3)/2+\y*sqrt(3)});
            \draw[dashed] ({-1/2+\x*3/2},{-sqrt(3)/2+\x*sqrt(3)/2+\y*sqrt(3)})--({1/2+\x*3/2},{-sqrt(3)/2+\x*sqrt(3)/2+\y*sqrt(3)});
            \draw[dashed] ({-1+\x*3/2},{0+\x*sqrt(3)/2+\y*sqrt(3)})--({-1/2+\x*3/2},{sqrt(3)/2+\x*sqrt(3)/2+\y*sqrt(3)});
            \draw[dashed] ({1/2+\x*3/2},{-sqrt(3)/2+\x*sqrt(3)/2+\y*sqrt(3)})--({1+\x*3/2},{0+\x*sqrt(3)/2+\y*sqrt(3)});
            \draw[dashed] ({1+\x*3/2},{0+\x*sqrt(3)/2+\y*sqrt(3)})--({1/2+\x*3/2},{sqrt(3)/2+\x*sqrt(3)/2+\y*sqrt(3)});
            \draw[dashed] ({-1/2+\x*3/2},{-sqrt(3)/2+\x*sqrt(3)/2+\y*sqrt(3)})--({-1+\x*3/2},{0+\x*sqrt(3)/2+\y*sqrt(3)});

        }
    }
\foreach \x in {1} { 
        \foreach \y in {0,-1} {
            \draw[dashed,fill=blue,opacity=0.1] ({-1/2+\x*3/2},{sqrt(3)/2+\x*sqrt(3)/2+\y*sqrt(3)})--({1/2+\x*3/2},{sqrt(3)/2+\x*sqrt(3)/2+\y*sqrt(3)})--({1+\x*3/2},{0+\x*sqrt(3)/2+\y*sqrt(3)})--({1/2+\x*3/2},{-sqrt(3)/2+\x*sqrt(3)/2+\y*sqrt(3)})--({-1/2+\x*3/2},{-sqrt(3)/2+\x*sqrt(3)/2+\y*sqrt(3)})--({-1+\x*3/2},{0+\x*sqrt(3)/2+\y*sqrt(3)})--({-1/2+\x*3/2},{sqrt(3)/2+\x*sqrt(3)/2+\y*sqrt(3)});

        }
    }

\foreach \x in {2} { 
        \foreach \y in {0,-2} {
            \draw[dashed] ({-1/2+\x*3/2},{sqrt(3)/2+\x*sqrt(3)/2+\y*sqrt(3)})--({1/2+\x*3/2},{sqrt(3)/2+\x*sqrt(3)/2+\y*sqrt(3)});
            \draw[dashed] ({-1/2+\x*3/2},{-sqrt(3)/2+\x*sqrt(3)/2+\y*sqrt(3)})--({1/2+\x*3/2},{-sqrt(3)/2+\x*sqrt(3)/2+\y*sqrt(3)});
            \draw[dashed] ({-1+\x*3/2},{0+\x*sqrt(3)/2+\y*sqrt(3)})--({-1/2+\x*3/2},{sqrt(3)/2+\x*sqrt(3)/2+\y*sqrt(3)});
            \draw[dashed] ({1/2+\x*3/2},{-sqrt(3)/2+\x*sqrt(3)/2+\y*sqrt(3)})--({1+\x*3/2},{0+\x*sqrt(3)/2+\y*sqrt(3)});
            \draw[dashed] ({1+\x*3/2},{0+\x*sqrt(3)/2+\y*sqrt(3)})--({1/2+\x*3/2},{sqrt(3)/2+\x*sqrt(3)/2+\y*sqrt(3)});
            \draw[dashed] ({-1/2+\x*3/2},{-sqrt(3)/2+\x*sqrt(3)/2+\y*sqrt(3)})--({-1+\x*3/2},{0+\x*sqrt(3)/2+\y*sqrt(3)});

        }
    }
\foreach \x in {2} { 
        \foreach \y in {-1} {
            \draw[dashed,fill=blue,opacity=0.1] ({-1/2+\x*3/2},{sqrt(3)/2+\x*sqrt(3)/2+\y*sqrt(3)})--({1/2+\x*3/2},{sqrt(3)/2+\x*sqrt(3)/2+\y*sqrt(3)})--({1+\x*3/2},{0+\x*sqrt(3)/2+\y*sqrt(3)})--({1/2+\x*3/2},{-sqrt(3)/2+\x*sqrt(3)/2+\y*sqrt(3)})--({-1/2+\x*3/2},{-sqrt(3)/2+\x*sqrt(3)/2+\y*sqrt(3)})--({-1+\x*3/2},{0+\x*sqrt(3)/2+\y*sqrt(3)})--({-1/2+\x*3/2},{sqrt(3)/2+\x*sqrt(3)/2+\y*sqrt(3)});

        }
    }

\foreach \x in {3} { 
        \foreach \y in {0,-3} {
            \draw[dashed] ({-1/2+\x*3/2},{sqrt(3)/2+\x*sqrt(3)/2+\y*sqrt(3)})--({1/2+\x*3/2},{sqrt(3)/2+\x*sqrt(3)/2+\y*sqrt(3)});
            \draw[dashed] ({-1/2+\x*3/2},{-sqrt(3)/2+\x*sqrt(3)/2+\y*sqrt(3)})--({1/2+\x*3/2},{-sqrt(3)/2+\x*sqrt(3)/2+\y*sqrt(3)});
            \draw[dashed] ({-1+\x*3/2},{0+\x*sqrt(3)/2+\y*sqrt(3)})--({-1/2+\x*3/2},{sqrt(3)/2+\x*sqrt(3)/2+\y*sqrt(3)});
            \draw[dashed] ({1/2+\x*3/2},{-sqrt(3)/2+\x*sqrt(3)/2+\y*sqrt(3)})--({1+\x*3/2},{0+\x*sqrt(3)/2+\y*sqrt(3)});
            \draw[dashed] ({1+\x*3/2},{0+\x*sqrt(3)/2+\y*sqrt(3)})--({1/2+\x*3/2},{sqrt(3)/2+\x*sqrt(3)/2+\y*sqrt(3)});
            \draw[dashed] ({-1/2+\x*3/2},{-sqrt(3)/2+\x*sqrt(3)/2+\y*sqrt(3)})--({-1+\x*3/2},{0+\x*sqrt(3)/2+\y*sqrt(3)});

        }
    }
\foreach \x in {3} { 
        \foreach \y in {-1,-2} {
            \draw[dashed,fill=blue,opacity=0.1] ({-1/2+\x*3/2},{sqrt(3)/2+\x*sqrt(3)/2+\y*sqrt(3)})--({1/2+\x*3/2},{sqrt(3)/2+\x*sqrt(3)/2+\y*sqrt(3)})--({1+\x*3/2},{0+\x*sqrt(3)/2+\y*sqrt(3)})--({1/2+\x*3/2},{-sqrt(3)/2+\x*sqrt(3)/2+\y*sqrt(3)})--({-1/2+\x*3/2},{-sqrt(3)/2+\x*sqrt(3)/2+\y*sqrt(3)})--({-1+\x*3/2},{0+\x*sqrt(3)/2+\y*sqrt(3)})--({-1/2+\x*3/2},{sqrt(3)/2+\x*sqrt(3)/2+\y*sqrt(3)});

        }
    }


\foreach \x in {-1} { 
        \foreach \y in {-1,2} {
            \draw[dashed] ({-1/2+\x*3/2},{sqrt(3)/2+\x*sqrt(3)/2+\y*sqrt(3)})--({1/2+\x*3/2},{sqrt(3)/2+\x*sqrt(3)/2+\y*sqrt(3)});
            \draw[dashed] ({-1/2+\x*3/2},{-sqrt(3)/2+\x*sqrt(3)/2+\y*sqrt(3)})--({1/2+\x*3/2},{-sqrt(3)/2+\x*sqrt(3)/2+\y*sqrt(3)});
            \draw[dashed] ({-1+\x*3/2},{0+\x*sqrt(3)/2+\y*sqrt(3)})--({-1/2+\x*3/2},{sqrt(3)/2+\x*sqrt(3)/2+\y*sqrt(3)});
            \draw[dashed] ({1/2+\x*3/2},{-sqrt(3)/2+\x*sqrt(3)/2+\y*sqrt(3)})--({1+\x*3/2},{0+\x*sqrt(3)/2+\y*sqrt(3)});
            \draw[dashed] ({1+\x*3/2},{0+\x*sqrt(3)/2+\y*sqrt(3)})--({1/2+\x*3/2},{sqrt(3)/2+\x*sqrt(3)/2+\y*sqrt(3)});
            \draw[dashed] ({-1/2+\x*3/2},{-sqrt(3)/2+\x*sqrt(3)/2+\y*sqrt(3)})--({-1+\x*3/2},{0+\x*sqrt(3)/2+\y*sqrt(3)});

        }
    }
\foreach \x in {-1} { 
        \foreach \y in {0,1} {
            \draw[dashed,fill=blue,opacity=0.1] ({-1/2+\x*3/2},{sqrt(3)/2+\x*sqrt(3)/2+\y*sqrt(3)})--({1/2+\x*3/2},{sqrt(3)/2+\x*sqrt(3)/2+\y*sqrt(3)})--({1+\x*3/2},{0+\x*sqrt(3)/2+\y*sqrt(3)})--({1/2+\x*3/2},{-sqrt(3)/2+\x*sqrt(3)/2+\y*sqrt(3)})--({-1/2+\x*3/2},{-sqrt(3)/2+\x*sqrt(3)/2+\y*sqrt(3)})--({-1+\x*3/2},{0+\x*sqrt(3)/2+\y*sqrt(3)})--({-1/2+\x*3/2},{sqrt(3)/2+\x*sqrt(3)/2+\y*sqrt(3)});

        }
    }

\foreach \x in {-2} { 
        \foreach \y in {0,2} {
            \draw[dashed] ({-1/2+\x*3/2},{sqrt(3)/2+\x*sqrt(3)/2+\y*sqrt(3)})--({1/2+\x*3/2},{sqrt(3)/2+\x*sqrt(3)/2+\y*sqrt(3)});
            \draw[dashed] ({-1/2+\x*3/2},{-sqrt(3)/2+\x*sqrt(3)/2+\y*sqrt(3)})--({1/2+\x*3/2},{-sqrt(3)/2+\x*sqrt(3)/2+\y*sqrt(3)});
            \draw[dashed] ({-1+\x*3/2},{0+\x*sqrt(3)/2+\y*sqrt(3)})--({-1/2+\x*3/2},{sqrt(3)/2+\x*sqrt(3)/2+\y*sqrt(3)});
            \draw[dashed] ({1/2+\x*3/2},{-sqrt(3)/2+\x*sqrt(3)/2+\y*sqrt(3)})--({1+\x*3/2},{0+\x*sqrt(3)/2+\y*sqrt(3)});
            \draw[dashed] ({1+\x*3/2},{0+\x*sqrt(3)/2+\y*sqrt(3)})--({1/2+\x*3/2},{sqrt(3)/2+\x*sqrt(3)/2+\y*sqrt(3)});
            \draw[dashed] ({-1/2+\x*3/2},{-sqrt(3)/2+\x*sqrt(3)/2+\y*sqrt(3)})--({-1+\x*3/2},{0+\x*sqrt(3)/2+\y*sqrt(3)});

        }
    }
\foreach \x in {-2} { 
        \foreach \y in {1} {
            \draw[dashed,fill=blue,opacity=0.1] ({-1/2+\x*3/2},{sqrt(3)/2+\x*sqrt(3)/2+\y*sqrt(3)})--({1/2+\x*3/2},{sqrt(3)/2+\x*sqrt(3)/2+\y*sqrt(3)})--({1+\x*3/2},{0+\x*sqrt(3)/2+\y*sqrt(3)})--({1/2+\x*3/2},{-sqrt(3)/2+\x*sqrt(3)/2+\y*sqrt(3)})--({-1/2+\x*3/2},{-sqrt(3)/2+\x*sqrt(3)/2+\y*sqrt(3)})--({-1+\x*3/2},{0+\x*sqrt(3)/2+\y*sqrt(3)})--({-1/2+\x*3/2},{sqrt(3)/2+\x*sqrt(3)/2+\y*sqrt(3)});

        }
    }

\foreach \x in {-3} { 
        \foreach \y in {0,3} {
            \draw[dashed] ({-1/2+\x*3/2},{sqrt(3)/2+\x*sqrt(3)/2+\y*sqrt(3)})--({1/2+\x*3/2},{sqrt(3)/2+\x*sqrt(3)/2+\y*sqrt(3)});
            \draw[dashed] ({-1/2+\x*3/2},{-sqrt(3)/2+\x*sqrt(3)/2+\y*sqrt(3)})--({1/2+\x*3/2},{-sqrt(3)/2+\x*sqrt(3)/2+\y*sqrt(3)});
            \draw[dashed] ({-1+\x*3/2},{0+\x*sqrt(3)/2+\y*sqrt(3)})--({-1/2+\x*3/2},{sqrt(3)/2+\x*sqrt(3)/2+\y*sqrt(3)});
            \draw[dashed] ({1/2+\x*3/2},{-sqrt(3)/2+\x*sqrt(3)/2+\y*sqrt(3)})--({1+\x*3/2},{0+\x*sqrt(3)/2+\y*sqrt(3)});
            \draw[dashed] ({1+\x*3/2},{0+\x*sqrt(3)/2+\y*sqrt(3)})--({1/2+\x*3/2},{sqrt(3)/2+\x*sqrt(3)/2+\y*sqrt(3)});
            \draw[dashed] ({-1/2+\x*3/2},{-sqrt(3)/2+\x*sqrt(3)/2+\y*sqrt(3)})--({-1+\x*3/2},{0+\x*sqrt(3)/2+\y*sqrt(3)});

        }
    }
\foreach \x in {-3} { 
        \foreach \y in {1,2} {
            \draw[dashed,fill=blue,opacity=0.1] ({-1/2+\x*3/2},{sqrt(3)/2+\x*sqrt(3)/2+\y*sqrt(3)})--({1/2+\x*3/2},{sqrt(3)/2+\x*sqrt(3)/2+\y*sqrt(3)})--({1+\x*3/2},{0+\x*sqrt(3)/2+\y*sqrt(3)})--({1/2+\x*3/2},{-sqrt(3)/2+\x*sqrt(3)/2+\y*sqrt(3)})--({-1/2+\x*3/2},{-sqrt(3)/2+\x*sqrt(3)/2+\y*sqrt(3)})--({-1+\x*3/2},{0+\x*sqrt(3)/2+\y*sqrt(3)})--({-1/2+\x*3/2},{sqrt(3)/2+\x*sqrt(3)/2+\y*sqrt(3)});

        }
    }
\end{tikzpicture}
}
\caption{Two examples of perturbed interface structures. (a) a compact defect, (b) a line defect. Both perturbations are $\mathcal{F}_x$ symmetric and localized in the longitudinal direction.}
\label{fig_interface structure_perturbed}
\end{figure}

Based on the knowledge from topological insulator theory, especially the bulk-edge correspondence, it is well-known that if the bulk Hamiltonians $\mathcal{H}_{\pm\delta}$ possess distinct topological characters, then the interface modes are topologically protected in the sense that they are robust against a wide range of perturbations. However, as we emphasized in the introduction, a nontrivial topological character is not easy to achieve. Importantly, \textit{the construction of interface modes in Theorem \ref{thm_existence_interface_modes} does not assume the existence of a nontrivial topological character for the bulk Hamiltonian}. In such cases, it is expected that the interface modes resulting from the band inversion mechanism are only robust against a limited class of perturbations \cite{fu2011crystalline}. This idea is proved rigorously by Theorem \ref{thm_robustness_interface_modes}, where we demonstrate that the interface modes arising in our system are \textit{symmetry-protected}, in the sense that they are robust if the perturbation is reflectional symmetric. To the best of our knowledge, Theorem \ref{thm_robustness_interface_modes} is the first rigorous result on the robustness of interface modes in the absence of topological protection.

Remarkably, as indicated in Example \ref{examp_symmetric_perturbation}, Theorem \ref{thm_robustness_interface_modes} shows that the interface modes proved in Theorem \ref{thm_existence_interface_modes} have the ability to pass a line-barrier and are free of scattering in the far field; we hope that this theoretical result will be manifested in experiments.

The main difficulty in proving Theorem \ref{thm_robustness_interface_modes} lies in the fact that the interface eigenvalue generally belongs to the continuous spectrum of $\mathcal{H}_{zig}$, following from the periodicity along the interface \cite{Sobolev02absolute}, which prevents a direct implementation of eigenvalue perturbation theory. This problem is handled by a periodic approximation argument, which has been extensively applied to study solitary waves in nonlinear systems \cite{pankov2010soliton_discrete,pelinovsky2011localization,Pankov2005soliton_continuous,qiu2025nonlinear}. The idea is that we first impose the periodic boundary condition on a strip that transects the interface and has a width equal to $L>0$, and then, due to the periodic boundary condition, the perturbation theory is available to determine the interface modes inside such strips. Finally, by a uniform estimate on the norm of the interface modes, we can prove the convergence of strip-interface modes as $L\to\infty$ and hence conclude the proof of Theorem \ref{thm_robustness_interface_modes}. The details are presented in Section \ref{sec_robustness}.

\begin{remark}
Although the Hamiltonians in this paper are assumed to be finite-range, we remark that the main results, such as Theorem \ref{thm_existence_interface_modes}, can be generalized to include Hamiltonians with sufficiently fast (e.g., exponential) off-diagonal decay. This can be achieved by 1) truncating the Hamiltonians to be $N-$range, 2) utilizes Theorem \ref{thm_existence_interface_modes} to show the existence of interface modes for each $N$, and finally 3) proves the convergence of interface modes as $N\to \infty$. The key point in this approximation argument is an $N-$uniform estimate of gap size, i.e., Theorem \ref{thm_gap_open} and thus the error induced by truncating Hamiltonians, which is indeed available by estimating the constant \eqref{eq_gap_open_criterion}. We leave this extension to the interested reader.
\end{remark}

\section{Double Dirac Cone at the $\Gamma$ Point}
\label{sec_representation_analysis}

In this section, we prove Proposition \ref{prop_local_flat_double_eigenvalue} and Theorem \ref{thm_double_cone} by exploiting the symmetry properties of the bulk Hamiltonian $\mathcal{H}_b$.

\subsection{Preliminaries}
\label{sec_prelim_representation}

We briefly review some aspects of the representation theory and perturbation theory used in the proof, mainly following the lines of \cite{berkolaiko2018symmetry}.

Consider a general Hamiltonian $\mathcal{H}$ acting on $\mathcal{X}$. Let $\mathcal{G}$ be a finite group of unitary operators that act on $\mathcal{X}$ and commute with $\mathcal{H}$. It is well known that there is an isotypic decomposition of $\mathcal{X}$ into a finite orthogonal sum of subspaces, each carrying copies of an irrep $\rho$ of $\mathcal{G}$. To be specific, 
\begin{equation*}
\mathcal{X}=\oplus_{\rho}\mathcal{X}(\rho).
\end{equation*}
For each $u\in \mathcal{X}(\rho)$, the dimension of orbit $\mathcal{G}u$ is equal to the dimension of $\rho$. Since $\mathcal{H}$ commutes with $\mathcal{G}$, $\mathcal{H}$ is invariant on each isotypic component $\mathcal{X}(\rho)$. Consequently, the multiplicities of the eigenvalues of $\mathcal{H}\big|_{\mathcal{X}(\rho)}$ are divided by the dimension of $\rho$. We note that the above arguments also work on the quasi-periodic function space $\mathcal{X}_{\bm{\kappa}}$, whenever i) the group $\mathcal{G}$ and the Hamiltonian $\mathcal{H}$ are both invariant on $\mathcal{X}_{\bm{\kappa}}$ and ii) $[\mathcal{H},\mathcal{G}]=0$.

Let us now fix the discussion to the bulk Hamiltonian $\mathcal{H}_b$. Suppose that the group $\mathcal{G}$ satisfies conditions i) and ii). For each $g\in\mathcal{G}$, its counterpart $\hat{g}$ in the momentum space is defined as the unitary operator on $\mathbb{R}^2$ that satisfies
\begin{equation} \label{eq_general_momentum_symmetry}
    g\mathcal{H}_b(\bm{\kappa})g^{-1}=\mathcal{H}_b(\hat{g}\bm{\kappa}).
\end{equation}
For example, $\hat{\mathcal{R}}_6=(R^{ext}_6)^{-1}$ by \eqref{eq_symm_momen_sapce_2}. Let $[\hat{g}]=([\hat{g}]_{ij})$ be the matrix representation of $\hat{g}$ under the dual basis $\{\bm{\ell}_1^*,\bm{\ell}_2^*\}$. For $\bm{\kappa}_*$ being a fixed point of $\hat{g}$, we set $\bm{\kappa}=\bm{\kappa}_*+\kappa_1\cdot \bm{\ell}_1^*+\kappa_2\cdot \bm{\ell}_2^*$ and expand \eqref{eq_general_momentum_symmetry} to the first order term,
\begin{equation*}
g\big(\kappa_{1}\cdot \frac{\partial \mathcal{H}_{b}(\bm{\kappa}_*)}{\partial \kappa_1}+\kappa_{2}\cdot \frac{\partial \mathcal{H}_{b}(\bm{\kappa}_*)}{\partial \kappa_2}\big)g^{-1}
=\kappa_1\cdot \sum_{i=1,2}[\hat{g}]_{i1}\frac{\partial \mathcal{H}_{b}(\bm{\kappa}_*)}{\partial \kappa_i}+\kappa_2\cdot \sum_{i=1,2}[\hat{g}]_{i2}\frac{\partial \mathcal{H}_{b}(\bm{\kappa}_*)}{\partial \kappa_i}.
\end{equation*}
Separation of variables leads to
\begin{equation} \label{eq_momentum_symmetry_derivative}
g\cdot \frac{\partial \mathcal{H}_{b}(\bm{\kappa}_*)}{\partial \kappa_j}\cdot g^{-1}
=\sum_{i=1,2}[\hat{g}]_{ij}\frac{\partial \mathcal{H}_{b}(\bm{\kappa}_*)}{\partial \kappa_i},\quad j=1,2.
\end{equation}
Let $u_{1},\cdots,u_{n}$ be a basis of the representation $\rho(g)$ on $\mathcal{X}_{\bm{\kappa}}(\rho)$, i.e.,
\begin{equation} \label{eq_partner_definition}
g(u_{1}\,\cdots\,u_{n})^{\top}=\rho(g)(u_{1}\,\cdots\,u_{n})^{\top}.
\end{equation}
Then, by conjugating both sides of \eqref{eq_momentum_symmetry_derivative} with $(u_{1}\,\cdots\,u_{n})$, we obtain
\begin{equation} \label{eq_peturb_matrix_symmetry}
\rho(g)^{-1}H_j\rho(g)=\sum_{i=1,2}[\hat{g}]_{ij}H_i,\quad j=1,2,
\end{equation}
where the matrix $H_j$ is defined as
\begin{equation} \label{eq_sec2_6}
H_j:= \Big[\big( u_k,\frac{\partial \mathcal{H}_{b}(\bm{\kappa}_*)}{\partial \kappa_j}u_p\big)_{\mathcal{X}_{\bm{\kappa}_*}}\Big]_{1\leq k,p\leq n} .
\end{equation}
Note that equation \eqref{eq_peturb_matrix_symmetry} holds for all $g\in\mathcal{G}$. In the following section, we will see how the symmetry property \eqref{eq_peturb_matrix_symmetry} simplifies the expression of matrices $H_i$, which is essential to apply the perturbation theory. In fact, the first-order expansion of the dispersion surface near $(\bm{\kappa}_*,\lambda_*)$ is calculated by solving the following equation (see, e.g. \cite[Section 5]{qiu2024square_lattice}):
\begin{equation} \label{eq_first_order_spectrum_det}
\det(\kappa_1\cdot H_1+ \kappa_2\cdot H_2-\lambda\cdot\mathbb{I}_{2\times 2})=0.
\end{equation}
Thus, our methodology for studying the band structure of $\mathcal{H}_b$ near a high symmetry point $\bm{\kappa}_*$ is summarized as follows:
\begin{enumerate}
    \item Find the symmetry group $\mathcal{G}$ that commutes with $\mathcal{H}_b$;
    \item Fix an eigenvalue $\lambda_*\in\text{Spec}(\mathcal{H}_{b}(\bm{\kappa}_*))$. Then, find \textit{a unique irrep $\rho$ of $\mathcal{G}$ that fits the symmetry of eigenfunctions of 
    $\mathcal{H}_{b}(\bm{\kappa}_*)$ with the eigenvalue $\lambda_*$} in the sense of \eqref{eq_partner_definition};
    \item Apply \eqref{eq_peturb_matrix_symmetry} for all $g\in\mathcal{G}$ to solve the Hermitian matrices $H_i$;
    \item Substitute $H_i$ obtained in the third step into equation \eqref{eq_first_order_spectrum_det}. Then, the solution gives the desired first-order expansion of the dispersion surface.
\end{enumerate}

\subsection{Proof of Proposition \ref{prop_local_flat_double_eigenvalue}}
\label{sec_local_flat}
Now we apply the framework outlined in Section \ref{sec_prelim_representation} to prove Proposition \ref{prop_local_flat_double_eigenvalue}. For the high-symmetry point $\bm{\kappa}_*=\bm{0}$ and the choice of the $C_{6v}$ symmetry group $\mathcal{G}=\mathcal{S}$, it is direct to check that all the assumptions listed in Section \ref{sec_prelim_representation} are satisfied.

To proceed, we note that $\mathcal{S}$ has two distinct two-dimensional irreps on $\mathcal{X}_{\bm{0}}$:
\begin{equation*}
\rho_1:\quad \mathcal{R}_6\mapsto
\begin{pmatrix}
\tau & 0 \\ 0 & \overline{\tau}
\end{pmatrix},\quad
\mathcal{F}_x\mapsto
\begin{pmatrix}
0 & 1 \\ 1 & 0
\end{pmatrix},
\end{equation*}
and
\begin{equation*}
\rho_2:\quad \mathcal{R}_6\mapsto
\begin{pmatrix}
\tau^2 & 0 \\ 0 & \overline{\tau}^2
\end{pmatrix},\quad
\mathcal{F}_x\mapsto
\begin{pmatrix}
0 & 1 \\ 1 & 0
\end{pmatrix},
\end{equation*}
where $\tau=e^{i\frac{\pi}{3}}$. Next, suppose that $\lambda_*\in \text{Spec}(\mathcal{H}_{b}(\bm{0}))$ is an eigenvalue with multiplicity two, with associated eigenfunctions denoted by $u_1(\cdot;\bm{0})$ and $u_2(\cdot;\bm{0})$. Under the assumptions of Proposition \ref{prop_local_flat_double_eigenvalue}, these two eigenfunctions form a basis of either $\rho_1$ or $\rho_2$. We will show that in the former case, the dispersion surfaces must be locally flat near $\lambda_*$; the proof of the latter is similar.

To apply \eqref{eq_peturb_matrix_symmetry}, note that the matrices of $\hat{\mathcal{R}}_6$ and $\hat{\mathcal{F}}_x$ under the dual basis $\{\bm{\ell}_1^*,\bm{\ell}_2^*\}$ are given by
\begin{equation*}
[\hat{\mathcal{R}}_6]=
\begin{pmatrix}
0 & 1 \\
-1 & 1 
\end{pmatrix},\quad
[\hat{\mathcal{F}}_x]=
\begin{pmatrix}
1 & -1 \\
0 & -1
\end{pmatrix}.
\end{equation*}
With $\rho=\rho_1$, \eqref{eq_peturb_matrix_symmetry} becomes
\begin{equation} \label{eq_sec2_1}
\begin{pmatrix}
\overline{\tau} & 0 \\ 0 & \tau
\end{pmatrix}
H_1
\begin{pmatrix}
\tau & 0 \\ 0 & \overline{\tau}
\end{pmatrix}
=-H_2,\quad
\begin{pmatrix}
\overline{\tau} & 0 \\ 0 & \tau
\end{pmatrix}
H_2
\begin{pmatrix}
\tau & 0 \\ 0 & \overline{\tau}
\end{pmatrix}
=H_1+H_2,
\end{equation}
and
\begin{equation} \label{eq_sec2_2}
\begin{pmatrix}
0 & 1 \\ 1 & 0
\end{pmatrix}
H_1
\begin{pmatrix}
0 & 1 \\ 1 & 0
\end{pmatrix}
=H_1,\quad
\begin{pmatrix}
0 & 1 \\ 1 & 0
\end{pmatrix}
H_2
\begin{pmatrix}
0 & 1 \\ 1 & 0
\end{pmatrix}
=-H_1-H_2.
\end{equation}
Since $H_i$ are Hermitian, one directly calculates from \eqref{eq_sec2_1} and \eqref{eq_sec2_2} that $H_i=0$. Hence, \eqref{eq_first_order_spectrum_det} indicates that the dispersion surfaces are locally flat near $(\bm{0},\lambda_*)$. This concludes the proof of Proposition \ref{prop_local_flat_double_eigenvalue}.

\subsection{Proof of Theorem \ref{thm_double_cone}}
\label{sec_double_cone}
Under the assumptions of Theorem \ref{thm_double_cone}, one can check that the Bloch Hamiltonian $\mathcal{H}_{b}(\bm{0})$ is invariant under the following group 
\begin{equation} \label{eq_super_symmetry_group}
\begin{aligned}
\tilde{\mathcal{S}}=\langle \mathcal{R}_6,\mathcal{F}_x,\tilde{\mathcal{T}}\big|_{\mathcal{X}_{\bm{0}}} :&\mathcal{R}_6^6=F^2_x=\Big(\tilde{\mathcal{T}}\big|_{\mathcal{X}_{\bm{0}}}\Big)^3=id,\\
&\mathcal{R}_6\mathcal{F}_x=\mathcal{F}_x\mathcal{R}_6^{-1},\, \mathcal{F}_x\tilde{\mathcal{T}}\big|_{\mathcal{X}_{\bm{0}}}=\tilde{\mathcal{T}}\big|_{\mathcal{X}_{\bm{0}}}\mathcal{F}_x,\,
\mathcal{R}_6\tilde{\mathcal{T}}\big|_{\mathcal{X}_{\bm{0}}}=\Big(\tilde{\mathcal{T}}\big|_{\mathcal{X}_{\bm{0}}}\Big)^{-1}\mathcal{R}_6\rangle ,
\end{aligned}
\end{equation}
with the super-symmetry operator $\tilde{\mathcal{T}}$ defined in \eqref{eq_super_symmetry}, and $\tilde{\mathcal{T}}\big|_{\mathcal{X}_{\bm{0}}}$ denoting its restriction to the periodic function space $\mathcal{X}_{\bm{\kappa}=\bm{0}}$. Remarkably, one can check that the finite group $\tilde{\mathcal{S}}$ has a unique four-dimensional irrep built by pinching the two-dimensional irreps $\rho_1$ and $\rho_2$ of $\mathcal{S}$:
\begin{equation} \label{eq_4d_irrep}
\tilde{\rho}:\mathcal{R}_6\mapsto
\begin{pmatrix}
\tau & 0 & 0 & 0 \\
0 & \overline{\tau} & 0 & 0 \\
0 & 0 & \tau^2 & 0 \\
0 & 0 & 0 & \overline{\tau}^2
\end{pmatrix},\quad
\mathcal{F}_x\mapsto
\begin{pmatrix}
0 & 1 & 0 & 0 \\
1 & 0 & 0 & 0 \\
0 & 0 & 0 & 1 \\
0 & 0 & 1 & 0 \\
\end{pmatrix},\quad
\tilde{\mathcal{T}}\big|_{\mathcal{X}_{\bm{0}}}\mapsto
\begin{pmatrix}
-\frac{1}{2} & 0 & 0 & \frac{\sqrt{3}}{2}i \\
0 & -\frac{1}{2} & \frac{\sqrt{3}}{2}i & 0 \\
0 & \frac{\sqrt{3}}{2}i & -\frac{1}{2} & 0 \\
\frac{\sqrt{3}}{2}i & 0 & 0 & -\frac{1}{2} \\
\end{pmatrix}.
\end{equation}
Suppose that $\lambda_*\in \text{Spec}(\mathcal{H}_{b}(\bm{\kappa_*}))$ is an eigenvalue with multiplicity four, with the associated eigenfunctions denoted by $u_j(\cdot;\bm{\kappa_*})$ for $1\leq j\leq 4$, which form a basis of the representation $\tilde{\rho}$. In that case, \eqref{eq_peturb_matrix_symmetry} becomes
\begin{equation} \label{eq_sec2_3}
\tilde{\rho}(\mathcal{R})^{-1}
H_1
\tilde{\rho}(\mathcal{R})
=-H_2,\quad
\tilde{\rho}(\mathcal{R})^{-1}
H_2
\tilde{\rho}(\mathcal{R})
=H_1+H_2,
\end{equation}
\begin{equation} \label{eq_sec2_4}
\tilde{\rho}(\mathcal{F})^{-1}
H_1
\tilde{\rho}(\mathcal{F})
=H_1,\quad
\tilde{\rho}(\mathcal{F})^{-1}
H_2
\tilde{\rho}(\mathcal{F})
=-H_1-H_2,
\end{equation}
\begin{equation} \label{eq_sec2_5}
\tilde{\rho}\big(\tilde{\mathcal{T}}\big|_{\mathcal{X}_{\bm{0}}}\big)^{-1}
H_1
\tilde{\rho}\big(\tilde{\mathcal{T}}\big|_{\mathcal{X}_{\bm{0}}}\big)
=H_1,\quad
\tilde{\rho}\big(\tilde{\mathcal{T}}\big|_{\mathcal{X}_{\bm{0}}}\big)^{-1}
H_2
\tilde{\rho}\big(\tilde{\mathcal{T}}\big|_{\mathcal{X}_{\bm{0}}}\big)
=H_2.
\end{equation}
With \eqref{eq_4d_irrep}, one can solve from \eqref{eq_sec2_3}-\eqref{eq_sec2_5} that
\begin{equation} \label{eq_sec2_7}
H_1=\begin{pmatrix}
0 & 0 & \alpha_* & 0 \\
0 & 0 & 0 & \alpha_* \\
\alpha_* & 0 & 0 & 0 \\
0 & \alpha_* & 0 & 0 \\
\end{pmatrix},\quad
H_2=\begin{pmatrix}
0 & 0 & \overline{\tau}^2\alpha_* & 0 \\
0 & 0 & 0 & \tau^2\alpha_* \\
\tau^2\alpha_* & 0 & 0 & 0 \\
0 & \overline{\tau}^2\alpha_* & 0 & 0 \\
\end{pmatrix}.
\end{equation}
with $\alpha_*:=\Big(u_{1}(\cdot;\bm{0}),\frac{\partial \mathcal{H}_b}{\partial \kappa_1}(\bm{0})u_{3}(\cdot;\bm{0}) \Big)_{\mathcal{X}_{\bm{0}}}\in \mathbb{R}$. Hence, when $\alpha_*\neq 0$, \eqref{eq_first_order_spectrum_det} has four solutions
\begin{equation*}
\begin{aligned}
\lambda_1(\bm{\kappa})=-\frac{\sqrt{3}}{2}|\alpha_*||\bm{\kappa}|, \quad \lambda_2(\bm{\kappa})=-\frac{\sqrt{3}}{2}|\alpha_*||\bm{\kappa}|,\quad
\lambda_3(\bm{\kappa})=\frac{\sqrt{3}}{2}|\alpha_*||\bm{\kappa}|,\quad
\lambda_4(\bm{\kappa})=\frac{\sqrt{3}}{2}|\alpha_*||\bm{\kappa}|,
\end{aligned}
\end{equation*}
which are the linear approximation of four dispersion surfaces that intersect conically.

\section{Band-gap Opening and Asymptotic Expansions of Floquet-Bloch Eigenpairs}
\label{sec_gap_open}

The main result of this paper is the asymptotic expansions of Floquet-Bloch eigenpairs of the perturbed bulk Hamiltonian $\mathcal{H}_{\pm\delta}(\bm{\kappa})$ for small $\bm{\kappa}$ and $\delta$, as shown below in Theorem \ref{thm_asymptotics_eigenpairs}. As we point out in Remark \ref{rmk_corallary_local_gap}, the band gap opening claimed in Theorem \ref{thm_gap_open} follows directly from Theorem \ref{thm_asymptotics_eigenpairs}. More importantly, the asymptotic expansions of Floquet-Bloch eigenmodes demonstrate the appearance of band inversion in the bulk Hamiltonians when we tune $-\delta\to \delta$, which is the fundamental mechanism leading to the interface modes that separates such two bulk Hamiltonians in distinct phases; see Remark \ref{rmk_phase_transition} for details.

\begin{theorem}
\label{thm_asymptotics_eigenpairs}
Suppose that the conditions in Theorem \ref{thm_double_cone}, Assumption \ref{assum_symmetry_break} and the non-degeneracy condition \eqref{eq_gap_open_criterion} hold true. Then there exist $K_0,\Delta_0>0$ such that if $|\bm{\kappa}|<K_0$ and $\delta<\Delta_0$, the first four Floquet-Bloch eigenvalues of $\mathcal{H}_{\pm\delta}(\bm{\kappa})$ admit the following expansions:
\begin{equation} \label{eq_asymptotics_eigenvalue}
\begin{aligned}
\lambda_{1,\pm\delta}(\bm{\kappa})&=\lambda_*-\sqrt{\beta_{*}^2\delta^2+\frac{3}{4}|\alpha_*|^2|\bm{\kappa}|^2}\cdot(1+\mathcal{O}(|\bm{\kappa}|+\delta)), \\
\lambda_{2,\pm\delta}(\bm{\kappa})&=\lambda_*-\sqrt{\beta_{*}^2\delta^2+\frac{3}{4}|\alpha_*|^2|\bm{\kappa}|^2}\cdot(1+\mathcal{O}(|\bm{\kappa}|+\delta)), \\
\lambda_{3,\pm\delta}(\bm{\kappa})&=\lambda_*+\sqrt{\beta_{*}^2\delta^2+\frac{3}{4}|\alpha_*|^2|\bm{\kappa}|^2}\cdot(1+\mathcal{O}(|\bm{\kappa}|+\delta)), \\
\lambda_{4,\pm\delta}(\bm{\kappa})&=\lambda_*+\sqrt{\beta_{*}^2\delta^2+\frac{3}{4}|\alpha_*|^2|\bm{\kappa}|^2}\cdot(1+\mathcal{O}(|\bm{\kappa}|+\delta)).
\end{aligned}
\end{equation}
Moreover, there exist $r_{n,\pm\delta}(\cdot;\bm{\kappa})\in \mathcal{X}_{\bm{\kappa}}$ ($n=1,2,3,4$) with the estimate
\begin{equation*}
\|r_{n,\pm\delta}(\cdot;\bm{\kappa})\|_{\mathcal{X}_{\bm{\kappa}}}=\mathcal{O}(|\bm{\kappa}|+\delta)
\end{equation*}
such that with in the same range $|\bm{\kappa}|<K_0$ and $\delta<\Delta_0$, the Floquet-Bloch eigenfunctions associated with \eqref{eq_asymptotics_eigenvalue} admit the following asymptotics:
\begin{equation} \label{eq_asymptotics_eigenfunction_positive}
\begin{aligned}
u_{1,\delta}(\cdot;\bm{\kappa})&=-\frac{(\kappa_1+\overline{\tau}^2\kappa_2)\alpha_*}{\beta_{*}\delta+\sqrt{\beta_{*}^2\delta^2+\frac{3}{4}|\alpha_*|^2|\bm{\kappa}|^2}}u_1(\cdot)+u_3(\cdot)+r_{1,\delta}(\cdot;\bm{\kappa}), \\
u_{2,\delta}(\cdot;\bm{\kappa})&=-\frac{(\kappa_1+\tau^2\kappa_2)\alpha_*}{\beta_{*}\delta+\sqrt{\beta_{*}^2\delta^2+\frac{3}{4}|\alpha_*|^2|\bm{\kappa}|^2}}u_2(\cdot)+u_4(\cdot)+r_{2,\delta}(\cdot;\bm{\kappa}), \\
u_{3,\delta}(\cdot;\bm{\kappa})&=
u_2(\cdot)+\frac{(\kappa_1+\overline{\tau}^2\kappa_2)\alpha_*}{\beta_{*}\delta+\sqrt{\beta_{*}^2\delta^2+\frac{3}{4}|\alpha_*|^2|\bm{\kappa}|^2}}u_4(\cdot)+r_{3,\delta}(\cdot;\bm{\kappa}), \\
u_{4,\delta}(\cdot;\bm{\kappa})&=u_1(\cdot)+\frac{(\kappa_1+\tau^2\kappa_2)\alpha_*}{\beta_{*}\delta+\sqrt{\beta_{*}^2\delta^2+\frac{3}{4}|\alpha_*|^2|\bm{\kappa}|^2}}u_3(\cdot)+r_{4,\delta}(\cdot;\bm{\kappa}),
\end{aligned}
\end{equation}
and
\begin{equation} \label{eq_asymptotics_eigenfunction_negative}
\begin{aligned}
u_{1,-\delta}(\cdot;\bm{\kappa})&=
u_1(\cdot)-\frac{(\kappa_1+\tau^2\kappa_2)\alpha_*}{\beta_{*}\delta+\sqrt{\beta_{*}^2\delta^2+\frac{3}{4}|\alpha_*|^2|\bm{\kappa}|^2}}u_3(\cdot)+r_{1,-\delta}(\cdot;\bm{\kappa}), \\
u_{2,-\delta}(\cdot;\bm{\kappa})&=u_2(\cdot)-\frac{(\kappa_1+\overline{\tau}^2\kappa_2)\alpha_*}{\beta_{*}\delta+\sqrt{\beta_{*}^2\delta^2+\frac{3}{4}|\alpha_*|^2|\bm{\kappa}|^2}}u_4(\cdot)+r_{2,-\delta}(\cdot;\bm{\kappa})), \\
u_{3,-\delta}(\cdot;\bm{\kappa})&=
\frac{(\kappa_1+\tau^2\kappa_2)\alpha_*}{\beta_{*}\delta+\sqrt{\beta_{*}^2\delta^2+\frac{3}{4}|\alpha_*|^2|\bm{\kappa}|^2}}u_2(\cdot)+u_4(\cdot)+r_{3,-\delta}(\cdot;\bm{\kappa}), \\
u_{4,-\delta}(\cdot;\bm{\kappa})&=\frac{(\kappa_1+\overline{\tau}^2\kappa_2)\alpha_*}{\beta_{*}\delta+\sqrt{\beta_{*}^2\delta^2+\frac{3}{4}|\alpha_*|^2|\bm{\kappa}|^2}}u_1(\cdot)+u_3(\cdot)+r_{4,-\delta}(\cdot;\bm{\kappa}),
\end{aligned}
\end{equation}
where $\bm{\kappa}=\kappa_1\cdot\bm{\ell}_1^*+\kappa_2\cdot\bm{\ell}_2^*$ is expanded in the dual vectors, and the eigenfunctions $\{u_1,u_2,u_3,u_4\}$ form a basis of the four-dimensional representation $\tilde{\rho}$ introduced in \eqref{eq_4d_irrep}.
\end{theorem}

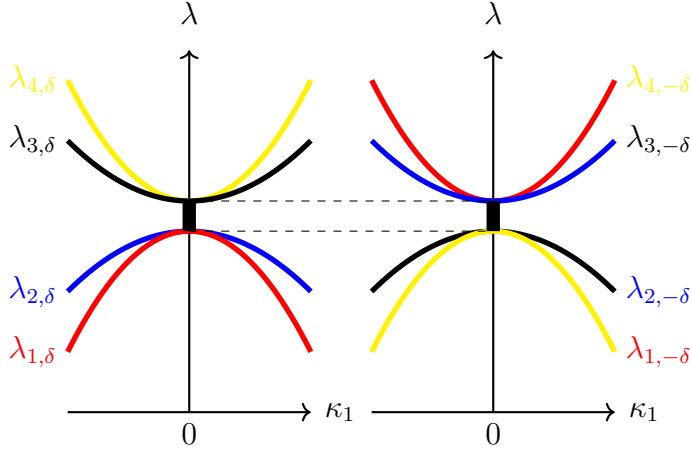
\begin{figure}
\begin{center}
\begin{tikzpicture}[scale=0.2]
\draw[thick,->] (-8,0)--(8,0);
\draw[thick,->] (0,0)--(0,24);
\node[below] at (0,0) {$0$};
\node[right] at (8.2,0) {$\kappa_1$};
\node[above] at (0,25) {$\lambda$};

\draw[yellow,line width=2pt,domain=-8:8, samples=50] plot(\x,{14+(\x)^2/8});
\node[left,yellow] at (-8,22) {$\lambda_{4,\delta}$};
\draw[line width=2pt,domain=-8:8, samples=50] plot(\x,{14+(\x)^2/16});
\node[left] at (-8,18) {$\lambda_{3,\delta}$};
\draw[blue,line width=2pt,domain=-8:8, samples=50] plot(\x,{12-(\x)^2/16});
\node[left,blue] at (-8,8) {$\lambda_{2,\delta}$};
\draw[red,line width=2pt,domain=-8:8, samples=50] plot(\x,{12-(\x)^2/8});
\node[left,red] at (-8,4) {$\lambda_{1,\delta}$};

\draw[thick,->] (12,0)--(28,0);
\draw[thick,->] (20,0)--(20,24);
\node[below] at (20,0) {$0$};
\node[right] at (28.2,0) {$\kappa_1$};
\node[above] at (20,25) {$\lambda$};

\draw[red,line width=2pt,domain=12:28, samples=50] plot(\x,{14+(\x-20)^2/8});
\node[right,yellow] at (28,22) {$\lambda_{4,-\delta}$};
\draw[blue,line width=2pt,domain=12:28, samples=50] plot(\x,{14+(\x-20)^2/16});
\node[right] at (28,18) {$\lambda_{3,-\delta}$};
\draw[line width=2pt,domain=12:28, samples=50] plot(\x,{12-(\x-20)^2/16});
\node[right,blue] at (28,8) {$\lambda_{2,-\delta}$};
\draw[yellow,line width=2pt,domain=12:28, samples=50] plot(\x,{12-(\x-20)^2/8});
\node[right,red] at (28,4) {$\lambda_{1,-\delta}$};

\draw[dashed] (0,14)--(20,14);
\draw[dashed] (0,12)--(20,12);
\draw[line width=5pt] (0,14)--(0,12);
\draw[line width=5pt] (20,14)--(20,12);
\end{tikzpicture}
\caption{Dispersion surfaces of $\mathcal{H}_{\pm\delta}$ (sliced along $\kappa_2=0$). A common band gap (the bold part on the $\lambda-$axis) is opened for both $\mathcal{H}_{\delta}$ and $\mathcal{H}_{-\delta}$. The switching of colours indicates the appearance of band inversion.}
\label{fig_gap_open}
\end{center}
\end{figure}

\begin{remark} \label{rmk_corallary_local_gap}
In Figure \ref{fig_gap_open}, we sketch the dispersion surfaces of $\mathcal{H}_{\pm\delta}$ as described in Theorem \ref{thm_asymptotics_eigenpairs}. One sees that there is a ``local" band gap $\mathcal{I}_\delta$ (defined in Theorem \ref{thm_gap_open}) opened near $\bm{\kappa}=\bm{0}$. Consequently, by the spectral no-fold condition \eqref{eq_no_fold}, it is direct to see that $\mathcal{I}_\delta$ is indeed a `global' band gap of $\mathcal{H}_{\pm\delta}$ as we claimed in Theorem \ref{thm_gap_open}.
\end{remark}

\begin{remark} \label{rmk_phase_transition}
Modulo the higher-order terms, we see from \eqref{eq_asymptotics_eigenfunction_positive} and \eqref{eq_asymptotics_eigenfunction_negative} that
\begin{equation*}
\begin{aligned}
&\text{span}\{u_{1,\delta}(\cdot;\bm{0}),u_{2,\delta}(\cdot;\bm{0})\}=\text{span}\{u_{3,-\delta}(\cdot;\bm{0}),u_{4,-\delta}(\cdot;\bm{0})\}\in \mathcal{X}_{\bm{0}}(\rho_{2}),\\
&\text{span}\{u_{3,\delta}(\cdot;\bm{0}),u_{4,\delta}(\cdot;\bm{0})\}=\text{span}\{u_{1,-\delta}(\cdot;\bm{0}),u_{2,-\delta}(\cdot;\bm{0})\}\in \mathcal{X}_{\bm{0}}(\rho_{1}),
\end{aligned}
\end{equation*}
where $\rho_i$ are the two-dimensional irreps of $C_{6v}$ introduced in Section \ref{sec_local_flat}. This symmetry-switching of eigenspaces between two perturbed systems is called \textit{band inversion} in the physics literature, which typically leads to a sign-changing of Berry curvature and hence is usually deemed as a phase transition (see, e.g., \cite[Section 5.1.1.]{vanderbilt2018berry}). For the interface model with band inversion across the interface, 
the emergence of interface modes as a manifestation of a phase transition is generally expected. This has been rigorously proved for a variety of systems; see the references listed in Section \ref{sec_background}. In particular, as we will see later, the band inversion in our setup is manifested clearly by the special expression of the boundary matching operator to solve the interface modes; see Remark \ref{rmk_absence_band_inversion}.
\end{remark}

\subsection{Proof of Theorem \ref{thm_asymptotics_eigenpairs}}

We aim to solve the following eigenvalue problem in a neighborhood of $(\bm{\kappa},\lambda)=(\bm{0},\lambda_*)$ using perturbation theory
\begin{equation} \label{eq_sec3_1}
\mathcal{H}_{\pm\delta}(\bm{\kappa})u=\lambda u,\quad u\in \mathcal{X}_{\bm{\kappa}}.
\end{equation}
In order to do so, we introduce the following unitary transformation of $\mathcal{H}_{\pm\delta}$, which has the advantages of having a $\bm{\kappa}-$independent domain and hence is a more appropriate candidate for perturbation analysis
\begin{equation} \label{eq_sec3_2}
\begin{aligned}
\underline{\mathcal{H}}_{\pm\delta}(\bm{\kappa}):=e^{-i\bm{\kappa}\cdot \bm{n}} \circ \mathcal{H}_{\pm\delta}(\bm{\kappa})\circ e^{i\bm{\kappa}\cdot \bm{n}}  : \quad &\mathcal{X}_{\bm{0}} \to \mathcal{X}_{\bm{0}}, \\
&u(\bm{n})\mapsto \sum_{\bm{m}\in\Lambda}e^{i\bm{\kappa}\cdot (\bm{n}-\bm{m})}\mathcal{H}_{\pm\delta}(\bm{n},\bm{m})u(\bm{m}).
\end{aligned}
\end{equation}
By unitary equivalence, it is now sufficient to solve the eigenvalue problem of the transformed Hamiltonian instead of \eqref{eq_sec3_1}
\begin{equation} \label{eq_sec3_3}
\underline{\mathcal{H}}_{\pm\delta}(\bm{\kappa})u=\lambda u,\quad u\in \mathcal{X}_{\bm{0}}.
\end{equation}
We note that the transformed Hamiltonians $\underline{\mathcal{H}}_{\pm\delta}(\bm{\kappa})$ depend analytically on both $\bm{\kappa}$ and $\delta$. As a consequence, the stability theorem of the eigenvalue problem for self-adjoint operators indicates that $\underline{\mathcal{H}}_{\pm\delta}(\bm{\kappa})$ has the same number of eigenvalues (counted with their multiplicities) as $\underline{\mathcal{H}}_{\delta=0}(\bm{0})=\underline{\mathcal{H}}_{b}(\bm{0})$ near $\lambda_*$ whenever $\bm{\kappa}$ and $\delta$ are sufficiently small. Hence, we conclude that \eqref{eq_sec3_3} has four solutions because $\lambda_*$ is an eigenvalue of $\underline{\mathcal{H}}_{\pm\delta}(\bm{\kappa})$ with multiplicity four by assumption. In the sequel, we calculate explicit expressions of the solutions to \eqref{eq_sec3_3}, which then completes the proof of Theorem \ref{thm_asymptotics_eigenpairs}. This proceeds in several steps.

{\color{blue}Step 1}. As the first step of perturbation analysis, we expand $\underline{\mathcal{H}}_{\pm\delta,\bm{\kappa}}$ to the first order:
\begin{equation} \label{eq_sec3_4}
\underline{\mathcal{H}}_{\pm\delta}(\bm{\kappa})
=\underline{\mathcal{H}}^{(0,0)}+\kappa_1 \underline{\mathcal{H}}^{(1,0)}_{1}+\kappa_2 \underline{\mathcal{H}}^{(1,0)}_{2} \pm \delta \underline{\mathcal{H}}^{(0,1)} + \underline{\mathcal{H}}_{\pm\delta}^{rem}(\bm{\kappa}),
\end{equation}
where
\begin{equation*}
\underline{\mathcal{H}}^{(0,0)}:=\underline{\mathcal{H}}_{b}(\bm{0}),\quad
\underline{\mathcal{H}}^{(1,0)}_{j}=\partial_{\kappa_j}\underline{\mathcal{H}}_{b}(\bm{0}),\quad
\underline{\mathcal{H}}^{(0,1)}=\underline{\mathcal{H}}_{per},
\end{equation*}
and the remainder $\underline{\mathcal{H}}_{\pm\delta,\bm{\kappa}}^{rem}$ includes all higher-order terms with the following estimate:\footnote{This bound follows directly from the expression \eqref{eq_sec3_2} and the fact that all involved Hamiltonians are finite-range.}
\begin{equation*}
\|\underline{\mathcal{H}}_{\pm\delta}^{rem}(\bm{\kappa})\|_{\mathcal{B}(\mathcal{X}_{\bm{0}})}=\mathcal{O}(|\bm{\kappa}|^2+\delta|\bm{\kappa}|).
\end{equation*}
As we did in Section \ref{sec_representation_analysis}, we introduce the following first-order perturbation matrices:
\begin{equation*}
\underline{H}^{(1,0)}_{j}:=\Big[\big(u_i,\underline{\mathcal{H}}^{(1,0)}_{j}u_k \big)_{\mathbb{C}^6}\Big]_{1\leq i,k\leq 4},\quad
\underline{H}^{(0,1)}:=\Big[\big(u_i,\underline{\mathcal{H}}^{(0,1)}u_k \big)_{\mathbb{C}^6} \Big]_{1\leq i,k\leq 4}
\end{equation*}
with $u_i:=u_i(\cdot;\bm{\kappa}=0)$ being the unperturbed Floquet-Bloch eigenmodes. Note that, by the unitary equivalence \eqref{eq_sec3_2}, the matrices $\underline{H}^{(1,0)}_{j}$ are equal to those $H_j$ introduced in \eqref{eq_sec2_7}; in other words,
\begin{equation} \label{eq_sec3_5}
\underline{H}^{(1,0)}_{1}=\begin{pmatrix}
0 & 0 & \alpha_* & 0 \\
0 & 0 & 0 & \alpha_* \\
\alpha_* & 0 & 0 & 0 \\
0 & \alpha_* & 0 & 0 \\
\end{pmatrix},\quad
\underline{H}^{(1,0)}_{2}=\begin{pmatrix}
0 & 0 & \overline{\tau}^2\alpha_* & 0 \\
0 & 0 & 0 & \tau^2\alpha_* \\
\tau^2\alpha_* & 0 & 0 & 0 \\
0 & \overline{\tau}^2\alpha_* & 0 & 0 \\
\end{pmatrix},\quad
\alpha_*\in \mathbb{R}.
\end{equation}
We also need an explicit form of $\underline{H}^{(0,1)}$, which is given as follows.
\begin{lemma}
\label{lem_delta_perturbation_matrix}
Under Assumption \ref{assum_symmetry_break}, we have
\begin{equation} \label{eq_sec3_6}
\underline{H}^{(0,1)}=
\begin{pmatrix}
\beta_{*,1} & 0 & 0 & 0 \\
0 & \beta_{*,1} & 0 & 0 \\
0 & 0 & \beta_{*,3} & 0 \\
0 & 0 & 0 & \beta_{*,3} \\
\end{pmatrix}
,\quad \beta_{*,k}:=\big(u_{k},\mathcal{H}_{per}(\bm{0}) u_{k} \big)_{\mathcal{X}_{\bm{0}}}\in\mathbb{R}. 
\end{equation}
\end{lemma}

\begin{proof}
Since the operator $\mathcal{H}_{per}(\bm{0})$ is Hermitian, it suffices to show that 1) the off-diagonal elements vanish and 2) $\underline{H}^{(0,1)}_{11}=\underline{H}^{(0,1)}_{22}$, $\underline{H}^{(0,1)}_{33}=\underline{H}^{(0,1)}_{44}$. This, again, is proved by exploiting the symmetry property of $\mathcal{H}_{per}$, following the lines in Section \ref{sec_representation_analysis}. In fact, thanks to the symmetry of $\mathcal{H}_{per}$ supposed in Assumption \ref{assum_symmetry_break} and the unitary equivalence \eqref{eq_sec3_2}, we see that the perturbation matrix $\underline{H}^{(0,1)}$ is invariant under representation of rotation and reflection operators. That is,
\begin{equation} \label{eq_sec3_7}
\tilde{\rho}(\mathcal{R}_6)^{-1}
\underline{H}^{(0,1)}
\tilde{\rho}(\mathcal{R}_6)
=\underline{H}^{(0,1)},\quad
\tilde{\rho}(\mathcal{F}_x)^{-1}
\underline{H}^{(0,1)}
\tilde{\rho}(\mathcal{F}_x)
=\underline{H}^{(0,1)},
\end{equation}
where $\tilde{\rho}$ is given in \eqref{eq_4d_irrep}. A direct computation shows that any Hermitian matrix respecting \eqref{eq_sec3_7} satisfies properties 1) and 2). This concludes the proof.
\end{proof}

{\color{blue}Step 2}. Now we are prepared to solve \eqref{eq_sec3_3} perturbatively, which follows a standard Lyapunov-Schmidt reduction argument as applied extensively to study the perturbed band structure of periodic operators. We show only the proof for the positively perturbed case, i.e., setting $\underline{\mathcal{H}}_{\delta}(\bm{\kappa})$ in \eqref{eq_sec3_3}; the proof for the other case is similar. Let us write
\begin{equation} \label{eq_sec3_8}
\lambda=\lambda_*+\lambda^{(1)},\quad
u=u^{(0)}+u^{(1)},
\end{equation}
with
\begin{equation} \label{eq_sec3_9}
u^{(0)}=\sum_{k=1}^{4}c_k\cdot u_k \in \text{Ker}(\underline{\mathcal{H}}^{(0,0)}-\lambda_*),\quad
u^{(1)}\in \text{Ker}(\underline{\mathcal{H}}^{(0,0)}-\lambda_*)^{\perp},
\end{equation}
where $\perp$ denotes the orthogonal complement in $\mathcal{X}_{\bm{0}}$. Substituting these ansatz into \eqref{eq_sec3_3} produces the equation:
\begin{equation*}
\begin{aligned}
(\underline{\mathcal{H}}^{(0,0)}-\lambda_*)u^{(1)}
=&\big( \lambda^{(1)}-\underline{\mathcal{H}}^{(0,0)}-\kappa_1 \underline{\mathcal{H}}^{(1,0)}_{1}-\kappa_2 \underline{\mathcal{H}}^{(1,0)}_{2}-\delta \underline{\mathcal{H}}^{(0,1)} - \underline{\mathcal{H}}_{\delta}^{rem}(\bm{\kappa}) \big)u^{(0)} \\
&+\big( \lambda^{(1)}-\underline{\mathcal{H}}^{(0,0)}-\kappa_1 \underline{\mathcal{H}}^{(1,0)}_{1}-\kappa_2 \underline{\mathcal{H}}^{(1,0)}_{2}-\delta \underline{\mathcal{H}}^{(0,1)} - \underline{\mathcal{H}}_{\delta}^{rem}(\bm{\kappa}) \big)u^{(1)} .
\end{aligned}
\end{equation*}
It is suggestive to decompose the above identity into orthogonal components: letting $Q_{\parallel}$ be the projection to $\text{Ker}(\underline{\mathcal{H}}^{(0,0)}-\lambda_*)$ and $Q_{\perp}:=\mathbbm{1}-Q_{\parallel}$, we obtain the equivalent equations
\begin{equation} \label{eq_sec3_10}
\begin{aligned}
(\underline{\mathcal{H}}^{(0,0)}-\lambda_*)u^{(1)}
=&Q_{\perp}\big( \lambda^{(1)}-\underline{\mathcal{H}}^{(0,0)}-\kappa_1 \underline{\mathcal{H}}^{(1,0)}_{1}-\kappa_2 \underline{\mathcal{H}}^{(1,0)}_{2}-\delta \underline{\mathcal{H}}^{(0,1)} - \underline{\mathcal{H}}_{\delta}^{rem}(\bm{\kappa}) \big)u^{(0)} \\
&+Q_{\perp}\big( \lambda^{(1)}-\underline{\mathcal{H}}^{(0,0)}-\kappa_1 \underline{\mathcal{H}}^{(1,0)}_{1}-\kappa_2 \underline{\mathcal{H}}^{(1,0)}_{2}-\delta \underline{\mathcal{H}}^{(0,1)} - \underline{\mathcal{H}}_{\delta}^{rem}(\bm{\kappa}) \big)u^{(1)}, 
\end{aligned}
\end{equation}
and
\begin{equation} \label{eq_sec3_11}
\begin{aligned}
0
=&Q_{\parallel}\big( \lambda^{(1)}-\underline{\mathcal{H}}^{(0,0)}-\kappa_1 \underline{\mathcal{H}}^{(1,0)}_{1}-\kappa_2 \underline{\mathcal{H}}^{(1,0)}_{2}-\delta \underline{\mathcal{H}}^{(0,1)} - \underline{\mathcal{H}}_{\delta}^{rem}(\bm{\kappa}) \big)u^{(0)} \\
&+Q_{\parallel}\big( \lambda^{(1)}-\underline{\mathcal{H}}^{(0,0)}-\kappa_1 \underline{\mathcal{H}}^{(1,0)}_{1}-\kappa_2 \underline{\mathcal{H}}^{(1,0)}_{2}-\delta \underline{\mathcal{H}}^{(0,1)} - \underline{\mathcal{H}}_{\delta}^{rem}(\bm{\kappa}) \big)u^{(1)} .
\end{aligned}
\end{equation}

{\color{blue}Step 3}. Now we show that $u^{(1)}$ is uniquely determined by $u^{(0)}$ for small parameters ($\kappa_i,\delta,\lambda^{(1)}$) using \eqref{eq_sec3_10}. Note that the self-adjoint operator $\underline{\mathcal{H}}^{(0,0)}-\lambda_*$ is invertible on the range of $Q_{\perp}$. Hence, we can write \eqref{eq_sec3_10} as
\begin{equation} \label{eq_sec3_12}
(\mathbbm{1}-T)u^{(1)}=Tu^{(1)},
\end{equation}
where
\begin{equation*}
T=T(\bm{\kappa},\delta,\lambda^{(1)})
:=(\underline{\mathcal{H}}^{(0,0)}-\lambda_*)^{-1}Q_{\perp}\big( \lambda^{(1)}-\underline{\mathcal{H}}^{(0,0)}-\kappa_1 \underline{\mathcal{H}}^{(1,0)}_{1}-\kappa_2 \underline{\mathcal{H}}^{(1,0)}_{2}-\delta \underline{\mathcal{H}}^{(0,1)} - \underline{\mathcal{H}}_{\delta}^{rem}(\bm{\kappa}) \big).
\end{equation*}
For sufficiently small parameters, it holds that $\|T\|_{\mathcal{B}(\mathcal{X}_{\bm{0}})}=\mathcal{O}(|\bm{\kappa}|+|\delta|+|\lambda^{(1)}|)$. This implies that there exist $K_0,\Delta_0,E_0>0$ such that whenever $|\bm{\kappa}|<K_0,\delta<\Delta_0,|\lambda^{(1)}|<E_0$, it holds that $\|T\|_{\mathcal{B}(\mathcal{X}_{\bm{0}})}<1$ and therefore, $\mathbbm{1}-T$ is invertible. In conclusion, we can solve $u^{(1)}$ from \eqref{eq_sec3_12}:
\begin{equation} \label{eq_sec3_13}
u^{(1)}=(\mathbbm{1}-T)^{-1}Tu^{(1)}=\sum_{k=1}^{4}c_k\cdot (\mathbbm{1}-T)^{-1}Tu_k .
\end{equation}

{\color{blue}Step 4}. Next, we solve \eqref{eq_sec3_11} to obtain a close expression of the Floquet-Bloch eigenvalues $\lambda^{(1)}=\lambda^{(1)}(\bm{\kappa},\delta)$ and the associated eigenmodes. Taking \eqref{eq_sec3_13} into \eqref{eq_sec3_11} leads to
\begin{equation*}
\begin{aligned}
&\sum_{k=1}^{4}c_k\cdot Q_{\parallel}\big( \lambda^{(1)}-\underline{\mathcal{H}}^{(0,0)}-\kappa_1 \underline{\mathcal{H}}^{(1,0)}_{1}-\kappa_2 \underline{\mathcal{H}}^{(1,0)}_{2}-\delta \underline{\mathcal{H}}^{(0,1)} - \underline{\mathcal{H}}_{\delta}^{rem}(\bm{\kappa}) \big)u_k \\
&+\sum_{k=1}^{4}c_k\cdot Q_{\parallel}\big( \lambda^{(1)}-\underline{\mathcal{H}}^{(0,0)}-\kappa_1 \underline{\mathcal{H}}^{(1,0)}_{1}-\kappa_2 \underline{\mathcal{H}}^{(1,0)}_{2}-\delta \underline{\mathcal{H}}^{(0,1)} - \underline{\mathcal{H}}_{\delta}^{rem}(\bm{\kappa}) \big)(\mathbbm{1}-T)^{-1}Tu_k=0.
\end{aligned}
\end{equation*}
Conjugating both sides with $u_k$ and applying the identities \eqref{eq_sec3_5}-\eqref{eq_sec3_6} gives the following linear equations:
\begin{equation} \label{eq_sec3_14}
\Big(\mathcal{N}^{(0)}(\bm{\kappa},\delta,\lambda^{(1)})+\mathcal{N}^{(1)}(\bm{\kappa},\delta,\lambda^{(1)})\Big)
\begin{pmatrix}
c_1 \\ c_2 \\ c_3 \\ c_4
\end{pmatrix}
=0,
\end{equation}
where
\begin{equation*}
\mathcal{N}^{(0)}=
\begin{pmatrix}
\lambda^{(1)}-\beta_{*,1}\delta & 0 & -\alpha_*(\kappa_1+\overline{\tau}^2\kappa_2) & 0 \\
0 & \lambda^{(1)}-\beta_{*,1}\delta & 0 & -\alpha_*(\kappa_1+\tau^2\kappa_2) \\
-\alpha_*(\kappa_1+\tau^2\kappa_2) & 0 & \lambda^{(1)}-\beta_{*,3}\delta & 0 \\
0 & -\alpha_*(\kappa_1+\overline{\tau}^2\kappa_2) & 0 & \lambda^{(1)}-\beta_{*,3}\delta \\
\end{pmatrix}
\end{equation*}
and
\begin{equation*}
\mathcal{N}^{(1)}=\big(u_i,(-\kappa_{1}\cdot \underline{\mathcal{H}}^{(1,0)}_{1}
-\kappa_{2}\cdot \underline{\mathcal{H}}^{(1,0)}_{2}
-\delta\cdot \underline{\mathcal{H}}^{(0,1)}
-\underline{\mathcal{H}}_{\delta}^{rem}(\bm{\kappa}))(\mathbbm{1}-T)^{-1}T u_j\big)_{1\leq i,j\leq 4}.
\end{equation*}
Note that each entry of $\mathcal{N}^{(1)}$ is of order $\mathcal{O}(|\bm{\kappa}|^2+|\delta|^2+|\lambda^{(1)}|^2)$. Hence, \eqref{eq_sec3_14} is solvable if and only if
\begin{equation} \label{eq_sec3_15}
\det(\mathcal{N}^{(0)}+\mathcal{N}^{(1)})
=\Big((\lambda^{(1)})^2-\beta_{*}^2\delta^2-\frac{3}{4}|\alpha_*|^2|\bm{\kappa}|^2 \Big)^2+g(\bm{\kappa},\delta,\lambda^{(1)})=0,
\end{equation}
where the function $g$ is analytic in all parameters (within a neighborhood of the origin) with the bound $g=\mathcal{O}(|\bm{\kappa}|^5+|\delta|^5+|\lambda^{(1)}|^5)$, and we have applied the assumption $\beta_{*,1}=-\beta_{*,3}=:\beta_{*}$. Note that the leading order equation of \eqref{eq_sec3_15} admits four branches of solutions:
\begin{equation} \label{eq_sec3_16}
\lambda^{(1)}_{1}=\lambda^{(1)}_{2}=-\sqrt{\beta_{*}^2\delta^2+\frac{3}{4}|\alpha_*|^2|\bm{\kappa}|^2}, \quad
\lambda^{(1)}_{3}=\lambda^{(1)}_{4}=\sqrt{\beta_{*}^2\delta^2+\frac{3}{4}|\alpha_*|^2|\bm{\kappa}|^2},
\end{equation}
which are exactly the leading order terms in \eqref{eq_asymptotics_eigenvalue}. When the remainder $g$ is taken into account, the four branches of solutions bifurcate into the four solutions of \eqref{eq_sec3_15}, where the remainders induced by the appearance of $g$ are estimated by the implicit function theorem, and are shown to take the form of \eqref{eq_asymptotics_eigenvalue}; here we omit the technical calculations on this point and refer the interested reader to the previous work \cite[Theorem 3.6]{qiu2026waveguide_localized} for a detailed discussion. After solving $\lambda^{(1)}_k=\lambda^{(1)}_k(\bm{\kappa},\delta)$, the size $K_0,\Delta_0$ imposed in Step 3 is further restricted so that if $|\bm{\kappa}|<K_0,\delta<\Delta_0$, it holds that $|\lambda^{(1)}_k(\bm{\kappa},\delta)|<E_0$; this makes the whole calculation consistent.

{\color{blue}Step 5}. Finally, by taking the eigenvalue $\lambda^{(1)}=\lambda^{(1)}(\bm{\kappa},\delta)$ solved from Step 4 into \eqref{eq_sec3_14}, one can solve the coefficient $c_i$, which then gives the full expansion of the Floquet-Bloch eigenfunction $u=u^{(0)}+u^{(1)}$ using \eqref{eq_sec3_9} and \eqref{eq_sec3_13}. This completes the proof of Theorem \ref{thm_asymptotics_eigenpairs}.

\section{Existence of Interface Modes} \label{sec_sec5}

This section is devoted to the proof of the first main result of this paper, Theorem \ref{thm_existence_interface_modes}. Our proof is based on a layer-potential framework, which is established and exploited in various aspects for continuous systems and now generalized to discrete structures in this paper. The key idea of this framework for studying interface modes lies in the impedance matching method, unveiling the band inversion phenomenon in terms of boundary matching operators\footnote{Specifically \eqref{eq_boundary_matching} and \eqref{eq_boundary_matching_auxliary}.}, and understanding the physical Green operator and its far-field asymptotics at the spectral degenerate points. 

The proof begins in Section \ref{sec_layer_potential_framework}, where a discrete layer potential formulation of the interface modes is established. Prior to that, we will first list some necessary preliminaries in Section \ref{sec_prelim_interface_modes}, especially the definition of the physical Green operator, its asymptotics at the double Dirac cone, and the propagating Floquet-Bloch modes appearing in the asymptotics. After establishing the framework in Section \ref{sec_layer_potential_framework}, we solve the boundary matching problem in the perturbation regime (i.e., for small $\delta$ in \eqref{eq_interface_Hamiltonian}), for doing which the asymptotic expansion of Floquet-Bloch eigenpairs we obtain in Theorem \ref{thm_asymptotics_eigenpairs} and the Green operator asymptotics presented in Section \ref{sec_physical_green} are critical\footnote{Shortly speaking, the asymptotic expansions of Floquet-Bloch eigenpairs are used to derive the limits of the boundary matching equations \eqref{eq_boundary_matching} and \eqref{eq_boundary_matching_auxliary}, see Proposition \ref{prop_boundary_operator_limit}, while the Green operator asymptotics help to solve these limiting equations.}; see Section \ref{sec_existence_interface_mode} for details.

\subsection{Preliminaries}
\label{sec_prelim_interface_modes}

\subsubsection{Analytic Labelling of Floquet-Bloch Eigenpairs}

In order to apply the limiting absorption principle to discuss the far-field Green operator asymptotics at the double Dirac cone, we need a smooth parametrization of the Floquet-Bloch eigenpairs. This is described in this section.

Let us first review the Floquet theory for the Hamiltonian restricted to the cylinder, that is, $\mathcal{H}_{b,\sharp}:=\mathcal{H}_{b}\Big|_{\mathcal{X}_{\sharp}}$. Clearly, $\mathcal{H}_{b,\sharp}$ is the section of the Floquet-Bloch Hamiltonian $\mathcal{H}_{b}(\bm{\kappa})$ along the plane $\kappa_2=0$. As a consequence of the $\mathbb{Z}\bm{\ell}_1$-translation invariance of $\mathcal{H}_{b,\sharp}$, its spectrum is decomposed by Floquet transform, i.e.,
\begin{equation*}
 \text{Spec}(\mathcal{H}_{b,\sharp})=\bigcup_{-\pi\leq \kappa_1<\pi}\text{Spec}(\mathcal{H}_{b,\sharp,\kappa_1}),
\end{equation*}
where $\mathcal{H}_{b,\sharp,\kappa_1}:=\mathcal{H}_{b}(\kappa_1\bm{\ell}_1^*)$. From this relation, if we define
\begin{equation*}
\lambda_{n,\sharp}(\kappa_1):=\lambda_{n}(\bm{\kappa})\quad \text{with }\bm{\kappa}=\kappa_1\bm{\ell}_1^*,
\end{equation*}
where $\lambda_{n}(\bm{\kappa})$ are the eigenvalues of $\mathcal{H}_{b}(\bm{\kappa})$, then it directly follows that
\begin{equation*}
\text{Spec}(\mathcal{H}_{b,\sharp})=\cup_{-\pi\leq \kappa_1<\pi}\text{Spec}(\mathcal{H}_{b,\sharp,\kappa_1})
=\cup_{1\leq n\leq 6}\big\{\lambda_{n,\sharp}(\kappa_1):\, -\pi\leq \kappa_1<\pi \big\}
\end{equation*}
and
\begin{equation*}
\lambda_{1,\sharp}(\kappa_1)\leq\cdots\leq \lambda_{6,\sharp}(\kappa_1).
\end{equation*}
Hence, we say that $\{\lambda_{n,\sharp}(\kappa_1)\}_{1\leq n\leq 6}$ is \textit{the increasingly labeled Floquet-Bloch eigenvalues of $\mathcal{H}_{b,\sharp}$}. Apparently, this increasingly labeling leads to a singularity of $\lambda_{n,\sharp}(\kappa_1)$ for $1\leq n\leq 4$ at $\kappa_1=0$, manifested by the jump of derivatives as seen in Figure \ref{fig_ascending_label}, due to the existence of a double Dirac cone. This singularity is settled by a proper rearrangement of these eigenvalues. In fact, since the Hamiltonian $\mathcal{H}_{b,\sharp,\kappa_1}$ is a six-dimensional matrix and depends analytically on $\kappa_1$, the Kato-Rellich theorem states that there exist six functions $\{\mu_{n,\sharp}(\kappa_1)\}_{1\leq n\neq 6}$ that are analytic in a neighborhood of the real line, and exhaust the eigenvalues of $\mathcal{H}_{b,\sharp,\kappa_1}$ for each $\kappa_1\in [-\pi,\pi]$. We call $\{\mu_{n,\sharp}(\kappa_1)\}_{n\geq 1}$ \textit{the analytically labeled Floquet-Bloch eigenvalues of $\mathcal{H}_{b,\sharp}$}. For each momentum $\kappa_1$, the set of analytically labeled eigenvalues forms a rearrangement of the increasingly labeled ones:
\begin{equation} \label{eq_rearragenment_dispersion}
\cup_{1\leq n\leq 6}\{\mu_{n,\sharp}(\kappa_1)\}
=\cup_{1\leq n\leq 6}\{\lambda_{n,\sharp}(\kappa_1)\},\quad -\pi\leq \kappa_1<\pi.
\end{equation}
We illustrate this rearrangement near the double Dirac cone in Figure \ref{fig_ascending_analytic_labeling}, from which it is clear that the analytical labeling mollifies the singularity arising from the increasing labeling. Remarkably, the Kato-Rellich theorem states that the normalized Floquet-Bloch eigenfunctions associated with $\mu_{n}(\kappa_1)$ can also be chosen analytically\footnote{See \cite[Theorem 3.9, Chapter VII]{kato2013perturbation}. We also refer the reader to \cite[Section 3.3]{joly2016solutions}, the content of which is closely related to this paper; in fact, we learn the method of exploiting Green function asymptotics from \cite{joly2016solutions}, which is fundamental for our study of interface modes.}. One should keep in mind that the choice of the eigenfunctions is not unique (in fact, it is only unique up to a global gauge choice). In the sequel, we will fix \textit{a particular gauge} and denote the analytic eigenfunction associated with $\mu_{n,\sharp}(\kappa_1)$ as $v_{n,\sharp}(\bm{n};\kappa_1)$. This particular choice of gauge is specified by the conditions \eqref{eq_sec4_12} and \eqref{eq_relation_un_vn} that will be introduced later. Nonetheless, we point out that this choice of gauge is only for our convenience in calculation, and the main results of this paper are  gauge-independent.

\begin{figure}
\begin{center}
\subfigure[]{
\label{fig_ascending_label}
\begin{tikzpicture}[scale=0.3]

\draw[thick,->] (-8,0)--(8,0);
\draw[thick,->] (0,0)--(0,24);
\node[below] at (0,0) {$0$};
\node[right] at (8.2,0) {$\kappa_1$};
\node[above] at (0,25) {$\lambda$};

\draw[yellow,line width=2pt,domain=0:8, samples=50] plot(\x,{10+(\x)+(\x)^2/16});
\draw[line width=2pt,domain=0:8, samples=50] plot(\x,{10+(\x)+(\x)^2/64});
\draw[red,line width=2pt,domain=0:8, samples=50] plot(\x,{10-(\x)+(\x)^2/64});
\draw[blue,line width=2pt,domain=0:8, samples=50] plot(\x,{10-(\x)+(\x)^2/16});

\draw[blue,line width=2pt,domain=-8:0, samples=50] plot(\x,{10+(\x)+(\x)^2/16});
\node[left,blue] at (-8,6) {$\lambda_{2,\sharp}(\kappa_1)$};
\draw[red,line width=2pt,domain=-8:0, samples=50] plot(\x,{10+(\x)+(\x)^2/64});
\node[left,red] at (-8,3) {$\lambda_{1,\sharp}(\kappa_1)$};
\draw[line width=2pt,domain=-8:0, samples=50] plot(\x,{10-(\x)+(\x)^2/64});
\node[left] at (-8,19) {$\lambda_{3,\sharp}(\kappa_1)$};
\draw[yellow,line width=2pt,domain=-8:0, samples=50] plot(\x,{10-(\x)+(\x)^2/16});
\node[left,yellow] at (-8,22) {$\lambda_{4,\sharp}(\kappa_1)$};
\end{tikzpicture}
}
\subfigure[]{
\label{fig_analytic_label}
\begin{tikzpicture}[scale=0.3]

\draw[thick,->] (-8,0)--(8,0);
\draw[thick,->] (0,0)--(0,24);
\node[below] at (0,0) {$0$};
\node[right] at (8.2,0) {$\kappa_1$};
\node[above] at (0,25) {$\mu$};

\draw[red,line width=2pt,domain=0:8, samples=50] plot(\x,{10+(\x)+(\x)^2/16});
\node[right,red] at (8,22) {$\mu_{1,\sharp}(\kappa_1)$};
\draw[blue,line width=2pt,domain=0:8, samples=50] plot(\x,{10+(\x)+(\x)^2/64});
\node[right,blue] at (8,19) {$\mu_{2,\sharp}(\kappa_1)$};
\draw[yellow,line width=2pt,domain=0:8, samples=50] plot(\x,{10-(\x)+(\x)^2/64});
\node[right,yellow] at (8,3) {$\mu_{3,\sharp}(\kappa_1)$};
\draw[line width=2pt,domain=0:8, samples=50] plot(\x,{10-(\x)+(\x)^2/16});
\node[right] at (8,6) {$\mu_{4,\sharp}(\kappa_1)$};

\draw[blue,line width=2pt,domain=-8:0, samples=50] plot(\x,{10+(\x)+(\x)^2/16});
\draw[red,line width=2pt,domain=-8:0, samples=50] plot(\x,{10+(\x)+(\x)^2/64});
\draw[line width=2pt,domain=-8:0, samples=50] plot(\x,{10-(\x)+(\x)^2/64});
\draw[yellow,line width=2pt,domain=-8:0, samples=50] plot(\x,{10-(\x)+(\x)^2/16});
\end{tikzpicture}
}
\caption{(a) Increasing labeling. Each coloured curve represents a single branch of Floquet-Bloch eigenvalue $\lambda_{n,\sharp}$. Clearly, the derivatives of $\lambda_{n,\sharp}$ jump at $\kappa_1=0$ due to the conic structure. (b) Analytic labeling. Each colored curve is smooth in $\kappa_1$. Note that this rearrangement of eigenvalues can be written as, for example, $\mu_{1,\sharp}(\kappa_1)=\lambda_{1,\sharp}(\kappa_1)$ for $\kappa_1\leq 0$ and $\mu_{1,\sharp}(\kappa_1)=\lambda_{4,\sharp}(\kappa_1)$ for $\kappa_1> 0$.}
\label{fig_ascending_analytic_labeling}
\end{center}
\end{figure}
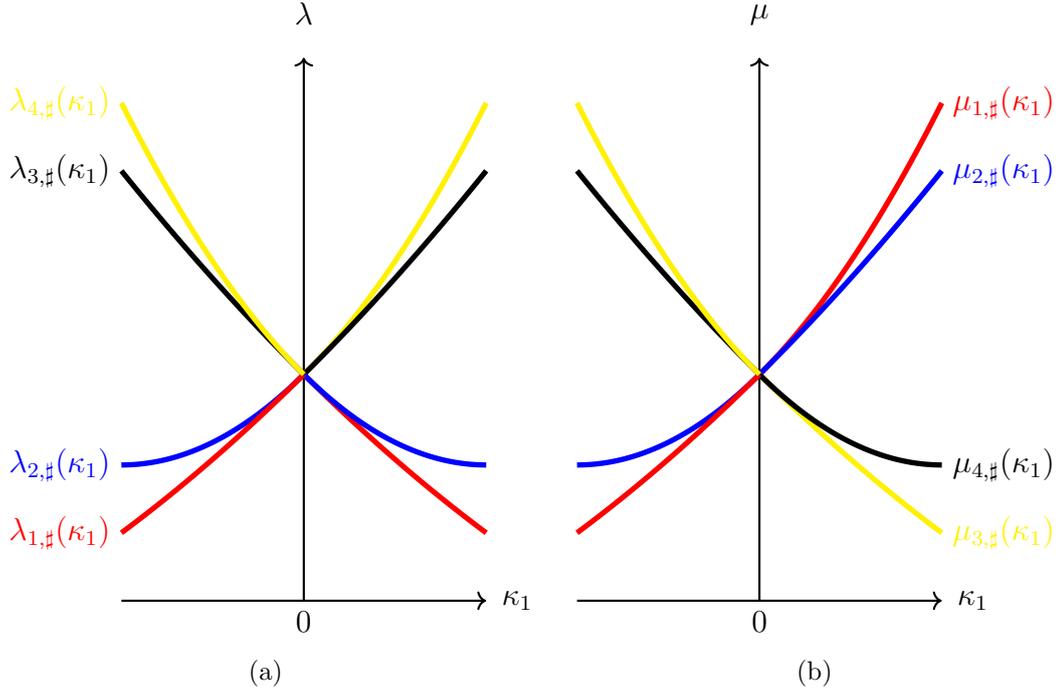

We summarize the above discussions in the following proposition.
\begin{proposition} \label{prop_analytic_label}
For each $1\leq n\leq 6$, there exists a complex neighborhood $\mathcal{D}_n\supset \mathbb{R}$, and two analytic maps 
\begin{equation*}
\kappa_1\in \mathcal{D}_n \mapsto \mu_{n,\sharp}(\kappa_1)\in\mathbb{C},\quad
\kappa_1\in \mathcal{D}_n \mapsto v_{n,\sharp}(\cdot;\kappa_1)\in \mathcal{X}_{\kappa_1\bm{\ell}_1^*}
\end{equation*}
such that $\{(\mu_{n,\sharp}(\kappa_1),v_{n,\sharp}(\cdot;\kappa_1))\}_{1\leq n\leq 6}$ represent all Floquet-Bloch eigenpairs of $\mathcal{H}_{b,\sharp,\kappa_1}$ for $-\pi\leq \kappa_1<\pi$. In particular, $\{v_{n,\sharp}\}$ satisfy conditions \eqref{eq_sec4_12} and \eqref{eq_relation_un_vn}.
\end{proposition}
Moreover, as a manifestation of the double Dirac cone described in Theorem \ref{thm_double_cone} and the $\mathcal{F}_x$-symmetry of the system, the analytically labeled eiegnavalues have the following expansions.
\begin{proposition} \label{prop_mu1234_asymptotic_even}
The analytically labeled eigenvalues are symmetric and admit the following asymptotics for $\kappa_1$ near zero:
\begin{equation} \label{eq_mu1234_dispersion}
 \left\{
 \begin{aligned}
    &\mu_{1,\sharp}(\kappa_1)=\mu_{3,\sharp}(-\kappa_1)=\lambda_*+\frac{\sqrt{3}}{2}|\alpha_*|\kappa_1+\mathcal{O}(|\kappa_1|^2), \\
    &\mu_{2,\sharp}(\kappa_1)=\mu_{4,\sharp}(-\kappa_1)=\lambda_*+\frac{\sqrt{3}}{2}|\alpha_*|\kappa_1+\mathcal{O}(|\kappa_1|^2).
    \end{aligned}
    \right.
\end{equation}
\end{proposition}

\begin{remark}
\label{rmk_construction_analytic_eigenfunction}
As pointed out in the preceding paragraphs, the analytic choice of Floquet-Bloch eigenfunctions is not unique. Here, we give an explicit construction of $v_{n,\sharp}(\cdot;\kappa_1)$ ($n=1,2,3,4$) near $\kappa_1=0$ satisfying the following symmetry properties:
\begin{equation} \label{eq_sec4_12}
v_{3,\sharp}(\cdot;\kappa_1)=\mathcal{F}_y v_{1,\sharp}(\cdot;-\kappa_1),\quad
v_{4,\sharp}(\cdot;\kappa_1)=\mathcal{F}_y v_{2,\sharp}(\cdot;-\kappa_1),
\end{equation}
and
\begin{equation}
\label{eq_relation_un_vn}
\begin{pmatrix}
v_{1,\sharp}(\cdot;0) \\ v_{2,\sharp}(\cdot;0) \\ v_{3,\sharp}(\cdot;0) \\ v_{4,\sharp}(\cdot;0)
\end{pmatrix}
=\frac{1}{2}
\begin{pmatrix}
\text{sgn}(\alpha_*) & \text{sgn}(\alpha_*) & 1 & 1 \\
\text{sgn}(\alpha_*) & -\text{sgn}(\alpha_*) & 1 & -1 \\
-\text{sgn}(\alpha_*) & -\text{sgn}(\alpha_*) & 1 & 1 \\
\text{sgn}(\alpha_*) & -\text{sgn}(\alpha_*) & -1 & 1
\end{pmatrix}
\begin{pmatrix}
u_1(\cdot) \\ u_2(\cdot) \\ u_3(\cdot) \\ u_4(\cdot)
\end{pmatrix},
\end{equation}
where $\{u_n\}$ are the basis of the four-dimensional irrep $\tilde{\rho}$ fixed in Theorem \ref{thm_asymptotics_eigenpairs}. The proof is based on the perturbation argument in Section \ref{sec_gap_open} (setting $\delta=0$ and $\kappa_2=0$ therein). This choice of analytic Floquet-Bloch eigenfunctions, as we will see later in Section \ref{sec_boundary_operator_limit}-\ref{sec_number_interface_mode}, is particularly convenient for studying interface modes.

Let us consider the first branch, i.e., $v_{1,\sharp}$. As seen in Section 3.1, for $\kappa_1$ near zero, $v_{1,\sharp}(\cdot;\kappa_1)$ is totally determined by $v_{1,\sharp}(\cdot;0)$, its value at zero momentum\footnote{In fact, after fixing a particular vector $v\in \mathcal{X}_{\bm{0}}$, one just applies the formula \eqref{eq_sec3_9} and \eqref{eq_sec3_13} with $u^{(0)}=v$. This will produce an analytic family of Floquet-Bloch eigenvalues of $\mathcal{H}_{b,\sharp,\kappa_1}$, which is equal to $v$ for $\kappa_1=0$.}. In other words, it is sufficient to determine the projection coefficients of $v_{1,\sharp}(\cdot;0)$ to $\text{span}\{u_n:n=1,2,3,4\}$. Taking $\lambda^{(1)}=\mu_{n,\sharp}(\kappa_1)-\lambda_*=\frac{\sqrt{3}}{2}|\alpha_*|\kappa_1+\mathcal{O}(|\kappa_1|^2)$, with $n=1,2$, in \eqref{eq_sec3_14}, one solves these coefficients and obtains 
\begin{equation} \label{eq_sec4_10}
v_{1,\sharp}(\cdot;0)=a\cdot \big(\text{sgn}(\alpha_*)u_1+u_3 \big)
+b\cdot \big(\text{sgn}(\alpha_*)u_2+u_4 \big),
\end{equation}
and
\begin{equation} \label{eq_sec4_11}
v_{2,\sharp}(\cdot;0)=c\cdot \big(\text{sgn}(\alpha_*)u_1+u_3 \big)
+d\cdot \big(\text{sgn}(\alpha_*)u_2+u_4 \big),
\end{equation}
where $a,b,c,d\in\mathbb{C}$ are constants to be determined. Once we constructed the first two branches of eigenfunctions $v_{1,\sharp}(\bm{x};\kappa_1)$ and $v_{2,\sharp}(\bm{x};\kappa_1)$ using \eqref{eq_sec4_10} and \eqref{eq_sec4_11}, we apply the $\mathcal{F}_y-$symmetry to fix the other two branches $v_{3,\sharp}(\bm{x};\kappa_1)$ and $v_{4,\sharp}(\bm{x};\kappa_1)$ as in \eqref{eq_sec4_12}. Now we exploit the $\mathcal{F}_x$-symmetry to calculate the coefficients $a,b,c,d$ in \eqref{eq_sec4_10}-\eqref{eq_sec4_11}, and, as we will see, this will lead to \eqref{eq_relation_un_vn}. Remark that, due to the periodicity of $\mathcal{X}_{\kappa_1\bm{\ell}_1^*}$ along $\bm{\ell}_2$, $\mathcal{F}_xv_{k,\sharp}(\cdot;\kappa_1)$ is still a Floquet-Bloch eigenfunction of $\mathcal{H}_{b,\sharp,\kappa_1}$. This implies that, in a neighborhood of $\kappa_1=0$, it must hold that $\mathcal{F}_x v_{k,\sharp}(\cdot;\kappa_1)$ is parallel to $v_{k,\sharp}(\cdot;\kappa_1)$ for $k=1,2$. Hence, with \eqref{eq_sec4_10}-\eqref{eq_sec4_11} and the symmetry of $u_n$ specified in the representation \eqref{eq_4d_irrep}, one can solve that
\begin{equation} \label{eq_sec4_13}
a^2=b^2,\quad c^2=d^2.
\end{equation}
On the other hand, the orthogonality between $v_{1,\sharp}(\bm{x};\kappa_1)$ and $v_{2,\sharp}(\bm{x};\kappa_1)$ gives
\begin{equation} \label{eq_sec4_14}
ac+bd=0.
\end{equation}
A nontrivial solution of \eqref{eq_sec4_13}-\eqref{eq_sec4_14} is $a=b=c=\frac{1}{2},d=-\frac{1}{2}$ (the $\frac{1}{2}-$factor is added for normalization), with which  \eqref{eq_relation_un_vn} can be directly verified using formulas \eqref{eq_sec4_10} and \eqref{eq_sec4_11}.
\end{remark}

\subsubsection{Limiting Absorption Principle and the Physical Green Operator} \label{sec_physical_green}

As pointed out at the beginning of this section, the physical Green operator and its far-field asymptotics play an important role in the study of interface modes. To be more precise, as we shall see later in Section \ref{sec_existence_interface_mode}, the physical Green operator lies in the dominant part of the boundary matching equations for solving the interface modes, while its far-field asymptotics is fundamental for the solvability of these equations.

Let us first introduce the definition of the physical Green operator. For the Hamiltonian $\mathcal{H}_{b,\sharp}$, the Green function is typically defined as the kernel associated with the resolvents $(\mathcal{H}_{b,\sharp}-\lambda)^{-1}$. However, it is apparently ill-defined for $\lambda=\lambda_*$, which lies in the spectrum $\text{Spec}(\mathcal{H}_{b,\sharp})$ as the energy level of the double Dirac cone. Nonetheless, thanks to the special spectral property of the double Dirac cone, \textit{especially the linear dispersion nearby}, one can still define \text{a physical Green operator} by choosing an appropriate radiation condition at infinity in the sense that it is still a right inverse of the Hamiltonian (see \eqref{eq_physical_green_right_inverse}). The choice of radiation condition is specified by the limiting absorption principle. To be precise, we denote $\mathcal{G}_{b,\sharp}(\lambda):=(\mathcal{H}_{b,\sharp}-\lambda)^{-1}$ as the resolvents associated with the Hamiltonian $\mathcal{H}_{b,\sharp}$. Following \cite[Theorem 6]{joly2016solutions}, one can prove that the following limit exists for any $u\in\ell^1(\Lambda/\mathbb{Z}\bm{\ell}_2)\otimes \mathbb{C}^6 \subset \mathcal{X}_{\sharp}$:\footnote{Note that the original setup of \cite[Theorem 6]{joly2016solutions} is on the resolvents of the Laplacian; nevertheless, the proof relies only on the band structure of the operator, not on its differential nature, and can be applied directly for the discrete Hamiltonian $\mathcal{H}_{b,\sharp}$. (In fact, the proof is even simpler in this case because there is no regularity problem to deal with.)}
\begin{equation} \label{eq_sec_l.a._1}
\mathcal{G}_{b,\sharp}^{l.a.}u:=\lim_{\epsilon\to 0^+}(\mathcal{H}_{b,\sharp}-\lambda_*-i\epsilon)^{-1}u.
\end{equation}
The physical Green operator is contained in $\mathcal{G}_{b,\sharp}^{l.a.}$. To see this, let us look at the main idea of the proof of \eqref{eq_sec_l.a._1}, while we refer the reader to the original paper \cite{joly2016solutions} for details. Using Fourier expansion, $(\mathcal{H}_{b,\sharp}-\lambda_*-i\epsilon)^{-1}u$ is written as a linear combination of the Floquet-Bloch eigenmodes of $\mathcal{H}_{b,\sharp}$:\footnote{There is a little abuse of notation in the bracket $\big(u(\cdot),v_{n,\sharp}(\cdot;\kappa_1)\big)_{\mathcal{X}_{\sharp}}$, which actually refers to the $\ell^1-\ell^{\infty}$ pairing (recall that $v_{n,\sharp}(\cdot;\kappa_1)$ is not $\ell^2$ localized), instead of the inner product of $\mathcal{X}_{\sharp}$. Nevertheless, their expressions are the same whenever the pairing is well-defined. We will keep this convention for ease of notation throughout this paper.}
\begin{equation}
\label{eq_sec_l.a._2}
(\mathcal{H}_{b,\sharp}-\lambda_*-i\epsilon)^{-1}u=\frac{1}{2\pi}\int_{-\pi}^{\pi}d\kappa_1 \sum_{1\leq n\leq 6}\frac{v_{n,\sharp}(\bm{n};\kappa_1)\big(u(\cdot),v_{n,\sharp}(\cdot;\kappa_1)\big)_{\mathcal{X}_{\sharp}}}{\mu_{n,\sharp}(\kappa_1)-(\lambda_*+i\epsilon)},
\end{equation}
where $(\mu_{n,\sharp},v_{n,\sharp})$ are the analytically labeled Floquet-Bloch eigenpairs introduced in the last section. From \eqref{eq_sec_l.a._2}, it is clear that $(\mathcal{H}_{b,\sharp}-\lambda_*)^{-1}$ is ill-defined due to the singularity in the denominator at the double Dirac cone (i.e., for $n=1,2,3,4$ and $\kappa_1=0$). However, \textit{thanks to the smoothness of the analytically labeled eigenfunction $\mu_{n,\sharp}$ and its non-vanishing derivatives at $\kappa_1=0$}, the limit of \eqref{eq_sec_l.a._2} for $\epsilon\to 0^+$ exists, which consists of a principal-value integral and a projection term at $\kappa_1=0$:
\begin{equation}
\label{eq_sec_l.a._3}
\begin{aligned}
&\lim_{\epsilon\to 0^+}\frac{1}{2\pi}\int_{-\pi}^{\pi}d\kappa_1 \sum_{1\leq n\leq 6}\frac{v_{n,\sharp}(\bm{n};\kappa_1)\big(u(\cdot),v_{n,\sharp}(\cdot;\kappa_1)\big)_{\mathcal{X}_{\sharp}}}{\mu_{n,\sharp}(\kappa_1)-(\lambda_*+i\epsilon)} \\
&=\frac{1}{2\pi}\text{p.v.}\int_{-\pi}^{\pi}d\kappa_1 \sum_{1\leq n\leq 6}\frac{v_{n,\sharp}(\bm{n};\kappa_1)\big(u(\cdot),v_{n,\sharp}(\cdot;\kappa_1)\big)_{\mathcal{X}_{\sharp}}}{\mu_{n,\sharp}(\kappa_1)-(\lambda_*+i\epsilon)}
+\frac{i}{2|\alpha_*|}\sum_{1\leq n\leq 4}v_{n,\sharp}(\bm{n};0)\big(u(\cdot),v_{n,\sharp}(\cdot;0)\big)_{\mathcal{X}_{\sharp}} \\
&=: \mathcal{G}_{b,\sharp}^{p.v.}u + \frac{i}{2|\alpha_*|}\sum_{1\leq n\leq 4}v_{n,\sharp}(\bm{n};0)\big(u(\cdot),v_{n,\sharp}(\cdot;0)\big)_{\mathcal{X}_{\sharp}}.
\end{aligned}
\end{equation}
Here, p.v. indicates that the $\kappa_1-$integral is understood in the sense of a Cauchy principal-value integral\footnote{In fact, the principal-value integral is only necessary for $1\leq n\leq 4$ in \eqref{eq_sec_l.a._3}, since there is no singularity in the integrands for $n=5,6$.}. The integral part, i.e., $\mathcal{G}_{b,\sharp}^{p.v.}$, is referred to as \textit{the physical Green operator}. Its basic properties are summarized as follows.

\begin{proposition} \label{prop_physical_green}
The physical Green operator $\mathcal{G}_{b,\sharp}^{p.v.}:\, \ell^1(\Lambda/\mathbb{Z}\bm{\ell}_2)\otimes \mathbb{C}^6 \subset \mathcal{X}_{\sharp} \to \ell^{\infty}(\Lambda/\mathbb{Z}\bm{\ell}_2)\otimes \mathbb{C}^6$ is bounded. It is a right-inverse of the Hamiltonian in the sense that
\begin{equation} \label{eq_physical_green_right_inverse}
(\mathcal{H}_{b,\sharp}-\lambda_*)\mathcal{G}_{b,\sharp}^{p.v.}=\mathbbm{1} \quad \text{on $\ell^1(\Lambda/\mathbb{Z}\bm{\ell}_2)\otimes \mathbb{C}^6$}.
\end{equation}
Moreover, for any $u\in\ell^1(\Lambda/\mathbb{Z}\bm{\ell}_2)\otimes \mathbb{C}^6$, the function $\mathcal{G}_{b,\sharp}^{p.v.}u$ admits the following far-field asymptotics
\begin{equation} \label{eq_physical_green_decay_1}
\begin{aligned}
\lim_{\bm{n}\cdot\bm{\ell}_1\to\infty}\Big[(\mathcal{G}_{b,\sharp}^{p.v.}u)(\bm{n})-\frac{i}{2|\alpha_*|}\big(& v_{1,\sharp}(\bm{n};0)\big(u(\cdot),v_{1,\sharp}(\cdot;0)\big)_{\mathcal{X}_{\sharp}}+v_{2,\sharp}(\bm{n};0)\big(u(\cdot),v_{2,\sharp}(\cdot;0)\big)_{\mathcal{X}_{\sharp}} \\
&-v_{3,\sharp}(\bm{n};0)\big(u(\cdot),v_{3,\sharp}(\cdot;0)\big)_{\mathcal{X}_{\sharp}}-v_{4,\sharp}(\bm{n};0)\big(u(\cdot),v_{4,\sharp}(\cdot;0)\big)_{\mathcal{X}_{\sharp}}  \big) \Big] =0,
\end{aligned}
\end{equation}
and
\begin{equation} \label{eq_physical_green_decay_2}
\begin{aligned}
\lim_{\bm{n}\cdot\bm{\ell}_1\to -\infty}\Big[(\mathcal{G}_{b,\sharp}^{p.v.}u)(\bm{n}) + \frac{i}{2|\alpha_*|}\big(& v_{1,\sharp}(\bm{n};0)\big(u(\cdot),v_{1,\sharp}(\cdot;0)\big)_{\mathcal{X}_{\sharp}}+v_{2,\sharp}(\bm{n};0)\big(u(\cdot),v_{2,\sharp}(\cdot;0)\big)_{\mathcal{X}_{\sharp}} \\
&-v_{3,\sharp}(\bm{n};0)\big(u(\cdot),v_{3,\sharp}(\cdot;0)\big)_{\mathcal{X}_{\sharp}}-v_{4,\sharp}(\bm{n};0)\big(u(\cdot),v_{4,\sharp}(\cdot;0)\big)_{\mathcal{X}_{\sharp}}  \big) \Big] =0,
\end{aligned}
\end{equation}
where the convergence rate is exponential.
\end{proposition}

The boundedness of the physical Green operator follows from the standard estimate of the principal-value integral, following the lines of \cite[Proposition 4.4]{qiu2026waveguide_localized}. The invertibility \eqref{eq_physical_green_right_inverse} and the far-field asymptotics \eqref{eq_physical_green_decay_1}-\eqref{eq_physical_green_decay_2} are proved by the same argument as in \cite[Theorem 6]{joly2016solutions}.

\begin{remark}
The physical meaning of the far-field asymptotics presented in Proposition \ref{prop_physical_green} will be clear if we substitute them into \eqref{eq_sec_l.a._3}. In fact, by this substitution, one sees that the operator $\mathcal{G}_{b,\sharp}^{l.a.}$ obtained by the limiting absorption principle satisfies
\begin{equation*}
\begin{aligned}
\lim_{\bm{n}\cdot\bm{\ell}_1\to\infty}\Big[(\mathcal{G}_{b,\sharp}^{l.a.}u)(\bm{n})-\frac{i}{|\alpha_*|}\big(& v_{1,\sharp}(\bm{n};0)\big(u(\cdot),v_{1,\sharp}(\cdot;0)\big)_{\mathcal{X}_{\sharp}}+v_{2,\sharp}(\bm{n};0)\big(u(\cdot),v_{2,\sharp}(\cdot;0)\big)_{\mathcal{X}_{\sharp}}   \big) \Big] =0,
\end{aligned}
\end{equation*}
and
\begin{equation*}
\lim_{\bm{n}\cdot\bm{\ell}_1\to -\infty}\Big[(\mathcal{G}_{b,\sharp}^{l.a.}u)(\bm{n})-\frac{i}{|\alpha_*|}\big( v_{3,\sharp}(\bm{n};0)\big(u(\cdot),v_{3,\sharp}(\cdot;0)\big)_{\mathcal{X}_{\sharp}}+v_{4,\sharp}(\bm{n};0)\big(u(\cdot),v_{4,\sharp}(\cdot;0)\big)_{\mathcal{X}_{\sharp}}   \big) \Big] =0.
\end{equation*}
This means that, for any $u$, $\mathcal{G}_{b,\sharp}^{l.a.}u$ consists of an evanescent part and a right-propagating (left-propagating) part, where the propagating part consists exactly of \textit{the Floquet-Bloch modes at the double Dirac cone with positive (negative) group velocity}, i.e., $v_{n,\sharp}(\cdot;0)$ for $n=1,2$ ($n=3,4$), as seen from Figure \ref{fig_analytic_label}. In the following section, the distinct propagation directions of the Floquet-Bloch modes at the double Dirac cone are illustrated in more detail, based on the notion of energy flux.
\end{remark}

\subsubsection{Discrete Energy Flux of Floquet-Bloch Modes at the Double Dirac Cone}
\label{sec_energy_flux}

In this section, we introduce a sesquilinear form defined on $\mathcal{X}_{\sharp}$. When applied on the Floquet-Bloch eigenmodes $v_{n,\sharp}(\cdot;\kappa_1)$, it produces the energy flux carried by the corresponding eigenmodes (or alternatively, the group velocity, if we drop all units). Moreover, as we will see later in Section \ref{sec_existence_interface_mode}, this sesquilinear form provides a convenient orthogonalization of $\mathcal{X}_{\sharp}$ to solve the boundary matching equation and to obtain the interface modes.

Let us first fix the notation. Since $\mathcal{X}_{\sharp}\simeq \ell^2 (\Lambda_1)\otimes \mathbb{C}^6$ is associated with the one-dimensional lattice $\Lambda_1=\mathbb{Z}\bm{\ell}_1$ generated by the unit vector $\bm{\ell}_1$, we will call $\mathcal{X}_{\sharp}$ \textit{the strip space}, and the bulk Hamiltonian $\mathcal{H}_{b,\sharp}$ on $\mathcal{X}_{\sharp}$ \textit{the strip Hamiltonian}. It is convenient to write the strip Hamiltonian into a nearest-neighbor-hopping form. This is achieved by introducing the following isomorphism:
\begin{equation} \label{eq_NN_isomorphism}
\begin{aligned}
\iota: \mathbb{Z} \times (\{1,\cdots,N\}\times \{1,\cdots,6\}) &\to \Lambda_1 \times \{1,\cdots,6\}, \\
(n,(\sigma_1,\sigma_2))& \mapsto \big((n\cdot N +\sigma_1)\bm{\ell}_1,\sigma_2\big),
\end{aligned}
\end{equation}
where $N>0$ is the Hamiltonian range (see Definition \ref{asmp_short_range}). With a little abuse of notation, the induced isomorphism by $\iota$ between $\tilde{\mathcal{X}}_{\sharp}:=\ell^2(\mathbb{Z})\otimes \mathbb{C}^{6N}$ (referred to as \textit{the blocked strip space}) and $\mathcal{X}_{\sharp}$ is still denoted by $\iota$, and its adjoint is denoted by $\iota^*$. With this notation, \textit{the nearest-neighbor-hopping (NNH) strip Hamiltonian is defined as}
\begin{equation*}
\tilde{\mathcal{H}}_{b,\sharp}:=\iota^* \mathcal{H}_{b,\sharp} \iota \in \mathcal{B}(\tilde{\mathcal{X}}_{\sharp}).
\end{equation*}
As suggested by the name, the pivotal reason for introducing the NNH strip Hamiltonian is that its kernel satisfies
\begin{equation*}
\tilde{\mathcal{H}}_{b,\sharp}(n,m)=0,\quad \text{for }|n-m|>1 .
\end{equation*}
Intuitively, we are encapsulating $N$ sites of the lattice (on which the original Hamiltonian $\mathcal{H}_{b,\sharp}$ is defined) into a single block such that the hopping term in the NNH Hamiltonian $\tilde{\mathcal{H}}_{b,\sharp}$ is restricted to the nearest neighbors.

With these preparations, we are now ready to introduce the sesquilinear form of energy flux.
\begin{definition}[Energy Flux Form] \label{def_sesqui_form}
For any $\phi,\varphi\in \tilde{\mathcal{X}}_{\sharp}$, the sesquilinear form of the energy flux at the site $n\in\mathbb{Z}$ is defined as
\begin{equation*}
\begin{aligned}
\mathfrak{a}(\phi,\varphi;n)&:=\big(\tilde{\mathcal{H}}_{b,\sharp}(n,n-1)\phi(n-1),\varphi(n) \big)_{\mathbb{C}^{6N}}-\big(\tilde{\mathcal{H}}_{b,\sharp}(n-1,n)\phi(n),\varphi(n-1) \big)_{\mathbb{C}^{6N}} \\
&=\big(\tilde{\mathcal{H}}_{b,\sharp}(1,0)\phi(n-1),\varphi(n) \big)_{\mathbb{C}^{6N}}-\big(\tilde{\mathcal{H}}_{b,\sharp}(0,1)\phi(n),\varphi(n-1) \big)_{\mathbb{C}^{6N}} .
\end{aligned}
\end{equation*}
Or more explicitly,
\begin{equation*}
\begin{aligned}
\mathfrak{a}(\phi,\varphi;n)=\overline{\varphi(n)}^{\top}\cdot\tilde{\mathcal{H}}_{b,\sharp}(1,0)\phi(n-1)-\overline{\varphi(n-1)}^{\top}\cdot \tilde{\mathcal{H}}_{b,\sharp}(0,1)\phi(n).
\end{aligned}
\end{equation*}
\end{definition}
 Note that one of the advantages of introducing the NNH strip Hamiltonian can be seen from Definition \ref{def_sesqui_form}: the inner product defining the sesquilinear form $\mathfrak{a}(\cdot,\cdot)$, which is actually a sum of $N$ terms if calculated on $\mathcal{X}_{\sharp}$, is now encapsulated into a simple product on $\tilde{\mathcal{X}}_{\sharp}$ that significantly simplifies the notation.

The main properties of this sesquilinear form are described as follows, which are important for the analysis of interface modes in the following sections. Their proof follows the lines of \cite[Theorem 3]{joly2016solutions} (see Remark \ref{rmk_integral_by_parts} for a brief comment). For ease of notation, we denote by
\begin{equation} \label{eq_blocked_analytic_Bloch_mode}
\tilde{v}_n:=\iota^* v_n(\cdot;0)\quad \text{for }1\leq n\leq 4,
\end{equation}
the blocked Floquet-Bloch eigenmodes at the double Dirac cone.

\begin{proposition} \label{prop_sesqui_form}
The value of $\mathfrak{a}(\tilde{v}_i,\tilde{v}_j;n)$ is independent of the site $n$, that is,
\begin{equation} \label{eq_sesqui_form_1}
\mathfrak{a}(\tilde{v}_i,\tilde{v}_j;n)\equiv \mathfrak{a}(\tilde{v}_i,\tilde{v}_j;0)=:\mathfrak{a}(\tilde{v}_i,\tilde{v}_j) \quad (1\leq i,j\leq 4). 
\end{equation}
When $i=j$, $\mathfrak{a}(\tilde{v}_i,\tilde{v}_j)$ is proportional to the slope of the dispersion curve $\mu_{j,\sharp}$ at $\kappa_1=0$, i.e.,
\begin{equation} \label{eq_sesqui_form_2}
\mathfrak{a}(\tilde{v}_j,\tilde{v}_j)=
\left\{
\begin{aligned}
&i|\alpha_*|,\quad j=1,2, \\ &-i|\alpha_*|,\quad j=3,4.
\end{aligned}
\right.
\end{equation}
When $i\neq j$, it holds that
\begin{equation} \label{eq_sesqui_form_3}
\mathfrak{a}(\tilde{v}_i,\tilde{v}_j)=0.
\end{equation}
\end{proposition}

We interpret \eqref{prop_sesqui_form} as follows. Physically, \eqref{eq_sesqui_form_2} justifies the name of the sesquilinear form $\mathfrak{a}(\cdot,\cdot)$: when applying it on the Floquet-Bloch eigenmode, its value is proportional to the energy flux carried by these modes, since the latter quantity equals the group velocity if we drop all units. Moreover, \eqref{eq_sesqui_form_1} states that the energy flux has the same value when evaluated on different sites; this follows simply from the fact that $\tilde{v}_i$ are eigenmodes of $\tilde{\mathcal{H}}_{b,\sharp}$, i.e., $$(\tilde{\mathcal{H}}_{b,\sharp}-\lambda_*)\tilde{v}_i=0$$ and no sources appear. Finally, by \eqref{eq_sesqui_form_3}, we see that different eigenmodes are decoupled under the action of $\mathfrak{a}(\cdot,\cdot)$. From a mathematical point of view, Proposition \ref{prop_sesqui_form} indicates that the blocked eigenmodes $\tilde{v}_i$ are orthogonalized by the sesquilinear form $\mathfrak{a}(\cdot,\cdot)$. As we stated at the beginning of this section, this orthogonalization appears to be important for the study of interface modes.

We refer the reader to the original paper \cite{joly2016solutions}, from which we learn the idea of energy flux associated with the Floquet-Bloch modes, for a detailed discussion on the continuous version of the sesquilinear form $\mathfrak{a}(\cdot,\cdot)$. We have also extended some aspects of this form that are not proved in \cite{joly2016solutions}; see \cite[Section 6]{qiu2024square_lattice}.

\begin{remark} \label{rmk_integral_by_parts}
We briefly comment on the proof of Proposition \ref{prop_sesqui_form}. Relying only on the band structure of the operator, the proof of \cite[Theorem 3]{joly2016solutions} directly carries to our case. We note that the key step in \cite[Theorem 3]{joly2016solutions} is the Gauss-Green formula that applies to functions in the strip, which has the following discrete counterparts in our case
\begin{equation} \label{eq_discrete_integral_by_parts}
\begin{aligned}
&\sum_{k\leq m \leq k+p}(f_i(m),\phi_j(m))_{\mathbb{C}^{6N}}-(\phi_i(m),f_j(m))_{\mathbb{C}^{6N}} =\mathfrak{a}(\phi_i,\phi_j;k)-\mathfrak{a}(\phi_i,\phi_j;k+p+1)
\end{aligned}
\end{equation}
with
\begin{equation*}
f_i:=(\tilde{\mathcal{H}}_{b,\sharp}-\lambda)\phi_i \quad (i=1,2).
\end{equation*}
The remaining details are left to the interested reader.
\end{remark}

\subsection{Discrete Layer Potential Formulation of Interface Modes} \label{sec_layer_potential_framework}

In this section, we set up the basic framework for studying interface modes based on a discrete layer potential framework. As in Section \ref{sec_energy_flux}, we prefer to work with the NNH strip Hamiltonian for the ease of notation. To be precise, we will solve the following equivalent eigenvalue problem:
\begin{equation} \label{eq_blocked_interface_problem}
(\tilde{\mathcal{H}}_{zig,\sharp}-\lambda)u=0 \quad \text{with } \tilde{\mathcal{H}}_{zig,\sharp}:=\iota \mathcal{H}_{zig,\sharp} \iota^*,
\end{equation}
for $\lambda\in\mathcal{I}_\delta$ (bulk spectral gap) and $u\in \tilde{\mathcal{X}}_{\sharp}$. By \eqref{eq_interface_Hamiltonian}, the kernel of the NNH interface Hamiltonian $\tilde{\mathcal{H}}_{zig,\sharp}$ can be written explicitly as
\begin{equation*}
\tilde{\mathcal{H}}_{zig,\sharp}(n,m)=
\left\{
\begin{aligned}
&\tilde{\mathcal{H}}_{\delta,\sharp}(n,m),\quad n,m\geq 0, \\
&\tilde{\mathcal{H}}_{-\delta,\sharp}(n,m),\quad n,m< 0, \\
&(\tilde{\mathcal{H}}_{b,\sharp}+\delta\tilde{\mathcal{E}}_{\sharp})(n,m),\quad \text{otherwise},
\end{aligned}
\right.
\end{equation*}
with $\tilde{\mathcal{O}}:=\iota^*\mathcal{O} \iota$ for $\mathcal{O}\in \{\mathcal{H}_{\pm\delta,\sharp},\mathcal{H}_{b,\sharp},\mathcal{E}_{\sharp}\}$, and the subscripts indicate the restriction to $\mathcal{X}_{\sharp}$.

The first key observation is that, if $u\in \tilde{\mathcal{X}}_{\sharp}$ solves \eqref{eq_blocked_interface_problem}, one can prove that $u(n)$ admits the following layer-potential expression by the (discrete) integration by parts\footnote{Specifically, to obtain  formula \eqref{eq_interface_mode_layer_potential_form} for $n\geq 0$, one just applies \eqref{eq_discrete_integral_by_parts} taking $\phi_1(m)=u(m)$, $\phi_2(m)=\tilde{\mathcal{G}}_{\delta,\sharp}(m,n)$, $k=0$, replacing $\tilde{\mathcal{H}}_{b,\sharp}$ by $\tilde{\mathcal{H}}_{zig,\sharp}$, and using the fact that the last term vanishes as $p\to\infty$; the case where $n<0$ can be proved similarly.}
\begin{equation} \label{eq_interface_mode_layer_potential_form}
u(n)=\left\{
\begin{aligned}
&\tilde{\mathcal{G}}_{\delta,\sharp}(n,-1)\tilde{\mathcal{H}}_{\delta,\sharp}(-1,0)a-\tilde{\mathcal{G}}_{\delta,\sharp}(n,0)\tilde{\mathcal{H}}_{zig,\sharp}(0,-1)b,\quad n\geq 0, \\
&-\tilde{\mathcal{G}}_{-\delta,\sharp}(n,-1)\tilde{\mathcal{H}}_{zig,\sharp}(-1,0)a+\tilde{\mathcal{G}}_{-\delta,\sharp}(n,0)\tilde{\mathcal{H}}_{-\delta,\sharp}(0,-1)b,\quad n< 0, \\
\end{aligned}
\right.
\end{equation}
with $a=u(0),b=u(-1)\in\mathbb{C}^{6N}$, and
\begin{equation*}
\tilde{\mathcal{G}}_{\pm\delta,\sharp}=\tilde{\mathcal{G}}_{\pm\delta,\sharp}(\lambda):= (\tilde{\mathcal{H}}_{\pm\delta,\sharp}-\lambda)^{-1} = \iota^{*} (\mathcal{H}_{\pm\delta,\sharp}-\lambda)^{-1} \iota.
\end{equation*}
Throughout, we will omit the $\lambda$ dependence of $\tilde{\mathcal{G}}_{\pm\delta,\sharp}$ if there is no confusion. Moreover, evaluating \eqref{eq_blocked_interface_problem} at $n=-1,0$ with \eqref{eq_interface_mode_layer_potential_form}, one can derive the following \textit{boundary matching equation}:
\begin{equation} \label{eq_boundary_matching}
\mathcal{M}(\lambda,\delta)
\begin{pmatrix}
a \\ b
\end{pmatrix}
=0,
\end{equation}
with the Hermitian matrix $\mathcal{M}(\lambda,\delta)\in \mathbb{C}^{12N\times 12N}$ defined as
\begin{equation} \label{eq_M_matrix_elements}
\mathcal{M}(\lambda,\delta):=
\begin{pmatrix}
\substack{-\tilde{\mathcal{H}}_{\delta,\sharp}(0,-1)\tilde{\mathcal{G}}_{\delta,\sharp}(-1,-1)\tilde{\mathcal{H}}_{\delta,\sharp}(-1,0) \\ -\tilde{\mathcal{H}}_{zig,\sharp}(0,-1)\tilde{\mathcal{G}}_{-\delta,\sharp}(-1,-1)\tilde{\mathcal{H}}_{zig,\sharp}(-1,0)} &
\substack{-\tilde{\mathcal{H}}_{zig,\sharp}(0,-1)+\tilde{\mathcal{H}}_{\delta,\sharp}(0,-1)\tilde{\mathcal{G}}_{\delta,\sharp}(-1,0)\tilde{\mathcal{H}}_{zig,\sharp}(0,-1) \\ +\tilde{\mathcal{H}}_{zig,\sharp}(0,-1)\tilde{\mathcal{G}}_{-\delta,\sharp}(-1,0)\tilde{\mathcal{H}}_{-\delta,\sharp}(0,-1)} \\
\\
\substack{-\tilde{\mathcal{H}}_{zig,\sharp}(-1,0)+\tilde{\mathcal{H}}_{-\delta,\sharp}(-1,0)\tilde{\mathcal{G}}_{-\delta,\sharp}(0,-1)\tilde{\mathcal{H}}_{zig,\sharp}(-1,0) \\ +\tilde{\mathcal{H}}_{zig,\sharp}(-1,0)\tilde{\mathcal{G}}_{\delta,\sharp}(0,-1)\tilde{\mathcal{H}}_{\delta,\sharp}(-1,0)} &
\substack{-\tilde{\mathcal{H}}_{-\delta,\sharp}(-1,0)\tilde{\mathcal{G}}_{-\delta,\sharp}(0,0)\tilde{\mathcal{H}}_{-\delta,\sharp}(0,-1) \\ -\tilde{\mathcal{H}}_{zig,\sharp}(-1,0)\tilde{\mathcal{G}}_{\delta,\sharp}(0,0)\tilde{\mathcal{H}}_{zig,\sharp}(0,-1)}
\end{pmatrix}.
\end{equation}

On the other hand, restricting $n=-1,0$ in \eqref{eq_interface_mode_layer_potential_form} produces the following condition:
\begin{equation} \label{eq_boundary_matching_auxliary}
\mathcal{M}^{aux}(\lambda,\delta)
\begin{pmatrix}
a \\ b
\end{pmatrix}
=\begin{pmatrix}
a \\ b
\end{pmatrix},
\end{equation}
with
\begin{equation} \label{eq_M_aux_matrix_elements}
\mathcal{M}^{aux}(\lambda,\delta):=
\begin{pmatrix}
\tilde{\mathcal{G}}_{\delta,\sharp}(0,-1)\tilde{\mathcal{H}}_{\delta,\sharp}(-1,0) & -\tilde{\mathcal{G}}_{\delta,\sharp}(0,0)\tilde{\mathcal{H}}_{zig,\sharp}(0,-1) \\
-\tilde{\mathcal{G}}_{-\delta,\sharp}(-1,-1)\tilde{\mathcal{H}}_{zig,\sharp}(-1,0) & \tilde{\mathcal{G}}_{-\delta,\sharp}(-1,0)\tilde{\mathcal{H}}_{-\delta,\sharp}(0,-1)
\end{pmatrix}.
\end{equation}

Remarkably, whenever the characteristic value problem \eqref{eq_boundary_matching} has a nontrivial solution and the left side of \eqref{eq_boundary_matching_auxliary} is nonzero, one can check that the function $u$ constructed by the discrete layer-potential formula \eqref{eq_interface_mode_layer_potential_form} is indeed a nontrivial eigenmode that solves \eqref{eq_blocked_interface_problem}. This means that we have derived an equivalent formulation for the \textit{existence of solutions to \eqref{eq_blocked_interface_problem}}.

\begin{proposition} \label{prop_existence_equivalence}
The eigenvalue problem \eqref{eq_blocked_interface_problem} has solutions if and only if
\begin{itemize}
    \item[1)] the characteristic value problem \eqref{eq_boundary_matching} has a nontrivial solution $(\lambda,(a,b)^{\top})\in \mathcal{I}_{\delta}\times (\mathbb{C}^{6N}\times \mathbb{C}^{6N})$;
    \item[2)] the solution in 1) satisfies $\mathcal{M}^{aux}(\lambda,\delta) (a,b)^{\top}\neq 0$.
\end{itemize}
\end{proposition}

Now we consider \textit{the precise number of interface modes}. To this end, the formulation outlined in Proposition \ref{prop_existence_equivalence} is not sufficient: it is possible that two distinct solutions, say $(a_i,b_i)^{\top}$ for $i=1,2$, may lead to the same interface modes. Importantly, this redundant degree of freedom is eliminated by \eqref{eq_boundary_matching_auxliary}: if $(a_i,b_i)^{\top}$ also solves \eqref{eq_boundary_matching_auxliary}, then the interface modes constructed by \eqref{eq_interface_mode_layer_potential_form} must be linearly independent as their data at $n=-1,0$ are independent. In conclusion, we have argued that the following result holds. 

\begin{proposition} \label{prop_number_equivalence}
Suppose that the eigenvalue problem \eqref{eq_blocked_interface_problem} has solutions. Then the total number of eigenvalues of \eqref{eq_blocked_interface_problem}, counted with their multiplicities, equals the number of characteristic values of the system of equations \eqref{eq_boundary_matching} and \eqref{eq_boundary_matching_auxliary}.
\end{proposition}

\begin{remark}
Another advantage (in fact, the main reason) of introducing the NNH strip Hamiltonian is seen from the layer-potential formulation of interface modes, i.e. \eqref{eq_boundary_matching} and \eqref{eq_boundary_matching_auxliary}, in which the matrix elements of $\mathcal{M}$ and $\mathcal{M}^{aux}$ are concisely expressed as products of nearest-block-hopping elements, instead of expanding into sums of all $N$ neighbors.
\end{remark}

\subsection{Limits of Boundary Operators}

\label{sec_boundary_operator_limit}

In this section, we study the limits of the boundary operators $\mathcal{M}$ and $\mathcal{M}^{aux}$ as $\delta\to 0^+$, which constitute the most important ingredients in solving the boundary matching equations \eqref{eq_boundary_matching} and \eqref{eq_boundary_matching_auxliary}. In the sequel, we reparameterize the energy in (a complex neighborhood of) the bulk spectral gap $\mathcal{I}_{\delta}$ as
\begin{equation*}
\lambda=\lambda_*+\delta\cdot h,\quad 
h\in\mathcal{J}:=\{z\in\mathbb{C}:\, |z|<|c_*\beta_*|\},
\end{equation*}
where $c_*,\beta_*\in\mathbb{R}$ are introduced in Theorem \ref{thm_gap_open}. With this notation, the main results of this section are as follows.
\begin{proposition} \label{prop_boundary_operator_limit}
As $\delta\to 0^+$, the matrix $\mathcal{M}(\lambda_*+\delta\cdot h,\delta)\in\mathbb{C}^{12N\times 12N}$ converges in matrix norm uniformly for $h\in\mathcal{J}$:
\begin{equation} \label{eq_M_limit}
\lim_{\delta\to 0^+}\Big\|\mathcal{M}(\lambda_*+\delta\cdot h,\delta)-\big(\mathcal{M}^{p.v.}+\xi(h)\mathcal{A}\big)\Big\|_{\mathbb{C}^{12N\times 12N}}=0,
\end{equation}
where, in the limit operator, $\mathcal{M}^{p.v.}$ is defined by replacing the Green operator $\tilde{\mathcal{G}}_{\pm\delta,\sharp}$ in \eqref{eq_M_matrix_elements} with $\tilde{\mathcal{G}}_{b,\sharp}^{p.v.}:=\iota^{*} \mathcal{G}_{b,\sharp}^{p.v.}\iota$ and Hamiltonians with $\tilde{\mathcal{H}}_{b,\sharp}$:
\begin{equation} \label{eq_M_PV_expression}
\mathcal{M}^{p.v.}:=
\begin{pmatrix}
\substack{-2\tilde{\mathcal{H}}_{b,\sharp}(0,-1)\tilde{\mathcal{G}}_{b,\sharp}^{p.v.}(-1,-1)\tilde{\mathcal{H}}_{b,\sharp}(-1,0)} &
\substack{-\tilde{\mathcal{H}}_{b,\sharp}(0,-1)+2\tilde{\mathcal{H}}_{b,\sharp}(0,-1)\tilde{\mathcal{G}}_{b,\sharp}^{p.v.}(-1,0)\tilde{\mathcal{H}}_{b,\sharp}(0,-1)} \\
\\
\substack{-\tilde{\mathcal{H}}_{b,\sharp}(-1,0)+2\tilde{\mathcal{H}}_{b,\sharp}(-1,0)\tilde{\mathcal{G}}_{b,\sharp}^{p.v.}(0,-1)\tilde{\mathcal{H}}_{b,\sharp}(-1,0) } &
\substack{-2\tilde{\mathcal{H}}_{b,\sharp}(-1,0)\tilde{\mathcal{G}}_{b,\sharp}^{p.v.}(0,0)\tilde{\mathcal{H}}_{b,\sharp}(0,-1)  }
\end{pmatrix}.
\end{equation}
The function $\xi(h)$ is analytic in $h\in\mathcal{J}$ with the explicit expression
\begin{equation} \label{eq_xi_function}
\xi(h):=\frac{h}{|\alpha_*|\sqrt{\beta_*^2-h^2}},
\end{equation}
and the matrix $\mathcal{A}$ is a four-dimensional projection on $\mathbb{C}^{12N}$:
\begin{equation} \label{eq_A_projection}
\mathcal{A}:=\sum_{k=1}^{4}
\begin{pmatrix}
\substack{-\tilde{\mathcal{H}}_{b,\sharp}(0,-1)\tilde{v}_k(-1)\overline{\tilde{\mathcal{H}}_{b,\sharp}(0,-1)\tilde{v}_k(-1)}^{\top}} & \substack{\tilde{\mathcal{H}}_{b,\sharp}(0,-1)\tilde{v}_k(-1)\overline{\tilde{\mathcal{H}}_{b,\sharp}(-1,0)\tilde{v}_k(0)}^{\top}} \\
\\
\substack{\tilde{\mathcal{H}}_{b,\sharp}(-1,0)\tilde{v}_k(0)\overline{\tilde{\mathcal{H}}_{b,\sharp}(0,-1)\tilde{v}_k(-1)}^{\top}} & \substack{-\tilde{\mathcal{H}}_{b,\sharp}(-1,0)\tilde{v}_k(0)\overline{\tilde{\mathcal{H}}_{b,\sharp}(-1,0)\tilde{v}_k(0)}^{\top}}
\end{pmatrix}
\end{equation}
in which the (blocked) Floquet-Bloch eigenmodes $\tilde{v}_k$ are introduced in \eqref{eq_blocked_analytic_Bloch_mode}. Similarly, the operator $\mathcal{M}^{aux}$ admits the following limit in norm:
\begin{equation} \label{eq_M_aux_limit}
\lim_{\delta\to 0^+}\Big\|\mathcal{M}^{aux}(\lambda_*+\delta\cdot h,\delta)-\big(\mathcal{M}^{aux,p.v.}+\xi(h)\mathcal{A}^{aux}_1+\eta(h)\mathcal{A}^{aux}_2\big)\Big\|_{\mathbb{C}^{12N\times 12N}}=0,
\end{equation}
where
\begin{equation} \label{eq_M_aux_PV_expression}
\mathcal{M}^{aux,p.v.}:=
\begin{pmatrix}
\substack{\tilde{\mathcal{G}}_{b,\sharp}^{p.v.}(0,-1)\tilde{\mathcal{H}}_{b,\sharp}(-1,0)} &
\substack{-\tilde{\mathcal{G}}_{b,\sharp}^{p.v.}(0,0)\tilde{\mathcal{H}}_{b,\sharp}(0,-1)} \\
\\
\substack{-\tilde{\mathcal{G}}_{b,\sharp}^{p.v.}(-1,-1)\tilde{\mathcal{H}}_{b,\sharp}(-1,0) } &
\substack{\tilde{\mathcal{G}}_{b,\sharp}^{p.v.}(-1,0)\tilde{\mathcal{H}}_{b,\sharp}(0,-1)  }
\end{pmatrix},
\end{equation}
\begin{equation} \label{eq_A_aux_1_projection}
\mathcal{A}^{aux}_1:=\sum_{k=1}^{4}\frac{1}{2}
\begin{pmatrix}
\substack{\tilde{v}_k(0)\overline{\tilde{\mathcal{H}}_{b,\sharp}(0,-1)\tilde{v}_k(-1)}^{\top}} & \substack{-\tilde{v}_k(0)\overline{\tilde{\mathcal{H}}_{b,\sharp}(-1,0)\tilde{v}_k(0)}^{\top}} \\
\\
\substack{-\tilde{v}_k(-1)\overline{\tilde{\mathcal{H}}_{b,\sharp}(0,-1)\tilde{v}_k(-1)}^{\top}} & \substack{\tilde{v}_k(-1)\overline{\tilde{\mathcal{H}}_{b,\sharp}(-1,0)\tilde{v}_k(0)}^{\top}}
\end{pmatrix},
\end{equation}
\begin{equation} \label{eq_A_aux_2_projection}
\mathcal{A}^{aux}_2:=\sum_{k=1}^{4}\frac{s(k)}{2}
\begin{pmatrix}
\substack{\tilde{v}_k(0)\overline{\tilde{\mathcal{H}}_{b,\sharp}(0,-1)\tilde{v}_{p(k)}(-1)}^{\top}} & \substack{-\tilde{v}_k(0)\overline{\tilde{\mathcal{H}}_{b,\sharp}(-1,0)\tilde{v}_{p(k)}(0)}^{\top}} \\
\\
\substack{\tilde{v}_k(-1)\overline{\tilde{\mathcal{H}}_{b,\sharp}(0,-1)\tilde{v}_{p(k)}(-1)}^{\top}} & \substack{-\tilde{v}_k(-1)\overline{\tilde{\mathcal{H}}_{b,\sharp}(-1,0)\tilde{v}_{p(k)}(0)}^{\top}}
\end{pmatrix},
\end{equation}
and
\begin{equation} \label{eq_eta_function}
\eta(h):=\frac{\beta_*}{|\alpha_*|\sqrt{\beta_*^2-h^2}}.
\end{equation}
Here, in \eqref{eq_A_aux_2_projection}, $s(k)$ denotes the parity of integers
\begin{equation*}
s(k)=1 \quad \text{for even $k$},\quad s(k)=-1 \quad \text{for odd $k$},
\end{equation*}
and $p:=(13)(24)$ is a permutation on the finite set $\{1,2,3,4\}$.
\end{proposition}

\begin{remark}
As will be seen later, the limit operator in \eqref{eq_M_limit} is of ``Fredholm type", in the sense that $\mathcal{M}^{p.v.}$ is pseudo-invertible with a kernel space of \textit{$N-$independent dimension}, while the dimension of the projection $\mathcal{A}$ is also $N-$independent. We point out that in the previous papers \cite{qiu2024square_lattice,qiu2026waveguide_localized}, where the boundary space is infinite-dimensional, the associated boundary operator is actually Fredholm in the usual sense.
\end{remark}

\begin{remark} \label{rmk_absence_band_inversion}
We remark on the scenario when a band inversion is absent across the interface. Consider the case where $\mathcal{H}_{-\delta}$ is replaced by $\mathcal{H}_{\delta}$ in \eqref{eq_interface_Hamiltonian}, i.e. the governing operators for bulk media across the interface are the same. In that case, even though there is still a common band gap between the bulk Hamiltonians, there is no band inversion at the end points of the gap. As a consequence, we will see a dramatic change in the limiting boundary matching operator: one replaces $-\delta \to \delta$ in \eqref{eq_M_matrix_elements} and calculates its limit as in Proposition \ref{prop_boundary_operator_limit}, the result of which will be
\begin{equation} \label{eq_limiting_operator_absence_band_inversion}
\mathcal{M}^{p.v.}+\xi(h)\mathcal{A}+2\eta(h)
\begin{pmatrix}
-\tilde{\mathcal{H}}_{b,\sharp}(0,-1) & 0 \\ 0 & \tilde{\mathcal{H}}_{b,\sharp}(-1,0)
\end{pmatrix}
\mathcal{A}^{aux}_2 .
\end{equation}
Compared to \eqref{eq_M_limit}, there is an additional projector proportional to $\eta(h)$. Following the calculation in Section \ref{sec_existence_interface_mode}, one can prove that the operator \eqref{eq_limiting_operator_absence_band_inversion} is invertible for all $h\in\mathcal{J}$, in sharp contrast to the limit operator in \eqref{eq_M_limit} which has four characteristic values. With this fact, one can directly prove the absence of interface modes for the interface model lacking a band inversion across the interface, which perfectly matches the physical intuition.
\end{remark}

Before showing the proof of Proposition \ref{prop_boundary_operator_limit}, which is based on an asymptotic analysis of the Green operator $\tilde{\mathcal{G}}_{\pm\delta,\sharp}$, we present the following property of the dominant operator $\mathcal{M}^{p.v.}$ in the limit \eqref{eq_M_limit}. It is critical to solve the boundary matching equation \eqref{eq_boundary_matching} to find the interface modes.

\begin{proposition} \label{prop_M_PV_kernel}
The matrix $\mathcal{M}^{p.v.}\in\mathbb{C}^{12N\times 12N}$ defined in \eqref{eq_M_PV_expression} is Hermitian. Moreover, the null space of $\mathcal{M}^{p.v.}$ is spanned by the boundary data of the (blocked) Floquet-Bloch eigenmodes, i.e.,
\begin{equation} \label{eq_M_PV_kernel}
\ker \mathcal{M}^{p.v.}=\text{span}\big\{ \big(\tilde{v}_{k}(0),\tilde{v}_{k}(-1) \big)^{\top},\,1\leq k\leq 4 \big\}.
\end{equation}
\end{proposition}

The remainder of this section is devoted to the proof of Propositions \ref{prop_boundary_operator_limit} and \ref{prop_M_PV_kernel}.

\begin{proof}[Proof of Proposition \ref{prop_boundary_operator_limit}]
We claim that the convergences \eqref{eq_M_limit} and \eqref{eq_M_aux_limit} follow from the following convergence of the Green operator:
\begin{equation} \label{eq_green_operator_convergence}
\begin{aligned}
&\lim_{\delta\to 0^+}\tilde{\mathcal{G}}_{\pm\delta,\sharp}(\lambda_*+\delta\cdot h)(n,m) \\
&=\tilde{\mathcal{G}}_{b,\sharp}^{p.v.}(n,m)+\frac{1}{2}\xi(h)\sum_{1\leq k\leq 4}\tilde{v}_{k}(n)\overline{\tilde{v}_{k}(m)}^{\top}
\pm \frac{1}{2}\eta(h)\sum_{1\leq k\leq 4}s(k)\tilde{v}_{k}(n)\overline{\tilde{v}_{p(k)}(m)}^{\top}
\end{aligned}
\end{equation}
for any $n,m\in\mathbb{Z}$. In fact, with \eqref{eq_green_operator_convergence}, one can check the convergence of $\mathcal{M}(\lambda_*+\delta\cdot h,\delta)$ and $\mathcal{M}^{aux}(\lambda_*+\delta\cdot h,\delta)$ elementwisely. For example, considering the upper-left element of $\mathcal{M}(\lambda_*+\delta\cdot h,\delta)$, one calculates that
\footnotesize
\begin{equation*}
\begin{aligned}
&\lim_{\delta\to 0^+}\Big(-\tilde{\mathcal{H}}_{\delta,\sharp}(0,-1)\tilde{\mathcal{G}}_{\delta,\sharp}(-1,-1)\tilde{\mathcal{H}}_{\delta,\sharp}(-1,0) -\tilde{\mathcal{H}}_{zig,\sharp}(0,-1)\tilde{\mathcal{G}}_{-\delta,\sharp}(-1,-1)\tilde{\mathcal{H}}_{zig,\sharp}(-1,0) \Big)\\
&\overset{(i)}{=} -\tilde{\mathcal{H}}_{b,\sharp}(0,-1)\Big(\tilde{\mathcal{G}}_{b,\sharp}^{p.v.}(-1,-1)+\frac{1}{2}\xi(h)\sum_{k=1}\tilde{v}_{k}(-1)\overline{\tilde{v}_{k}(-1)}^{\top}
+ \frac{1}{2}\eta(h)\sum_{k=1}s(k)\tilde{v}_{k}(-1)\overline{\tilde{v}_{p(k)}(-1)}^{\top}\Big)\tilde{\mathcal{H}}_{b,\sharp}(-1,0) \\
&-\tilde{\mathcal{H}}_{b,\sharp}(0,-1)\Big(\tilde{\mathcal{G}}_{b,\sharp}^{p.v.}(-1,-1)+\frac{1}{2}\xi(h)\sum_{k=1}\tilde{v}_{k}(-1)\overline{\tilde{v}_{k}(-1)}^{\top}
- \frac{1}{2}\eta(h)\sum_{k=1}s(k)\tilde{v}_{k}(-1)\overline{\tilde{v}_{p(k)}(-1)}^{\top} \Big)\tilde{\mathcal{H}}_{b,\sharp}(-1,0) \\
&\overset{(ii)}{=}-2\tilde{\mathcal{H}}_{b,\sharp}(0,-1)\tilde{\mathcal{G}}_{b,\sharp}^{p.v.}(-1,-1)\tilde{\mathcal{H}}_{b,\sharp}(-1,0)-\xi(h)\sum_{k=1}^{4}\tilde{\mathcal{H}}_{b,\sharp}(0,-1)\tilde{v}_{k}(-1)\overline{\tilde{\mathcal{H}}_{b,\sharp}(0,-1)\tilde{v}_{k}(-1)} \\
&=\mathcal{M}^{p.v.}_{11}+\xi(h)\mathcal{A}_{11},
\end{aligned}
\end{equation*}
\normalsize
where $(i)$ follows from \eqref{eq_green_operator_convergence}, and in $(ii)$ we have applied the Hermiticity property $$\tilde{\mathcal{H}}_{b,\sharp}(n,m)=\overline{\tilde{\mathcal{H}}_{b,\sharp}(m,n)}^{\top}.$$ The convergence of the other elements of the matrix is verified in a similar way.

{\color{blue}Step 1.} We now prove \eqref{eq_green_operator_convergence} for $\tilde{\mathcal{G}}_{\delta,\sharp}$, while the case of $\tilde{\mathcal{G}}_{-\delta,\sharp}$ is treated similarly. To do that, we first split the Green operator into several parts: this is achieved by first expanding \eqref{eq_green_operator_convergence} as a Fourier series (analogous to \eqref{eq_sec_l.a._2}), then decomposing the Fourier expansion according to distinct integral domains. Specifically, we write
\begin{equation} \label{eq_boundary_operator_limit_proof_1}
\tilde{\mathcal{G}}_{\delta,\sharp}(\lambda_*+\delta\cdot h)(n,m)=\frac{1}{2\pi}\int_{-\pi}^{\pi}d\kappa_1 \sum_{1\leq k\leq 6}\frac{\tilde{v}_{k,\delta,\sharp}(n;\kappa_1)\overline{\tilde{v}_{k,\delta,\sharp}(m;\kappa_1)}^{\top}}{\mu_{k,\delta,\sharp}(\kappa_1)-(\lambda_*+\delta\cdot h)},
\end{equation}
where $\mu_{k,\delta,\sharp}(\kappa_1)$ is the analytically labeled Floquet-Bloch eigenvalue of $\tilde{\mathcal{H}}_{\delta,\sharp}$, and $\tilde{v}_{k,\delta,\sharp}(\cdot;\kappa_1)=\iota^{*} v_{k,\delta}(\cdot;\kappa_1\bm{\ell}_1^*)$ are the associated (blocked) eigenmodes; note that for $1\leq n\leq 4$ and small $\kappa_1$, these eigenpairs admit a detailed asymptotic expansions obtained from Theorem \ref{thm_asymptotics_eigenpairs}. Next, we decompose \eqref{eq_boundary_operator_limit_proof_1} as follows:
\begin{equation} \label{eq_boundary_operator_limit_proof_2}
\begin{aligned}
&\tilde{\mathcal{G}}_{\delta,\sharp}(\lambda_*+\delta\cdot h)(n,m) \\
&=\underbrace{\frac{1}{2\pi}\int_{-\pi}^{\pi}d\kappa_1 \sum_{k=5,6}\frac{\tilde{v}_{k,\delta,\sharp}(n;\kappa_1)\overline{\tilde{v}_{k,\delta,\sharp}(m;\kappa_1)}^{\top}}{\mu_{k,\delta,\sharp}(\kappa_1)-(\lambda_*+\delta\cdot h)} + \frac{1}{2\pi}\int_{(-\pi,-\delta^{1/3})\cup (\delta^{1/3},\pi)}d\kappa_1 \sum_{1\leq k\leq 4}\frac{\tilde{v}_{k,\delta,\sharp}(n;\kappa_1)\overline{\tilde{v}_{k,\delta,\sharp}(m;\kappa_1)}^{\top}}{\mu_{k,\delta,\sharp}(\kappa_1)-(\lambda_*+\delta\cdot h)}}_{\text{regular part}} \\
&\quad + \underbrace{\frac{1}{2\pi}\int_{-\delta^{1/3}}^{\delta^{1/3}}d\kappa_1 \sum_{1\leq k\leq 4}\frac{\tilde{v}_{k,\delta,\sharp}(n;\kappa_1)\overline{\tilde{v}_{k,\delta,\sharp}(m;\kappa_1)}^{\top}}{\mu_{k,\delta,\sharp}(\kappa_1)-(\lambda_*+\delta\cdot h)}}_{\text{singular part}} \\
&=:\tilde{\mathcal{G}}_{\delta,\sharp}^{reg}(\lambda_*+\delta\cdot h)(n,m)+\tilde{\mathcal{G}}_{\delta,\sharp}^{sing}(\lambda_*+\delta\cdot h)(n,m).
\end{aligned}
\end{equation}
As will become clearer soon, the superscript `reg' for the first two integrals indicates that their integrands behave regularly with respect to the perturbation parameter $\delta$ as the denominators are sufficiently large (at least of the order $\sim \delta^{1/3}$, which is much larger than $\sim \delta$); similarly, `sing' indicates a singular integral, the convergence of which requires a detailed asymptotic analysis using the results of Theorem \ref{thm_asymptotics_eigenpairs}.

We claim that the three integrals in \eqref{eq_boundary_operator_limit_proof_2} converge as follows:
\begin{equation} \label{eq_boundary_operator_limit_proof_3}
\lim_{\delta\to 0^+} \tilde{\mathcal{G}}_{\delta,\sharp}^{reg}(\lambda_*+\delta\cdot h)(n,m) =\tilde{\mathcal{G}}_{b,\sharp}^{p.v.}(n,m)
\end{equation}
and
\begin{equation} \label{eq_boundary_operator_limit_proof_4}
\lim_{\delta\to 0^+} \tilde{\mathcal{G}}_{\delta,\sharp}^{sing}(\lambda_*+\delta\cdot h)(n,m) =\frac{1}{2}\xi(h)\sum_{k=1}\tilde{v}_{k}(n)\overline{\tilde{v}_{k}(m)}^{\top}
+ \frac{1}{2}\eta(h)\sum_{k=1}s(k)\tilde{v}_{k}(n)\overline{\tilde{v}_{p(k)}(m)}^{\top}.
\end{equation}
Then, the convergence \eqref{eq_green_operator_convergence} follows directly. In the sequel, we prove \eqref{eq_boundary_operator_limit_proof_3} and \eqref{eq_boundary_operator_limit_proof_4}, respectively.

{\color{blue}Step 2.} The proof of \eqref{eq_boundary_operator_limit_proof_3} is based on a standard application of regular perturbation theory, following the lines of our previous works \cite[Appendix F]{qiu2026waveguide_localized} and \cite[Lemma 5.5]{qiu2024square_lattice}; here, we only outline the main idea of the proof and refer the reader to the aforementioned literature for details. We first note the following estimate:
\begin{equation} \label{eq_boundary_operator_limit_proof_5}
\int_{-\pi}^{\pi}d\kappa_1 \sum_{k=5,6}\frac{\tilde{v}_{k,\delta,\sharp}(n;\kappa_1)\overline{\tilde{v}_{k,\delta,\sharp}(m;\kappa_1)}^{\top}}{\mu_{k,\delta,\sharp}(\kappa_1)-(\lambda_*+\delta\cdot h)}
=\int_{-\pi}^{\pi}d\kappa_1 \sum_{k=5,6}\frac{\tilde{v}_{k,\sharp}(n;\kappa_1)\overline{\tilde{v}_{k,\sharp}(m;\kappa_1)}^{\top}}{\mu_{k,\sharp}(\kappa_1)-\lambda_*}+\mathcal{O}(\delta),
\end{equation}
where $\mu_{k,\sharp}(\kappa_1)=\mu_{k}(\kappa_1\bm{\ell}_1^*)$ and $\tilde{v}_{k,\sharp}(\cdot;\kappa_1)=\iota^{*} v_{k}(\cdot;\kappa_1\bm{\ell}_1^*)$ are the (analytically labeled blocked) unperturbed Floquet-Bloch eigenpairs. In fact, since the fifth and sixth bands are uniformly bounded away from $\lambda_*$ (with a $\delta-$independent isolation distance) due to assumptions \eqref{eq_no_fold} and \eqref{eq_isolated_bands}, the denominators of \eqref{eq_boundary_operator_limit_proof_5} are uniformly bounded, i.e., there exists $c>0$ such that
\begin{equation} \label{eq_boundary_operator_limit_proof_6}
\big|\mu_{k,\delta,\sharp}(\kappa_1)-(\lambda_*+\delta\cdot h)\big|>c>0 \quad \forall\, k=5,6,\,\kappa_1\in (-\pi,\pi),\, h\in\mathcal{J} \text{ and small $\delta$.}
\end{equation}
On the other hand, standard perturbation theory gives 
\begin{equation} \label{eq_boundary_operator_limit_proof_7}
\big|\mu_{k,\delta,\sharp}(\kappa_1)+\delta\cdot h -\mu_{k,\sharp}(\kappa_1)\big|=\mathcal{O}(\delta),\quad \big\|\tilde{v}_{k,\delta,\sharp}(n;\kappa_1)-\tilde{v}_{k,\sharp}(n;\kappa_1)\big\|_{\mathbb{C}^{6N}}=\mathcal{O}(\delta).
\end{equation}
With \eqref{eq_boundary_operator_limit_proof_6} and \eqref{eq_boundary_operator_limit_proof_7}, the estimate \eqref{eq_boundary_operator_limit_proof_5} follows directly.

Now, we prove the convergence of the other part of \eqref{eq_boundary_operator_limit_proof_3}. The first step is to obtain the following estimate:
\begin{equation} \label{eq_boundary_operator_limit_proof_8}
\begin{aligned}
&\int_{(-\pi,-\delta^{1/3})\cup (\delta^{1/3},\pi)}d\kappa_1 \sum_{1\leq k\leq 4}\frac{\tilde{v}_{k,\delta,\sharp}(n;\kappa_1)\overline{\tilde{v}_{k,\delta,\sharp}(m;\kappa_1)}^{\top}}{\mu_{k,\delta,\sharp}(\kappa_1)-(\lambda_*+\delta\cdot h)} \\
&=\int_{(-\pi,-\delta^{1/3})\cup (\delta^{1/3},\pi)}d\kappa_1 \sum_{1\leq k\leq 4}\frac{\tilde{v}_{k,\sharp}(n;\kappa_1)\overline{\tilde{v}_{k,\sharp}(m;\kappa_1)}^{\top}}{\mu_{k,\sharp}(\kappa_1)-\lambda_*}+\mathcal{O}(\delta^{1/3}).
\end{aligned}
\end{equation}
The idea is similar to that for the proof of \eqref{eq_boundary_operator_limit_proof_5}, except for the fact that the uniform lower bound \eqref{eq_boundary_operator_limit_proof_6} of the denominator no longer holds. Nevertheless, thanks to our design of the termination point $\delta^{1/3}$ of the integral domain, Theorems \ref{thm_asymptotics_eigenpairs} and \ref{thm_double_cone} imply that we can still control the denominator. In fact, there exists $c>0$ such that
\begin{equation} \label{eq_boundary_operator_limit_proof_9}
\big|\mu_{k,\delta,\sharp}(\kappa_1)-(\lambda_*+\delta\cdot h)\big|>c\delta^{1/3} \quad \forall\, k=1,2,3,4,\,|\kappa_1|>\delta^{1/3},\, h\in\mathcal{J} \text{ and small $\delta$.}
\end{equation}
With \eqref{eq_boundary_operator_limit_proof_9}, and also \eqref{eq_boundary_operator_limit_proof_7}, one obtains the estimate \eqref{eq_boundary_operator_limit_proof_8}. Next, \textit{since $\mu_{k,\sharp}(\kappa_1)$ possess non-vanishing derivatives at $\kappa_1=0$} (thanks to the Dirac cone and analytic labeling), we can push $\delta\to 0$ in \eqref{eq_boundary_operator_limit_proof_8} for the integral on the right side and obtain a principal-value integral
\begin{equation*}
\lim_{\delta\to 0^{+}}\int_{(-\pi,-\delta^{1/3})\cup (\delta^{1/3},\pi)}d\kappa_1 \sum_{1\leq k\leq 4}\frac{\tilde{v}_{k,\sharp}(n;\kappa_1)\overline{\tilde{v}_{k,\sharp}(m;\kappa_1)}^{\top}}{\mu_{k,\sharp}(\kappa_1)-\lambda_*}={p.v.}\int_{-\pi}^{\pi}d\kappa_1 \sum_{1\leq k\leq 4}\frac{\tilde{v}_{k,\sharp}(n;\kappa_1)\overline{\tilde{v}_{k,\sharp}(m;\kappa_1)}^{\top}}{\mu_{k,\sharp}(\kappa_1)-\lambda_*}.
\end{equation*}
Then, with \eqref{eq_boundary_operator_limit_proof_8} and \eqref{eq_boundary_operator_limit_proof_5}, we conclude the proof of \eqref{eq_boundary_operator_limit_proof_3}.

{\color{blue}Step 2.} Now, we prove the convergence of the singular part \eqref{eq_boundary_operator_limit_proof_4} by a careful asymptotic analysis. We first note that since the first four bands constitute all the bands near the double Dirac cone, the sum of integrals defining \eqref{eq_boundary_operator_limit_proof_4} (see \eqref{eq_boundary_operator_limit_proof_2}) can be alternatively written in terms of the increasingly labeled eigenpairs:
\begin{equation} \label{eq_boundary_operator_limit_proof_10}
\int_{-\delta^{1/3}}^{\delta^{1/3}}d\kappa_1 \sum_{1\leq k\leq 4}\frac{\tilde{v}_{k,\delta,\sharp}(n;\kappa_1)\overline{\tilde{v}_{k,\delta,\sharp}(m;\kappa_1)}^{\top}}{\mu_{k,\delta,\sharp}(\kappa_1)-(\lambda_*+\delta\cdot h)}
=\int_{-\delta^{1/3}}^{\delta^{1/3}}d\kappa_1 \sum_{1\leq k\leq 4}\frac{\tilde{u}_{k,\delta,\sharp}(n;\kappa_1)\overline{\tilde{u}_{k,\delta,\sharp}(m;\kappa_1)}^{\top}}{\lambda_{k,\delta,\sharp}(\kappa_1)-(\lambda_*+\delta\cdot h)}.
\end{equation}
Next, we apply the results of Theorem \ref{thm_asymptotics_eigenpairs} to estimate the right side. In light of the asymptotics of eigenpairs obtained in Theorem \ref{thm_asymptotics_eigenpairs}, we decompose \eqref{eq_boundary_operator_limit_proof_10} according to the leading-order and the remaining contributions of each band, that is,
\begin{equation} \label{eq_boundary_operator_limit_proof_15}
\begin{aligned}
\frac{1}{2\pi}\int_{-\delta^{1/3}}^{\delta^{1/3}}d\kappa_1 \sum_{1\leq k\leq 4}\frac{\tilde{u}_{k,\delta,\sharp}(n;\kappa_1)\overline{\tilde{u}_{k,\delta,\sharp}(m;\kappa_1)}^{\top}}{\lambda_{k,\delta,\sharp}(\kappa_1)-(\lambda_*+\delta\cdot h)}
=\sum_{1\leq k\leq 4}\Big( \mathcal{K}_{k,\delta}^{(0)}(n,m) + \mathcal{K}_{k,\delta}^{(1)}(n,m) \Big).
\end{aligned}
\end{equation}
Here, $\mathcal{K}_{k,\delta}^{(0)}$ includes the leading-order contribution of the $k-$th band:
\begin{equation} \label{eq_boundary_operator_limit_proof_11}
\mathcal{K}_{1,\delta}^{(0)}(n,m)
:=\frac{1}{2\pi} \int_{-\delta^{1/3}}^{\delta^{1/3}}\frac{d\kappa_1}{1+f^2(\kappa_1,\delta)} \frac{\big(-f(\kappa_1,\delta)\tilde{u}_1(n)+\tilde{u}_3(n)
\big)\overline{\big(-f(\kappa_1,\delta)\tilde{u}_1(m)+\tilde{u}_3(m)
\big)}^{\top}}{-\sqrt{\beta_{*}^2\delta^2+|\alpha_*|^2\kappa_1^2}-\delta\cdot h},
\end{equation}
\begin{equation} \label{eq_boundary_operator_limit_proof_12}
\mathcal{K}_{2,\delta}^{(0)}(n,m)
:=\frac{1}{2\pi} \int_{-\delta^{1/3}}^{\delta^{1/3}}\frac{d\kappa_1}{1+f^2(\kappa_1,\delta)} \frac{\big(-f(\kappa_1,\delta)\tilde{u}_2(n)+\tilde{u}_4(n)
\big)\overline{\big(-f(\kappa_1,\delta)\tilde{u}_2(m)+\tilde{u}_4(m)
\big)}^{\top}}{-\sqrt{\beta_{*}^2\delta^2+|\alpha_*|^2\kappa_1^2}-\delta\cdot h},
\end{equation}
\begin{equation} \label{eq_boundary_operator_limit_proof_13}
\mathcal{K}_{3,\delta}^{(0)}(n,m)
:=\frac{1}{2\pi} \int_{-\delta^{1/3}}^{\delta^{1/3}}\frac{d\kappa_1}{1+f^2(\kappa_1,\delta)} \frac{\big(\tilde{u}_2(n)+f(\kappa_1,\delta)\tilde{u}_4(n)
\big)\overline{\big(\tilde{u}_2(m)+f(\kappa_1,\delta)\tilde{u}_4(m)
\big)}^{\top}}{\sqrt{\beta_{*}^2\delta^2+|\alpha_*|^2\kappa_1^2}-\delta\cdot h},
\end{equation}
and
\begin{equation} \label{eq_boundary_operator_limit_proof_14}
\mathcal{K}_{4,\delta}^{(0)}(n,m)
:=\frac{1}{2\pi} \int_{-\delta^{1/3}}^{\delta^{1/3}}\frac{d\kappa_1}{1+f^2(\kappa_1,\delta)} \frac{\big(\tilde{u}_1(n)+f(\kappa_1,\delta)\tilde{u}_3(n)
\big)\overline{\big(\tilde{u}_1(m)+f(\kappa_1,\delta)\tilde{u}_3(m)
\big)}^{\top}}{\sqrt{\beta_{*}^2\delta^2+|\alpha_*|^2\kappa_1^2}-\delta\cdot h},
\end{equation}
where the function $f(\kappa_1,\delta)$ is defined as
\begin{equation*}
f(\kappa_1,\delta):=\frac{\alpha_*\kappa_1}{\beta_*\delta+\sqrt{\beta_{*}^2\delta^2+|\alpha_*|^2\kappa_1^2}}
\end{equation*}
and $\tilde{u}_k:=\iota^* u_k$ with $u_k$ being the short-hand denotation of Floquet-Bloch modes at the double Dirac point, as in Theorem \ref{thm_asymptotics_eigenpairs}. The matrix $\mathcal{K}_{k,\delta}^{(1)}$ includes the remaining contribution of the $k-$th band:
\begin{equation*}
\mathcal{K}_{k,\delta}^{(1)}(n,m):=\frac{1}{2\pi}\int_{-\delta^{1/3}}^{\delta^{1/3}}d\kappa_1 \sum_{1\leq k\leq 4}\frac{\tilde{u}_{k,\delta,\sharp}(n;\kappa_1)\overline{\tilde{u}_{k,\delta,\sharp}(m;\kappa_1)}^{\top}}{\lambda_{k,\delta,\sharp}(\kappa_1)-(\lambda_*+\delta\cdot h)} - \mathcal{K}_{k,\delta}^{(0)}(n,m).
\end{equation*}
By the estimate of eigenpairs shown in Theorem \ref{thm_asymptotics_eigenpairs}, one can check that
\begin{equation} \label{eq_boundary_operator_limit_proof_16}
\big\|\mathcal{K}_{k,\delta}^{(1)}(n,m)\big\|_{\mathbb{C}^{6N\times 6N}}\leq C\delta^{\frac{1}{3}}\big\|\mathcal{K}_{k,\delta}^{(0)}(n,m)\big\|_{\mathbb{C}^{6N\times 6N}} 
\end{equation}
for some $C>0$. 

With these preparations, we will calculate the limits of $\sum_{1\leq k\leq 4}\mathcal{K}_{k,\delta}^{(0)}(n,m)$ and $\sum_{1\leq k\leq 4}\mathcal{K}_{k,\delta}^{(1)}(n,m)$ in two separate steps, which then complete the proof of \eqref{eq_boundary_operator_limit_proof_4} recalling the decomposition \eqref{eq_boundary_operator_limit_proof_15}.

{\color{blue}Step 2.1.} We claim that
\begin{equation} \label{eq_boundary_operator_limit_proof_17}
\begin{aligned}
&\lim_{\delta\to 0^{+}}\Big(\mathcal{K}_{1,\delta}^{(0)}(n,m) + \mathcal{K}_{4,\delta}^{(0)}(n,m)\Big) \\
&=\frac{\xi(h)}{2}\big(\tilde{u}_1(n)\overline{\tilde{u}_1(m)}^{\top}+\tilde{u}_3(n)\overline{\tilde{u}_3(m)}^{\top}\big) + \frac{\eta(h)}{2}\big(\tilde{u}_1(n)\overline{\tilde{u}_1(m)}^{\top}-\tilde{u}_3(n)\overline{\tilde{u}_3(m)}^{\top}\big),
\end{aligned}
\end{equation}
and
\begin{equation} \label{eq_boundary_operator_limit_proof_18}
\begin{aligned}
&\lim_{\delta\to 0^{+}}\Big(\mathcal{K}_{2,\delta}^{(0)}(n,m) + \mathcal{K}_{3,\delta}^{(0)}(n,m)\Big) \\
&=\frac{\xi(h)}{2}\big(\tilde{u}_2(n)\overline{\tilde{u}_2(m)}^{\top}+\tilde{u}_4(n)\overline{\tilde{u}_4(m)}^{\top}\big) + \frac{\eta(h)}{2}\big(\tilde{u}_2(n)\overline{\tilde{u}_2(m)}^{\top}-\tilde{u}_4(n)\overline{\tilde{u}_4(m)}^{\top}\big).
\end{aligned}
\end{equation}
Then, by replacing $\tilde{u}_k$ with $\tilde{v}_k$ following the rules \eqref{eq_relation_un_vn} in Remark \ref{rmk_construction_analytic_eigenfunction}, we conclude that
\begin{equation} \label{eq_boundary_operator_limit_proof_19}
\begin{aligned}
&\lim_{\delta\to 0^{+}}\sum_{1\leq k\leq 4}\mathcal{K}_{k,\delta}^{(0)}(n,m) \\
&=\frac{\xi(h)}{2}\big(\tilde{u}_1(n)\overline{\tilde{u}_1(m)}^{\top}+\tilde{u}_2(n)\overline{\tilde{u}_2(m)}^{\top}+\tilde{u}_3(n)\overline{\tilde{u}_3(m)}^{\top}+ \tilde{u}_4(n)\overline{\tilde{u}_4(m)}^{\top}\big) \\
&\quad + \frac{\eta(h)}{2}\big(\tilde{u}_1(n)\overline{\tilde{u}_1(m)}^{\top}+\tilde{u}_2(n)\overline{\tilde{u}_2(m)}^{\top}-\tilde{u}_3(n)\overline{\tilde{u}_3(m)}^{\top}- \tilde{u}_4(n)\overline{\tilde{u}_4(m)}^{\top}\big) \\
&=\frac{\xi(h)}{2}\sum_{1\leq k\leq 4}\tilde{v}_{k}(n)\overline{\tilde{v}_{k}(m)}^{\top}
\pm \frac{\eta(h)}{2}\sum_{1\leq k\leq 4}s(k)\tilde{v}_{k}(n)\overline{\tilde{v}_{p(k)}(m)}^{\top},
\end{aligned}
\end{equation}
which exactly corresponds to the right side of \eqref{eq_boundary_operator_limit_proof_4}.

Now we prove \eqref{eq_boundary_operator_limit_proof_17}, while the proof of \eqref{eq_boundary_operator_limit_proof_18} is similar. Note that $f(\kappa_1,\delta)$ is an odd function of $\kappa_1$, whose integral on $(-\delta^{1/3},\delta^{1/3})$ vanishes. Hence,
\footnotesize
\begin{equation*}
\begin{aligned}
&\mathcal{K}_{1,\delta}^{(0)}(n,m) + \mathcal{K}_{4,\delta}^{(0)}(n,m) \\
&=-\frac{1}{2\pi}\int_{-\delta^{\frac{1}{3}}}^{\delta^{\frac{1}{3}}}\frac{d\kappa_1}{1+f^2(\kappa_1,\delta)}
\Big[\frac{f^2(\kappa_1,\delta)\tilde{u}_1(n)\overline{\tilde{u}_1(n)}^{\top}+\tilde{u}_3(n)\overline{\tilde{u}_3(n)}^{\top}}{\delta\cdot h+\sqrt{\beta_{*}^2\delta^2+|\alpha_*|^2\kappa_1^2}} + \frac{\tilde{u}_1(n)\overline{\tilde{u}_1(m)}^{\top}+f^2(\kappa_1,\delta)\tilde{u}_3(n)\overline{\tilde{u}_3(m)}^{\top}}{\delta\cdot h-\sqrt{\beta_{*}^2\delta^2+|\alpha_*|^2\kappa_1^2}}  \Big] .
\end{aligned}
\end{equation*}
\normalsize
Using the following identity:
\begin{equation*}
\frac{1-f^2(\kappa_1,\delta)}{1+f^2(\kappa_1,\delta)}=\frac{\beta_*\delta}{\sqrt{\beta_{*}^2\delta^2+|\alpha_*|^2\kappa_1^2}},
\end{equation*}
one then obtains
\begin{equation} \label{eq_boundary_operator_limit_proof_20}
\begin{aligned}
&\mathcal{K}_{1,\delta}^{(0)}(n,m) + \mathcal{K}_{4,\delta}^{(0)}(n,m) \\
&=\Big[\frac{h}{2\pi}\int_{-\delta^{\frac{1}{3}}}^{\delta^{\frac{1}{3}}}
\frac{\delta}{\big(\beta_{*}^2-h^2\big)\delta^2+|\alpha_*|^2\kappa_{1}^2}d\kappa_1
\Big]\big(\tilde{u}_1(n)\overline{\tilde{u}_1(m)}^{\top}+\tilde{u}_3(n)\overline{\tilde{u}_3(m)}^{\top}\big) \\
&\quad + \Big[\frac{\beta_*}{2\pi}\int_{-\delta^{\frac{1}{3}}}^{\delta^{\frac{1}{3}}}
\frac{\delta}{\big(\beta_{*}^2-h^2\big)\delta^2+|\alpha_*|^2\kappa_{1}^2}d\kappa_1
\Big]\big(\tilde{u}_1(n)\overline{\tilde{u}_1(m)}^{\top}-\tilde{u}_3(n)\overline{\tilde{u}_3(m)}^{\top}\big).
\end{aligned}
\end{equation}
The integral in \eqref{eq_boundary_operator_limit_proof_20} can be calculated explicitly using the anti-trigonometric function. We give directly the result
\begin{equation*}
\lim_{\delta\to 0^+}\int_{-\delta^{\frac{1}{3}}}^{\delta^{\frac{1}{3}}}
\frac{\delta}{\big(\beta_{*}^2-h^2\big)\delta^2+|\alpha_*|^2\kappa_{1}^2}d\kappa_1=\frac{\pi}{|\alpha_*|\sqrt{\beta_{*}^2-h^2}}.
\end{equation*}
By this identity and \eqref{eq_boundary_operator_limit_proof_20}, the proof of \eqref{eq_boundary_operator_limit_proof_17} is complete.

{\color{blue}Step 2.2.} With \eqref{eq_boundary_operator_limit_proof_19}, it is sufficient to show the following to conclude the proof of \eqref{eq_boundary_operator_limit_proof_4}:
\begin{equation} \label{eq_boundary_operator_limit_proof_21}
\lim_{\delta\to 0^+}\mathcal{K}_{k,\delta}^{(1)}(n,m)=0,\quad \text{for $1\leq k\leq 4$}.
\end{equation}
We prove it only for $k=1$; the other cases are treated following the same lines. Using \eqref{eq_boundary_operator_limit_proof_11}, the estimate \eqref{eq_boundary_operator_limit_proof_16}, and the fact that both $f(\kappa_1,\delta)$ and $\|\tilde{u}_k(n)\|_{\mathbb{C}^{6N}}$ are uniformly bounded, one sees that $\mathcal{K}_{1,\delta}^{(1)}(n,m)$ is simply estimated by the following elementary integral
\begin{equation} \label{eq_boundary_operator_limit_proof_22}
\big\|\mathcal{K}_{1,\delta}^{(1)}(n,m)\big\|_{\mathbb{C}^{6N\times 6N}}\leq C\delta^{\frac{1}{3}}\int_{-\delta^{\frac{1}{3}}}^{\delta^{\frac{1}{3}}}\frac{d\kappa_1}{\delta\cdot h+\sqrt{\beta_*^2\delta^2+|\alpha_*|^2\kappa_1^2}}
\end{equation}
for some $C>0$. For $h\in\mathcal{J}$, it holds that $|h|\leq c_*|\beta_*|$ with the constant $c_*<1$ introduced in Theorem \ref{thm_gap_open}. Thus,
\begin{equation*}
\int_{-\delta^{\frac{1}{3}}}^{\delta^{\frac{1}{3}}}\frac{d\kappa_1}{\delta\cdot h+\sqrt{\beta_*^2\delta^2+|\alpha_*|^2\kappa_1^2}}
\leq \frac{1}{1-c_*}\int_{-\delta^{\frac{1}{3}}}^{\delta^{\frac{1}{3}}}\frac{d\kappa_1}{\sqrt{\beta_*^2\delta^2+|\alpha_*|^2\kappa_1^2}}
=\mathcal{O}(\log \delta).
\end{equation*}
Then, with \eqref{eq_boundary_operator_limit_proof_22}, we conclude the proof of \eqref{eq_boundary_operator_limit_proof_21}.
\end{proof}

\begin{proof}[Proof of Proposition \ref{prop_M_PV_kernel}]

{\color{blue}Step 1.} The Hermiticity of $\mathcal{M}^{p.v.}$ follows directly from that of the Hamiltonian and the Green operator. To prove \eqref{eq_M_PV_kernel}, we start with
\begin{equation} \label{eq_M_PV_kernel_proof_1}
\mathcal{M}^{p.v.}
\begin{pmatrix}
\tilde{v}_{k}(0) \\ \tilde{v}_{k}(-1)
\end{pmatrix}
=0, \quad 1\leq k\leq 4.
\end{equation}
We will prove it only for $k=1$. In fact, applying the discrete integral-by-part formula \eqref{eq_discrete_integral_by_parts} for $\phi_1(m)=\tilde{v}_1(m)$, $\phi_2(m)=\tilde{\mathcal{G}}_{b,\sharp}^{p.v.}(m,n)$ (with $n\geq 0$), $k=0$, $p>0$, and recalling the identity \eqref{eq_physical_green_right_inverse}, we obtain
\begin{equation} \label{eq_M_PV_kernel_proof_2}
-\tilde{v}_1(n)=\mathfrak{a}\big(\tilde{v}_1(\cdot),\tilde{\mathcal{G}}_{b,\sharp}^{p.v.}(\cdot,n);0\big)-\mathfrak{a}\big(\tilde{v}_1(\cdot),\tilde{\mathcal{G}}_{b,\sharp}^{p.v.}(\cdot,n);p+1\big).
\end{equation}
By the far-field asymptotics of the physical Green operator shown in Proposition \ref{eq_physical_green_right_inverse}, we know
\begin{equation*}
\lim_{m\to\infty}\Big[ \tilde{\mathcal{G}}_{b,\sharp}^{p.v.}(m,n)-
\frac{i}{2|\alpha_*|}\big(
\tilde{v}_{1}(m)\overline{\tilde{v}_{1}(n)}^{\top}+\tilde{v}_{2}(m)\overline{\tilde{v}_{2}(n)}^{\top}-\tilde{v}_{3}(m)\overline{\tilde{v}_{3}(n)}^{\top}-\tilde{v}_{4}(m)\overline{\tilde{v}_{4}(n)}^{\top} \big)
\Big]=0.
\end{equation*}
Hence, the last term in \eqref{eq_M_PV_kernel_proof_2} tends to
\begin{equation*}
\begin{aligned}
&\lim_{p\to\infty}\mathfrak{a}\big(\tilde{v}_1(\cdot),\tilde{\mathcal{G}}_{b,\sharp}^{p.v.}(\cdot,n);p+1\big) \\
&=-\frac{i}{2|\alpha_*|} \Big[\mathfrak{a}(\tilde{v}_1,\tilde{v}_1)\tilde{v}_1(n)+\mathfrak{a}(\tilde{v}_1,\tilde{v}_2)\tilde{v}_2(n)-\mathfrak{a}(\tilde{v}_1,\tilde{v}_3)\tilde{v}_3(n)-\mathfrak{a}(\tilde{v}_1,\tilde{v}_4)\tilde{v}_4(n) \Big] 
=\frac{1}{2}\tilde{v}_1(n),
\end{aligned}
\end{equation*}
where the last step follows from the orthogonality shown in Proposition \ref{prop_sesqui_form}. Thus, \eqref{eq_M_PV_kernel_proof_2} equals
\begin{equation} \label{eq_M_PV_kernel_proof_9}
\begin{aligned}
-\frac{1}{2}\tilde{v}_1(n)&=\mathfrak{a}\big(\tilde{v}_1(\cdot),\tilde{\mathcal{G}}_{b,\sharp}^{p.v.}(\cdot,n);0\big) \\
&\overset{(i)}{=} \overline{\tilde{\mathcal{G}}_{b,\sharp}^{p.v.}(0,n)}^{\top} \tilde{\mathcal{H}}_{b,\sharp}(1,0)\tilde{v}_1(-1)-\overline{\tilde{\mathcal{G}}_{b,\sharp}^{p.v.}(-1,n)}^{\top} \tilde{\mathcal{H}}_{b,\sharp}(0,1)\tilde{v}_1(0) \\
&\overset{(ii)}{=} \tilde{\mathcal{G}}_{b,\sharp}^{p.v.}(n,0) \tilde{\mathcal{H}}_{b,\sharp}(1,0)\tilde{v}_1(-1)-\tilde{\mathcal{G}}_{b,\sharp}^{p.v.}(n,-1) \tilde{\mathcal{H}}_{b,\sharp}(0,1)\tilde{v}_1(0),
\end{aligned}
\end{equation}
where (i) follows from the definition of the form $\mathfrak{a}$, and (ii) is a consequence of the Hermiticity $$\tilde{\mathcal{G}}_{b,\sharp}^{p.v.}(n,m)=\overline{\tilde{\mathcal{G}}_{b,\sharp}^{p.v.}(m,n)}^{\top}.$$ Taking $n=0$ and multiplying both sides by $\tilde{\mathcal{H}}_{b,\sharp}(-1,0)$, we obtain
$$
\begin{array}{lll}
\ds -\tilde{\mathcal{H}}_{b,\sharp}(-1,0)\tilde{v}_1(0) &=& \ds
-2\tilde{\mathcal{H}}_{b,\sharp}(-1,0)\tilde{\mathcal{G}}_{b,\sharp}^{p.v.}(0,-1) \tilde{\mathcal{H}}_{b,\sharp}(0,1)\tilde{v}_1(0)\\
\nm
&& \ds +2\tilde{\mathcal{H}}_{b,\sharp}(-1,0)\tilde{\mathcal{G}}_{b,\sharp}^{p.v.}(0,0) \tilde{\mathcal{H}}_{b,\sharp}(1,0)\tilde{v}_1(-1).
\end{array}
$$
Using the periodicity $\tilde{\mathcal{H}}_{b,\sharp}(n,m)=\tilde{\mathcal{H}}_{b,\sharp}(n-1,m-1)$ and writing the identity in matrix form, we get the result
\begin{equation*}
\begin{pmatrix}
\substack{-\tilde{\mathcal{H}}_{b,\sharp}(-1,0)+2\tilde{\mathcal{H}}_{b,\sharp}(-1,0)\tilde{\mathcal{G}}_{b,\sharp}^{p.v.}(0,-1)\tilde{\mathcal{H}}_{b,\sharp}(-1,0) } &
\substack{-2\tilde{\mathcal{H}}_{b,\sharp}(-1,0)\tilde{\mathcal{G}}_{b,\sharp}^{p.v.}(0,0)\tilde{\mathcal{H}}_{b,\sharp}(0,-1)  }
\end{pmatrix}
\begin{pmatrix}
\substack{\tilde{v}_1(0)} \\ \substack{\tilde{v}_1(-1)}
\end{pmatrix}
=0,
\end{equation*}
which is exactly the second component of the desired identity \eqref{eq_M_PV_kernel_proof_1} if one recalls the definition of $\mathcal{M}^{p.v.}$. The proof for the first component is similar, which we now need to play the integral by part trick in the left half-strip, i.e., for $n,k,p<0$; the details are skipped here.

{\color{blue}Step 2.} Now we continue to prove the other direction, i.e. 
$$\ker \mathcal{M}^{p.v.}\subset \text{span}\big\{ \big(\tilde{v}_{k}(0),\tilde{v}_{k}(-1) \big)^{\top},\,1\leq k\leq 4 \big\},
$$
which, together with \eqref{eq_M_PV_kernel_proof_1}, concludes the proof of \eqref{eq_M_PV_kernel}. Recalling that the sesquilinear form $\mathfrak{a}$ orthogonalizes the space at the right side, it is sufficient to prove the following statement:
\begin{equation} \label{eq_M_PV_kernel_proof_3}
\mathcal{M}^{p.v.}
\begin{pmatrix}
\phi_{0} \\ \phi_{-1}
\end{pmatrix}=0 ,\quad
\mathfrak{a}(\phi,\tilde{v}_k;0)=0
\Longrightarrow
\begin{pmatrix}
\phi_{0} \\ \phi_{-1}
\end{pmatrix}=0,
\end{equation}
where $\phi:=\mathbbm{1}_{\{0\}}\otimes \phi_{0}+\mathbbm{1}_{\{-1\}}\otimes \phi_{-1}\in \tilde{\mathcal{X}}_{\sharp}$ (i.e., extension by zero). Suppose that the conditions of \eqref{eq_M_PV_kernel_proof_3} hold true. Then we construct the following function $u_{\phi}\in \tilde{\mathcal{X}}_{\sharp}$ by a layer-potential expression:
\begin{equation} \label{eq_M_PV_kernel_proof_4}
u_\phi (n):=\left\{
\begin{aligned}
&\tilde{\mathcal{G}}_{b,\sharp}^{p.v.}(n,-1)\tilde{\mathcal{H}}_{b,\sharp}(-1,0)\phi_{0}-\tilde{\mathcal{G}}_{b,\sharp}^{p.v.}(n,0)\tilde{\mathcal{H}}_{b,\sharp}(0,-1)\phi_{-1},\quad n\geq 0, \\
&-\tilde{\mathcal{G}}_{b,\sharp}^{p.v.}(n,-1)\tilde{\mathcal{H}}_{b,\sharp}(-1,0)\phi_{0}+\tilde{\mathcal{G}}_{b,\sharp}^{p.v.}(n,0)\tilde{\mathcal{H}}_{-b,\sharp}(0,-1)\phi_{-1},\quad n< 0 . \\
\end{aligned}
\right.
\end{equation}
By definition \eqref{eq_M_PV_kernel_proof_4}, it is clear that
\begin{equation} \label{eq_M_PV_kernel_proof_5}
\Big[(\tilde{\mathcal{H}}_{b,\sharp}-\lambda_*)u_\phi\Big](n)=0\quad \text{for all $n\neq 0,-1$}.
\end{equation}
Importantly, the assumption \eqref{eq_M_PV_kernel_proof_3} implies that \eqref{eq_M_PV_kernel_proof_5} also holds for $n=0,-1$. In fact,
\begin{equation} \label{eq_M_PV_kernel_proof_6}
\begin{aligned}
&\Big[(\tilde{\mathcal{H}}_{b,\sharp}-\lambda_*)u_\phi\Big](0) \\
&=\tilde{\mathcal{H}}_{b,\sharp}(0,1)u_{\phi}(1)+\tilde{\mathcal{H}}_{b,\sharp}(0,0)u_{\phi}(0)+\tilde{\mathcal{H}}_{b,\sharp}(0,-1)u_{\phi}(-1)-\lambda_*u_{\phi}(0) \\
&=\tilde{\mathcal{H}}_{b,\sharp}(0,1)\tilde{\mathcal{G}}_{b,\sharp}^{p.v.}(1,-1)\tilde{\mathcal{H}}_{b,\sharp}(-1,0)\phi_{0}-\tilde{\mathcal{H}}_{b,\sharp}(0,1)\tilde{\mathcal{G}}_{b,\sharp}^{p.v.}(1,0)\tilde{\mathcal{H}}_{b,\sharp}(0,-1)\phi_{-1} \\
&\quad + \tilde{\mathcal{H}}_{b,\sharp}(0,0)\tilde{\mathcal{G}}_{b,\sharp}^{p.v.}(0,-1)\tilde{\mathcal{H}}_{b,\sharp}(-1,0)\phi_{0}-\tilde{\mathcal{H}}_{b,\sharp}(0,0)\tilde{\mathcal{G}}_{b,\sharp}^{p.v.}(0,0)\tilde{\mathcal{H}}_{b,\sharp}(0,-1)\phi_{-1} \\
&\quad -\tilde{\mathcal{H}}_{b,\sharp}(0,-1)\tilde{\mathcal{G}}_{b,\sharp}^{p.v.}(-1,-1)\tilde{\mathcal{H}}_{b,\sharp}(-1,0)\phi_{0}+\tilde{\mathcal{H}}_{b,\sharp}(0,-1)\tilde{\mathcal{G}}_{b,\sharp}^{p.v.}(-1,0)\tilde{\mathcal{H}}_{b,\sharp}(0,-1)\phi_{-1} \\
&\quad -\lambda_* \tilde{\mathcal{G}}_{b,\sharp}^{p.v.}(0,-1)\tilde{\mathcal{H}}_{b,\sharp}(-1,0)\phi_{0}+\lambda_*\tilde{\mathcal{G}}_{b,\sharp}^{p.v.}(0,0)\tilde{\mathcal{H}}_{b,\sharp}(0,-1)\phi_{-1} \\
&\overset{(i)}{=}-2\tilde{\mathcal{H}}_{b,\sharp}(0,-1)\tilde{\mathcal{G}}_{b,\sharp}^{p.v.}(-1,-1)\tilde{\mathcal{H}}_{b,\sharp}(-1,0)\phi_{0}+2\tilde{\mathcal{H}}_{b,\sharp}(0,-1)\tilde{\mathcal{G}}_{b,\sharp}^{p.v.}(-1,0)\tilde{\mathcal{H}}_{b,\sharp}(0,-1)\phi_{-1} \\
&\quad -\tilde{\mathcal{H}}_{b,\sharp}(0,-1)\phi_{-1} \\
&\overset{(ii)}{=} 0,
\end{aligned}
\end{equation}
where (i) is a consequence of the following identity (equivalent to \eqref{eq_physical_green_right_inverse}):
\begin{equation*}
\sum_{k\in\mathbb{Z}:|k-n|\leq 1}\big(\tilde{\mathcal{H}}_{b,\sharp}-\lambda_*\big)(n,k)\tilde{\mathcal{G}}_{b,\sharp}^{p.v.}(k,m)=\delta_{n,m}
\end{equation*}
and (ii) follows directly from our assumption $\mathcal{M}^{p.v.}
\begin{pmatrix}
\phi_{0} \\ \phi_{-1}
\end{pmatrix}=0$. The verification of \eqref{eq_M_PV_kernel_proof_5} for $n=-1$ is similar. Hence, we conclude that $u_{\phi}$ solves \eqref{eq_M_PV_kernel_proof_5} for all $n\in\mathbb{Z}$. On the other hand, the condition $\mathfrak{a}(\phi,\tilde{v}_k;0)=0$ together with the far-field asymptotics of the physical Green operator present in Proposition \ref{prop_physical_green} tells $u_\phi(n)\to 0$ as $n\to\infty$. In conclusion, we have argued that $u_\phi$ constructed in \eqref{eq_M_PV_kernel_proof_4} is an $\ell^2$ eigenmode of the strip Hamiltonian $\tilde{\mathcal{H}}_{b,\sharp}$ and hence, $\lambda_*$ is the associated eigenvalue. However, this contradicts the absolute continuity of the spectrum of the periodic operator $\tilde{\mathcal{H}}_{b,\sharp}$ \cite{Sobolev02absolute}. Hence, it must hold that
\begin{equation*}
    u_{\phi}(n)\equiv 0,\quad \text{for all $n\in\mathbb{Z}$} .
\end{equation*}
But then, the only remaining term in \eqref{eq_M_PV_kernel_proof_6} is
\begin{equation}  \label{eq_M_PV_kernel_proof_7}
-\tilde{\mathcal{H}}_{b,\sharp}(0,-1)\phi_{-1} =0 .
\end{equation}
Similarly, one can prove that
\begin{equation} \label{eq_M_PV_kernel_proof_8}
\tilde{\mathcal{H}}_{b,\sharp}(-1,0)\phi_{0} =0,
\end{equation}
which implies that $\phi_{-1}=\phi_{1}=0$ because both $\tilde{\mathcal{H}}_{b,\sharp}(0,-1)$ and $\tilde{\mathcal{H}}_{b,\sharp}(-1,0)$ are non-singular\footnote{Considering assumption \eqref{eq_non_singular_hopping} in Theorem \ref{thm_existence_interface_modes}, one can check that these matrices are upper-diagonal with non-singular diagonal blocks.}. This concludes the proof.
\end{proof}

\begin{remark}
As we see in the proof of Proposition \ref{prop_M_PV_kernel}, the far-field asymptotics of the physical Green operator is mainly used in studying the kernel space of the matrix $\mathcal{M}^{p.v.}$, which consists of boundary data of propagating modes at the double Dirac point. This indicates that $\mathcal{M}^{p.v.}$ actually encodes the information on the coupling between the evanescent and the propagating waves of the unperturbed system. This information is critical for studying the interface modes, as we shall see in the next section.
\end{remark}

\subsection{Existence of Interface Modes} \label{sec_existence_interface_mode}

Now we prove the first part of Theorem \ref{thm_existence_interface_modes}, i.e., the existence of interface modes. 

{\color{blue}Step 1.} According to Proposition \ref{prop_existence_equivalence}, we first show that the characteristic value problem \eqref{eq_boundary_matching} has a nontrivial solution. This is achieved by the Gohberg-Sigal \cite{gohberg, ammari2009} theory as follows. Denote
\begin{equation} \label{eq_existence_proof_1}
\underline{\mathcal{M}}(h;\delta):=\mathcal{M}(\lambda_*+\delta\cdot h,\delta).
\end{equation}
We first examine the limit of $\underline{\mathcal{M}}(h;\delta)$ as $\delta\to 0^+$. By Proposition \ref{prop_boundary_operator_limit}, we obtain
\begin{equation*}
\lim_{\delta\to 0^+}\underline{\mathcal{M}}(h;\delta)=\mathcal{M}^{p.v.}+\xi(h)\mathcal{A}=:\underline{\mathcal{M}}(h;0).
\end{equation*}
Both matrices at the right side are Hermitian, as checked from their definitions \eqref{eq_M_PV_expression} and \eqref{eq_A_projection}; they are in fact perfectly `orthogonalized' as illustrated below. We claim that the limiting operator $\underline{\mathcal{M}}(h;0)$ has four characteristic values (up to multiplicities) within $h\in\mathcal{J}$. In fact, suppose that
\begin{equation} \label{eq_existence_proof_2}
\underline{\mathcal{M}}(h;0)
\begin{pmatrix}
\phi_0 \\ \phi_{-1}
\end{pmatrix}
=0,
\end{equation}
or equivalently, using \eqref{eq_A_projection},
\begin{equation} \label{eq_existence_proof_3}
\mathcal{M}^{p.v.}\begin{pmatrix}
\phi_0 \\ \phi_{-1}
\end{pmatrix}
+\xi(h)\sum_{1\leq k\leq 4}\overline{\mathfrak{a}(\tilde{v}_{k},\phi;0)}
\begin{pmatrix}
-\tilde{\mathcal{H}}_{b,\sharp}(0,-1)\tilde{v}_{k}(-1) \\ \tilde{\mathcal{H}}_{b,\sharp}(-1,0)\tilde{v}_{k}(0)
\end{pmatrix}
=0,
\end{equation}
with $\phi:=\mathbbm{1}_{\{0\}}\otimes \phi_{0}+\mathbbm{1}_{\{-1\}}\otimes \phi_{-1}\in \tilde{\mathcal{X}}_{\sharp}$. Then, multiplying $(\overline{\tilde{v}_k(0)}, \overline{\tilde{v}_k(-1)})$ to both sides of \eqref{eq_existence_proof_3}, recalling Proposition \ref{prop_M_PV_kernel} and Proposition \ref{prop_sesqui_form}, we conclude that
\begin{equation*}
\xi(h)\mathfrak{a}(\tilde{v}_{k},\phi;0)=0.
\end{equation*}
We have two cases: (i) $\xi(h)=0$ or (ii) $\mathfrak{a}(\tilde{v}_{k},\phi;0)=0$ for all $1\leq k\leq 4$. Suppose that (ii) is true. Then, \eqref{eq_existence_proof_2}, \eqref{eq_existence_proof_3} and Proposition \ref{prop_M_PV_kernel} show that
\begin{equation*}
\begin{pmatrix}
\phi_0 \\ \phi_{-1}
\end{pmatrix}
\in \text{span}\big\{ \big(\tilde{v}_{k}(0),\tilde{v}_{k}(-1) \big)^{\top},\,1\leq k\leq 4 \big\}.
\end{equation*}
But this is in contradiction to assumption (ii) if $(\phi_0, \phi_{-1})^{\top}$ is nontrivial. Hence, we are left with the only possibility (i). Remarkably, by the explicit form of $\xi(h)$ obtained in \eqref{eq_xi_function}, we know that $\xi(h)=0$ has a unique solution $h=0$. In conclusion, we have four characteristic solutions to \eqref{eq_existence_proof_2}
\begin{equation} \label{eq_existence_proof_4}
h^{(k)}=0,\quad \begin{pmatrix}
\phi_0^{(k)} \\ \phi_{-1}^{(k)}
\end{pmatrix}
=\begin{pmatrix}
\tilde{v}_{k}(0) \\ \tilde{v}_{k}(-1)
\end{pmatrix}
,\quad 1\leq k\leq 4.
\end{equation}
For $h\neq 0$, we know that $\xi(h)\neq 0$; then the same argument as before shows that $\underline{\mathcal{M}}(h;0)$ is invertible.

Now, we collect all the information we have learned:
\begin{itemize}
    \item[(i)] The limit of $\underline{\mathcal{M}}(h;\delta)$ as $\delta\to 0^+$, denoted by $\underline{\mathcal{M}}(h;0)$, has a unique four-fold characteristic value $h=0$ within $h\in\mathcal{J}$. For all $h\neq 0$, $\underline{\mathcal{M}}(h;0)$ is invertible;
    \item[(ii)] The operator $\underline{\mathcal{M}}(h;\delta)$ is analytic in $h\in\mathcal{J}$, which is a consequence of the analyticity of the resolvents $\tilde{\mathcal{G}}_{\pm\delta,\sharp}$ (of the bulk Hamiltonians $\tilde{\mathcal{H}}_{\pm\delta,\sharp}$) within the spectral gap $\mathcal{I}_{\delta}$;
    \item[(iii)] $\underline{\mathcal{M}}(h;\delta)$ is Hermitian.
\end{itemize}

Based on (i)-(iii), the generalized Rouché theorem (cf. \cite[Chapter 1]{ammari2018mathematical}) concludes that there exists a neighborhood of $\delta=0$ in which the operator $\underline{\mathcal{M}}(h;\delta)$ (or in other words, the problem \eqref{eq_boundary_matching}) has four characteristic values counted with their multiplicities, which are bifurcated from \eqref{eq_existence_proof_4} and are denoted as
\begin{equation} \label{eq_existence_proof_5}
h^{(k)}_{\delta}=o(1),\quad \begin{pmatrix}
\phi_{\delta,0}^{(k)} \\ \phi_{\delta,-1}^{(k)}
\end{pmatrix}
=\begin{pmatrix}
\tilde{v}_{k}(0) \\ \tilde{v}_{k}(-1)
\end{pmatrix}+o(1)
,\quad 1\leq k\leq 4.
\end{equation}
This concludes the first part of the proof.

{\color{blue}Step 2.} With \eqref{eq_existence_proof_5}, it is now direct to verify condition (2) of Proposition \ref{prop_existence_equivalence}. In fact, we only need to check that the leading-order term is nonzero. Taking, for example, $k=1$ yields
\begin{equation} \label{eq_existence_proof_6}
\begin{aligned}
&\lim_{\delta\to 0^+}\mathcal{M}^{aux}(\lambda_*+\delta\cdot h^{(1)}_{\delta},\delta)\begin{pmatrix}
\phi_{\delta,0}^{(1)} \\ \phi_{\delta,-1}^{(1)}
\end{pmatrix} \\
&\overset{(i)}{=}\big(\mathcal{M}^{aux,p.v.}+\frac{\text{sgn}(\beta_*)}{|\alpha_*|}\mathcal{A}^{aux}_2\big)
\begin{pmatrix}
\tilde{v}_{1}(0) \\ \tilde{v}_{1}(-1)
\end{pmatrix} \\
&\overset{(ii)}{=} 
\frac{1}{2}\begin{pmatrix}
\tilde{v}_{1}(0) \\ \tilde{v}_{1}(-1)
\end{pmatrix}
-\frac{\text{sgn}(\beta_*)}{2|\alpha_*|}\overline{\mathfrak{a}(\tilde{v}_{1},\tilde{v}_{1})} \begin{pmatrix}
\tilde{v}_{3}(0) \\ \tilde{v}_{3}(-1)
\end{pmatrix} \\
&\overset{(iii)}{=} 
\frac{1}{2}\begin{pmatrix}
\tilde{v}_{1}(0) \\ \tilde{v}_{1}(-1)
\end{pmatrix}
+\frac{i\cdot\text{sgn}(\beta_*)}{2}\begin{pmatrix}
\tilde{v}_{3}(0) \\ \tilde{v}_{3}(-1)
\end{pmatrix} \neq 0.
\end{aligned}
\end{equation}
Here, equalities (i) and (iii) and the last step follow by recalling the convergence \eqref{eq_M_aux_limit}, Proposition \ref{prop_sesqui_form}, and the fact that the Bloch modes are independent. The middle step (ii) is derived by the following identity:
\begin{equation*}
\mathcal{M}^{aux,p.v.}
\begin{pmatrix}
\tilde{v}_{1}(0) \\ \tilde{v}_{1}(-1)
\end{pmatrix}
=\frac{1}{2}\begin{pmatrix}
\tilde{v}_{1}(0) \\ \tilde{v}_{1}(-1)
\end{pmatrix},
\end{equation*}
where the first component of this equation is proved by \eqref{eq_M_PV_kernel_proof_9}, and the second component is checked similarly. This concludes the proof of the existence of interface modes.

\begin{remark}
In fact, the calculation in \eqref{eq_M_PV_kernel_proof_9} can be generalized to any $1\leq k\leq 4$ with the result
\begin{equation} \label{eq_existence_proof_7}
\mathcal{M}^{aux,p.v.}
\begin{pmatrix}
\tilde{v}_{k}(0) \\ \tilde{v}_{k}(-1)
\end{pmatrix}
=\frac{1}{2}\begin{pmatrix}
\tilde{v}_{k}(0) \\ \tilde{v}_{k}(-1)
\end{pmatrix}.
\end{equation}
The details are left to the interested reader.
\end{remark}

\subsection{Number of Interface Modes} \label{sec_number_interface_mode}

Now we proceed with the second part of Theorem \ref{thm_existence_interface_modes}, with regard to the precise number of interface modes and their parities. The idea is the same as in our previous works, e.g., \cite[Section 9]{qiu2024square_lattice} and \cite[Section 8]{li2024interface_mode_honeycomb}. In particular, based on the discrete layer-potential formulation of interface mode problem we have established, the analysis in the aforementioned studies in the context of continuous systems can be carried directly to the discrete setup considered here. We only show the main idea and skip the details.

The first step is to show that there are \textit{at least two independent interface modes}. In fact, it is easy to construct such two interface modes using the layer-potential formula \eqref{eq_interface_mode_layer_potential_form} and the boundary data in \eqref{eq_existence_proof_5}. We define
\begin{equation} \label{eq_number_interface_mode_proof_1}
u_{\delta}^{(k)}(n)=\left\{
\begin{aligned}
&\tilde{\mathcal{G}}_{\delta,\sharp}(\lambda_*+h_{\delta}^{(k)})(n,-1)\tilde{\mathcal{H}}_{\delta,\sharp}(-1,0)\phi_{\delta,0}^{(k)} \\
&\quad -\tilde{\mathcal{G}}_{\delta,\sharp}(\lambda_*+h_{\delta}^{(k)})(n,0)\tilde{\mathcal{H}}_{zig,\sharp}(0,-1)\phi_{\delta,-1}^{(k)},\quad n\geq 0, \\
&-\tilde{\mathcal{G}}_{-\delta,\sharp}(\lambda_*+h_{\delta}^{(k)})(n,-1)\tilde{\mathcal{H}}_{zig,\sharp}(-1,0)\phi_{\delta,0}^{(k)} \\
&\quad +\tilde{\mathcal{G}}_{-\delta,\sharp}(\lambda_*+h_{\delta}^{(k)})(n,0)\tilde{\mathcal{H}}_{-\delta,\sharp}(0,-1)\phi_{\delta,-1}^{(k)},\quad n< 0, \\
\end{aligned}
\right.
\end{equation}
for $k=1,2$. The function $u_{\delta}^{(k)}$ satisfies $\tilde{\mathcal{H}}_{zig,\sharp}u_{\delta}^{(k)}=(\lambda_*+h_{\delta}^{(k)})u_{\delta}^{(k)}$, as seen from the discussion in the last section. By a similar calculation to the one in \eqref{eq_existence_proof_6}, one can check that the values of $u_{\delta}^{(k)}(0)$ are linearly independent. In fact, the leading-order term of $u_{\delta}^{(k)}(0)$, for $k=1,2,$ are given by
\begin{equation*}
u_{\delta}^{(1)}(0)\sim \frac{1}{2}\tilde{v}_{1}(0)
+\frac{i\cdot\text{sgn}(\beta_*)}{2}\tilde{v}_{3}(0),\quad
u_{\delta}^{(2)}(0)\sim \frac{1}{2}\tilde{v}_{2}(0)
-\frac{i\cdot\text{sgn}(\beta_*)}{2}\tilde{v}_{4}(0),
\end{equation*}
and are clearly independent by the orthogonality of the Bloch modes $\tilde{v}_{k}$. Hence, \eqref{eq_number_interface_mode_proof_1} indeed gives two independent interface modes.

The second step is to show that there are \textit{at most two independent interface modes}. That is proved using the auxiliary equation \eqref{eq_boundary_matching_auxliary}. In fact, by Proposition \ref{prop_boundary_operator_limit}, the leading-order equation of \eqref{eq_boundary_matching_auxliary} as $\delta\to 0^+$ is given by
\begin{equation*}
\Big(\mathcal{M}^{aux,p.v.}+\xi(h)\mathcal{A}^{aux}_1+\eta(h)\mathcal{A}^{aux}_2-\mathbb{I} \Big)
\begin{pmatrix}
\phi_{0} \\ \phi_{-1}
\end{pmatrix}.
\end{equation*}
Using \eqref{eq_existence_proof_7}, one can show that the above equation \textit{only has two independent solutions}, instead of four, at $h=0$:
\begin{equation} \label{eq_number_interface_mode_proof_2}
\begin{array}{l}
\ds \begin{pmatrix}
\phi_{0}^{mode,(1)} \\ \phi_{-1}^{mode,(1)}
\end{pmatrix}
=\begin{pmatrix}
\tilde{v}_{1}(0) \\ \tilde{v}_{1}(-1)
\end{pmatrix}
+i\cdot\text{sgn}(\beta_*)
\begin{pmatrix}
\tilde{v}_{3}(0) \\ \tilde{v}_{3}(-1)
\end{pmatrix},\\
\nm 
\ds
\begin{pmatrix}
\phi_{0}^{mode,(2)} \\ \phi_{-1}^{mode,(2)}
\end{pmatrix}
=\begin{pmatrix}
\tilde{v}_{2}(0) \\ \tilde{v}_{2}(-1)
\end{pmatrix}
-i\cdot\text{sgn}(\beta_*)
\begin{pmatrix}
\tilde{v}_{4}(0) \\ \tilde{v}_{4}(-1)
\end{pmatrix}.
\end{array}
\end{equation}
The superscript `mode' indicates that they are leading-order terms of the boundary data of the interface modes. This confirms that there are at most two interface modes with characteristic values near $h=0$. In fact, the complete solution to \eqref{eq_boundary_matching_auxliary} (including the remainder) for small $\delta$ can be constructed using the Lyapunov-Schmidt argument, as detailed in \cite[Proposition 7.1]{li2024interface_mode_honeycomb}.

We note that the above calculations also indicate the parity of the interface modes. For example, denote by $u_{1}^{zig}$ the interface mode associated with the boundary data $(\phi_{0}^{mode,(1)},\phi_{-1}^{mode,(1)})^{\top}$. Since the interface Hamiltonian $\mathcal{H}_{zig}$ is $\mathcal{F}_x$ invariant, its eigenfunction $u_{1}^{zig}$ attains either even or odd parity with respect to $\mathcal{F}_x$. To determine the precise parity, it is sufficient to look at the leading order of the origin value $u_{1}^{zig}(\bm{0})$:
\begin{equation*}
u_{1}^{zig}(\bm{0})\sim v_{1}(\bm{0})+i\cdot\text{sgn}(\beta_*) v_{3}(\bm{0}).
\end{equation*}
Then, from \eqref{eq_relation_un_vn} and the reflection parity of the Bloch modes specified by the representation \eqref{eq_4d_irrep}, one can check that
\begin{equation*}
F_{x}^{int}\Big(v_{1}(\bm{0})+i\cdot\text{sgn}(\beta_*) v_{3}(\bm{0})\Big)
=v_{1}(\bm{0})+i\cdot\text{sgn}(\beta_*)v_{3}(\bm{0}).
\end{equation*}
This implies that $u_{1}^{zig}$ has even parity with respect to the $\mathcal{F}_x$ reflection. Similarly, one can check that the other interface mode, associated with the boundary data $(\phi_{0}^{mode,(2)},\phi_{-1}^{mode,(2)})^{\top}$, possesses odd parity. We omit the details here.

\section{Robustness of Interface Modes} \label{sec_robustness}

Now we prove the symmetry-protection of the interface modes. As indicated in the discussion following Theorem \ref{thm_robustness_interface_modes}, this is achieved by a periodic approximation argument based on exploiting the behaviour of the perturbed interface modes on a sequence of $L-$strips. Before we present the details, we need to introduce some preliminaries.

We first define the following Hilbert space:
\begin{equation*}
\mathcal{X}_{L,\sharp}:=\{u\in \ell^2_{loc}(\Lambda)\otimes \mathbb{C}^6:\, u(\bm{n}+L\bm{\ell}_2)=u(\bm{n}),\, \sum_{\substack{n_1\in\mathbb{Z} \\ 0\leq n_2\leq L-1}}\|u(n_1\bm{e}_1)\|^2<\infty \}.
\end{equation*}
(we will omit the subscripts for $\mathbb{C}^6$ vector and $\mathbb{C}^{6\times 6}$ norm hereafter.) The functions in $\mathcal{X}_{L,\sharp}$ are $L\bm{\ell}_2$-periodic. $\mathcal{X}_{L,\sharp}$ is equipped with the inner product:
\begin{equation*}
(u,v)_{\mathcal{X}_{L,\sharp}}:=\sum_{\substack{n_1\in\mathbb{Z} \\ 0\leq n_2\leq L-1}} \big(u(\bm{n}),v(\bm{n})\big)_{\mathbb{C}^6}
=\sum_{\bm{n}\in \Lambda_L} \big(u(\bm{n}),v(\bm{n})\big)_{\mathbb{C}^6},\quad \forall\, u,v\in \mathcal{X}_{L,\sharp},
\end{equation*}
where $\Lambda_L$ is a strip transecting the interface and has a width equal to $L>0$
\begin{equation*}
\Lambda_L:=\big\{\bm{n}\in \Lambda:\, |\bm{n}\cdot \bm{\ell}_2|\leq \frac{L}{2} \big\}.
\end{equation*}
Note that $\Lambda_L$ is a special isomorphic copy of the cylinder $\Lambda/L\mathbb{Z}\bm{\ell}_2$, which is chosen to be reflectional symmetric about the $x-$axis. For example, the background lattice sites in Figure \ref{fig_interface structure_perturbed} comprise $\Lambda_L$ with $L=3$. Furthermore, we introduce the subspace of $\mathcal{X}_{L,\sharp}$ with even parity with respect to the reflection operator $\mathcal{F}_{x}:$
\begin{equation*}
\mathcal{X}_{L,\sharp}^{e}:=\{u\in \mathcal{X}_{L,\sharp}:\, \mathcal{F}_{x}u=u \}.
\end{equation*}
The superscript `e' indicates that any $u\in \mathcal{X}_{L,\sharp}^{e}$ attains even parity with respect to the reflection symmetry $\mathcal{F}_{x}$. Similarly, the function space with odd parity is defined as
\begin{equation*}
\mathcal{X}_{L,\sharp}^{o}:=\{u\in \mathcal{X}_{L,\sharp}:\, \mathcal{F}_{x}u=-u \}.
\end{equation*}
It is direct to check that the unperturbed interface Hamiltonian, $\mathcal{H}_{zig}$, is invariant on both $\mathcal{X}_{L,\sharp}^{e}$ and $\mathcal{X}_{L,\sharp}^{o}$ because $\mathcal{H}_{zig}$ is reflectional and $\mathbb{Z}\bm{\ell}_2-$translational invariant; consequently, its restriction $\mathcal{H}_{zig}^{L,e/o}:=\mathcal{H}_{zig}\Big|_{\mathcal{X}_{L,\sharp}^{e/o}}$ is well-defined. However, the perturbed interface Hamiltonian $\mathcal{H}_{zig,\mathcal{W}}=\mathcal{H}_{zig}+\mathcal{W}$ is not invariant on these strip spaces due to the loss of periodicity. To define its restriction, we introduce the following restriction and periodization operators. First, the restriction operator $\mathfrak{R}^L$ simply truncates functions on the whole lattice $\Lambda$ to the strip $\Lambda_L$:
\begin{equation*}
\mathfrak{R}^L:\ell^{\infty}(\Lambda)\otimes \mathbb{C}^{6}\to \ell^{\infty}(\Lambda_L)\otimes \mathbb{C}^{6},\quad
(\mathfrak{R}^Lu)(\bm{n}):=u(\bm{n}),\quad \bm{n}\in \Lambda_L.
\end{equation*}
The periodization operator $\mathfrak{S}^L$ first restricts a function to $\Lambda_{L/2}$, then extends it periodically:
\begin{equation*}
\mathfrak{S}^L:\ell^{\infty}(\Lambda)\otimes \mathbb{C}^{6}\to \ell^{\infty}(\Lambda)\otimes \mathbb{C}^{6},\quad
(\mathfrak{S}^Lu)(\bm{n}):=\left\{
\begin{aligned}
&(\mathfrak{R}^{L/2}u)(\bm{m}),\quad \text{if $\bm{n}-\bm{m}\in L\mathbb{Z}\bm{\ell}_2$ with $\bm{m}\in \Lambda_{L/2}$,} \\
&0,\quad \text{otherwise.}
\end{aligned}
\right.
\end{equation*}
Clearly, $\mathfrak{S}^Lu$ is periodic and, due to the symmetric shape of $\Lambda_L$, preserves the $\mathcal{F}_x$-symmetry.
\begin{proposition} \label{prop_periodization operator}
For any $u\in \ell^{\infty}(\Lambda)\otimes \mathbb{C}^6$, the function $u^{L}:=\mathfrak{S}^Lu$ satisfies $u^{L}(\bm{n}+L\bm{\ell}_2)=u^{L}(\bm{n})$ for all $\bm{n}\in\Lambda$. Moreover, if $\mathcal{F}_x u=\pm u$, we have, accordingly, $\mathcal{F}_x u^{L}=\pm u^{L}$.
\end{proposition}
With these preparations, we can now define the restriction of perturbed Hamiltonians on the strips. Define the following operator:
\begin{equation*}
\mathcal{W}^{L}:=\mathfrak{S}^L \mathcal{W} \mathfrak{R}^L \in \mathcal{B}(\ell^{\infty}(\Lambda)\otimes \mathbb{C}^{6}).
\end{equation*}
Then, $\mathcal{W}^{L}$ is invariant on both $\mathcal{X}_{L,\sharp}^{e}$ and $\mathcal{X}_{L,\sharp}^{o}$ by Proposition \ref{prop_periodization operator}. Hence, the restriction of the Hamiltonian $\mathcal{H}_{zig}+\mathcal{W}^{L}$ to $\mathcal{X}_{L,\sharp}^{e/o}$ is well-defined; we denote it by
\begin{equation*}
\mathcal{H}_{zig,\mathcal{W}}^{L,e/o}:=\big(\mathcal{H}_{zig}+\mathcal{W}^{L}\big)\Big|_{\mathcal{X}_{L,\sharp}^{e/o}}.
\end{equation*}
It is also important to note that the operator $\mathcal{W}^{L}$ converges weakly to $\mathcal{W}$ in the sense that
\begin{equation*}
(\mathcal{W}^{L}u,v)_{\mathcal{X}}\to (\mathcal{W}u,v)_{\mathcal{X}} \quad \text{for any $u\in \ell^{\infty}(\Lambda)\otimes \mathbb{C}^6$ and compactly supported $v$.}
\end{equation*}

\subsection{Main Idea of Proof} \label{sec_robustness_main_idea}

The idea of proof is first to construct a sequence of perturbed interface modes, bifurcated from the unperturbed ones found in Theorem \ref{thm_existence_interface_modes} and located in the strips $\mathcal{X}_{L,\sharp}^{e/o}$. Then, we  prove the weak convergence of (a subsequence of) this sequence as $L\to \infty$, the limit of which is exactly the whole-space perturbed interface modes claimed in Theorem \ref{thm_robustness_interface_modes}. As one may expect, the key point of this compactness argument is the uniform boundedness of perturbed interface modes on the $L-$strips. This is rooted in the following lemma. Recall that we will fix the parameter $\delta$ (in the interface Hamiltonian $\mathcal{H}_{zig}$) from now on. To emphasize the parity of the interface mode, we will denote the unperturbed eigenpairs by $\lambda_{zig,e/o}:=\lambda_{zig,1/2}$ and  $u_{zig,e/o}:=u_{zig,1/2}$.

\begin{lemma} \label{lem_isolation_strip_spectrum}
Suppose that the conditions in Theorem \ref{thm_existence_interface_modes} hold. Then, $\mathcal{I}_{\delta}$ is an essential gap of $\mathcal{H}_{zig}^{L,e/o}$ for all $L>0$, in the sense that $\text{Spec}_{ess}(\mathcal{H}_{zig}^{L,e/o})\cap \mathcal{I}_{\delta}=\emptyset$. Moreover, $\mathcal{H}_{zig}^{L,e/o}$ has a unique eigenvalue inside $\mathcal{I}_{\delta}$:
\begin{equation} \label{eq_isolation_strip_spectrum}
\text{Spec}_{disc}(\mathcal{H}_{zig}^{L,e/o})\cap \mathcal{I}_{\delta} = \{\lambda_{zig,e/o}\}.
\end{equation}
Consequently, denoting the projection of $\mathcal{H}_{zig}^{L,e/o}$ to its in-gap eigenvalue by 
$$
\mathbb{P}_{zig}^{L,e/o}:=\mathbbm{1}_{\{\lambda_{zig,e/o}\}}(\mathcal{H}_{zig}^{L,e/o}),
$$
then the Fredholm operator $$\mathbb{Q}_{zig}^{L,e/o}:=(\mathbb{P}_{zig}^{L,e/o})^{\perp}(\mathcal{H}_{zig}^{L,e/o}-\lambda_{zig,e/o})(\mathbb{P}_{zig}^{L,e/o})^{\perp}$$ is invertible on $\text{Ran}(\mathbb{P}_{zig}^{L,e/o})^{\perp}$ with the uniform bound:
\begin{equation*}
\big\| (\mathbb{Q}_{zig}^{L,e/o})^{-1} \big\|_{\text{Ran}(\mathbb{P}_{zig}^{L,e/o})^{\perp}\to \text{Ran}(\mathbb{P}_{zig}^{L,e/o})^{\perp}} \leq d_{zig,e/o}^{-1},
\end{equation*}
where the isolation distance $d_{zig,e/o}:=d_{zig,1/2}>0$ (defined in \eqref{eq_isolation_distance_interface_eigenvalue}) is independent of $L$.
\end{lemma}

In the sequel, we always assume that the conditions in Theorem \ref{thm_existence_interface_modes} hold, which will no longer be repeated. Based on Lemma \ref{lem_isolation_strip_spectrum}, the reflection symmetry of $\mathcal{W}$, and the fact that the perturbed Hamiltonian $\mathcal{H}_{zig,\mathcal{W}}^{L,e/o}$ acts invariantly on $\mathcal{X}_{L,\sharp}^{e/o}$, the standard perturbation theory indicates that the interface eigenvalue in the strip persists if the size of the perturbation is suitably controlled. Moreover, using the Lyapunov-Schmidt reduction argument, we are able to explicitly construct the perturbed interface mode for each $L>0$. 
\begin{proposition} \label{prop_strip_interface_modes}
Suppose that the perturbation operator $\mathcal{W}$ is of class (A). Then, there exists $c_{\mathcal{W}}\in (0,\frac{1}{2})$ such that whenever the constant $M_{\mathcal{W}}$ defined in \eqref{eq_longi_loc_cond} satisfies
\begin{equation*}
M_{\mathcal{W}}<c_{\mathcal{W}}d_{zig,e/o},
\end{equation*}
the perturbed strip Hamiltonian $\mathcal{H}_{zig,\mathcal{W}}^{L,e/o}$ has a unique in-gap eigenvalue $\lambda_{zig,\mathcal{W}}^{L,e/o}\in\mathcal{I}_{\delta}$ with the estimate
\begin{equation} \label{eq_perturbed_interface_eigenvalue_bound}
|\lambda_{zig,\mathcal{W}}^{L,e/o}- \lambda_{zig,e/o}|<\frac{1}{2}d_{zig,e/o}.
\end{equation}
The associated eigenmode can be explicitly expressed as
\begin{equation} \label{eq_strip_interface_mode_ansatz}
\begin{aligned}
u_{zig,\mathcal{W}}^{L,e/o}
&=u_{zig,e/o}-
\Big[ 1+(\mathbb{Q}_{zig}^{L,e/o})^{-1}\big( \mathcal{W}^{L}-\lambda_{zig,\mathcal{W}}^{L,e/o}+\lambda_{zig,e/o} \big) \Big]^{-1}\mathbb{Q}_{zig}^{L,e/o}\mathcal{W}^{L} u_{zig,e/o} \\
&=: u_{zig,e/o}+u_{zig,\mathcal{W}}^{L,e/o,(1)} .
\end{aligned}
\end{equation}
In particular, $u_{zig,\mathcal{W}}^{L,e/o,(1)}$ is uniformly bounded in the sense that there exists $C>0$ such that
\begin{equation} \label{eq_strip_interface_mode_uniform_bound}
\|u_{zig,\mathcal{W}}^{L,e/o,(1)}\|_{\mathcal{X}_{L,\sharp}}\leq C<\infty ,\quad \text{for all $L>0$}.
\end{equation}
\end{proposition}

Note that \textit{the leading-order part $u_{zig,e/o}$ in \eqref{eq_strip_interface_mode_ansatz} is independent of $L$}. Hence, we only need to consider the convergence of the perturbation part $u_{zig,\mathcal{W}}^{L,e/o,(1)}$. Thanks to its uniform bound, it is possible to extract a subsequence $L_i$ by a compactness argument such that $u_{zig,\mathcal{W}}^{L_i,e/o,(1)}$ converges weakly in a suitable sense to a bound state, which is exactly the scattering wave in Theorem \ref{thm_robustness_interface_modes}.

\begin{proposition} \label{prop_convergence_scattering_modes}
There exists a subsequence $L_i$ such that the perturbed eigenvalue $\lambda_{zig,\mathcal{W}}^{L_i,e/o}$ converges to $\lambda_{zig,\mathcal{W}}^{e/o}\in \mathcal{I}_{\delta}$, and the scattering state $u_{zig,\mathcal{W}}^{L_i,e/o,(1)}$ defined in \eqref{eq_strip_interface_mode_ansatz} converges weakly to $u_{zig,\mathcal{W}}^{e/o,(1)}\in \ell^2(\Lambda)\otimes \mathbb{C}^6$ in the sense that
\begin{equation*}
(u_{zig,\mathcal{W}}^{L_i,e/o,(1)},v)_{\mathcal{X}}\to (u_{zig,\mathcal{W}}^{e/o,(1)},v)_{\mathcal{X}} \quad \text{for any compactly supported $v\in \mathcal{X}$.}
\end{equation*}
Moreover, the function $u_{zig,\mathcal{W}}^{e/o}:=u_{zig,e/o}+u_{zig,\mathcal{W}}^{e/o,(1)}$ solves the following equation for all $\bm{n}\in\Lambda$:
\begin{equation} \label{eq_limiting_full_wave_eigenmode}
\Big[\big(\mathcal{H}_{zig,\mathcal{W}}-\lambda_{zig,\mathcal{W}}^{e/o}\big)u_{zig,\mathcal{W}}^{e/o}\Big](\bm{n})=0.
\end{equation}
\end{proposition}
This concludes the proof of Theorem \ref{thm_robustness_interface_modes}. In the rest of this section, we prove Lemma \ref{lem_isolation_strip_spectrum} and Propositions \ref{prop_strip_interface_modes} and \ref{prop_convergence_scattering_modes}. We only consider the case of even parity, i.e., taking the index `e' in the above paragraphs. The proof for the odd-parity case is the same.

\subsection{Proof of Lemma \ref{lem_isolation_strip_spectrum} and Proposition \ref{prop_strip_interface_modes}} \label{sec_uniform_norm_estimate}

\begin{proof} [Proof of Lemma \ref{lem_isolation_strip_spectrum}]
The key observation is that, by the decomposition $\mathcal{X}_{L,\sharp}=\mathcal{X}_{L,\sharp}^{e}\oplus \mathcal{X}_{L,\sharp}^{o}$, we have
\begin{equation} \label{eq_isolation_strip_spectrum_proof_1}
\text{Spec}(\mathcal{H}_{zig}^{L,e})=\text{Spec}\big(\mathcal{H}_{zig}\big|_{\mathcal{X}_{L,\sharp}^{e}}\big)\subset \text{Spec}\big(\mathcal{H}_{zig}\big|_{\mathcal{X}_{L,\sharp}}\big).
\end{equation}
Thanks to the periodicity of $\mathcal{H}_{zig}$ along $\mathbb{Z}\bm{\ell}_2$, the spectrum $\text{Spec}\big(\mathcal{H}_{zig}\big|_{\mathcal{X}_{L,\sharp}}\big)$ is decomposed by the discrete Fourier transform:
\begin{equation*}
\text{Spec}\big(\mathcal{H}_{zig}\big|_{\mathcal{X}_{L,\sharp}}\big)=\bigcup_{\kappa_{\parallel}\in\mathcal{A}_L}\text{Spec}\big(\mathcal{H}_{zig}\big|_{\mathcal{X}_{\kappa_{\parallel}}}\big),
\end{equation*}
where $\mathcal{X}_{\kappa_{\parallel}}$ is defined in \eqref{eq_quasi-periodic_strip_space_def}, and the discrete set $\mathcal{A}_L$ of parallel momentum is defined as
\begin{equation*}
\mathcal{A}_L:=\big\{\kappa_{\parallel}=\frac{2n\pi}{L}:\,-\frac{L}{2}\leq n \leq \frac{L}{2},\, n\in\mathbb{Z} \big\}.
\end{equation*}
In particular, for each $\kappa_{\parallel}\in\mathcal{A}_L$, $\mathcal{H}_{zig}\big|_{\mathcal{X}_{\kappa_{\parallel}}}$ has only discrete spectrum within the bulk spectral gap $\mathcal{I}_{\delta}$ (see Figure \ref{fig_disctre_spectrum_strip_space}(a) for an illustration), which can be proved by a standard argument based on the Weyl criteria. Hence, \eqref{eq_isolation_strip_spectrum_proof_1} concludes that $\mathcal{H}_{zig}^{L,e}$ has only discrete spectrum within $\mathcal{I}_{\delta}$, which are also eigenvalues of $\mathcal{H}_{zig}\big|_{\mathcal{X}_{\kappa_{\parallel}}}$ for some $\kappa_{\parallel}\in \mathcal{A}_L$. Remarkably, in order to find an eigenvalue of $\mathcal{H}_{zig}\big|_{\mathcal{X}_{\kappa_{\parallel}}}$ with its associated eigenfunction attaining even parity, we must have $\kappa_{\parallel}=0,\pm\pi$. In fact, if $u\in \mathcal{X}_{\kappa_{\parallel}}$ is such an eigenfunction written as $u(\bm{n})=e^{i\kappa_{\parallel}n_2}u_{per}(\bm{n})$ with the periodic part $u_{per}(\bm{n}+\bm{\ell}_2)=u_{per}(\bm{n})$, then it holds
\begin{equation*}
\mathcal{F}_{x}u=u\quad \Longrightarrow \quad (\mathcal{F}_{x} u_{per})(\bm{n})=e^{i2\kappa_{\parallel}n_2}u_{per}(\bm{n}),
\end{equation*}
which immediately leads to a contradiction to the periodicity of $u_{per}$ if $2\kappa_{\parallel}\neq 0,\pm 2\pi$. Then, together with the Assumption \ref{eq_absence_eigenvalue_pi}, we conclude that the only in-gap eigenvalue of $\mathcal{H}_{zig}^{L,e}$ is that of $\mathcal{H}_{zig}\big|_{\mathcal{X}_{\sharp}}$ with even parity, i.e. \eqref{eq_isolation_strip_spectrum} holds. The rest of Lemma \ref{lem_isolation_strip_spectrum} follows by standard estimate of reduced resolvent of self-adjoint operators (cf., e.g., \cite[Chapter I, Section 3]{kato2013perturbation}).

\begin{figure}
\centering
\subfigure[$\text{Spec}\big(\mathcal{H}_{zig}\big|_{\mathcal{X}_{L,\sharp}}\big)$]{
\begin{tikzpicture}[scale=0.3]

\draw[thick,->] (-8,0)--(8,0);
\draw[thick] (0,0.3)--(0,-0.3);
\node[below] at (0,0) {$0$};
\node[right] at (8.2,0) {$\kappa_{\parallel}$};

\draw[white,line width=0pt, name path=one] plot [smooth] coordinates {(-8,9) (-6,8.8) (-2,8.2) (0,8) (2,8.2) (6,8.8) (8,9)};
\draw[white,line width=0pt, name path=two] plot [smooth] coordinates {(-8,3) (-6,3.2) (-2,3.8) (0,4) (2,3.8) (6,3.2) (8,3)};
\draw[white,line width=0pt, name path=three] (-8,9)--(8,9);
\draw[white,line width=0pt, name path=four] (-8,3)--(8,3);
\tikzfillbetween[
    of=one and three,split
  ] {pattern=north west lines};
\tikzfillbetween[
    of=two and four,split
  ] {pattern=north west lines};
\fill[pattern=north west lines] (-8,12) rectangle (8,9);
\fill[pattern=north west lines] (-8,3) rectangle (8,1);
\draw[red,line width=2pt,dashed] plot [smooth] coordinates {(-2,3.8) (-1.5,4) (-1,4.4) (0,6) (1,7.4) (1.5,8) (2,8.2)};
\draw[red,line width=2pt,dashed] plot [smooth] coordinates {(2,3.8) (1.5,4) (1,4.4) (0,6) (-1,7.4) (-1.5,8) (-2,8.2)};
\draw[fill=red] plot [smooth] (0,6) ellipse (0.2cm and 0.2cm);
\end{tikzpicture}
}
\subfigure[$\text{Spec}(\mathcal{H}_{zig}^{L,e})$]{
\begin{tikzpicture}[scale=0.3]

\draw[thick,->] (-8,0)--(8,0);
\draw[thick] (0,0.3)--(0,-0.3);
\node[below] at (0,0) {$0$};
\node[right] at (8.2,0) {$\kappa_{\parallel}$};

\draw[white,line width=0pt, name path=one] plot [smooth] coordinates {(-8,9) (-6,8.8) (-2,8.2) (0,8) (2,8.2) (6,8.8) (8,9)};
\draw[white,line width=0pt, name path=two] plot [smooth] coordinates {(-8,3) (-6,3.2) (-2,3.8) (0,4) (2,3.8) (6,3.2) (8,3)};
\draw[white,line width=0pt, name path=three] (-8,9)--(8,9);
\draw[white,line width=0pt, name path=four] (-8,3)--(8,3);
\tikzfillbetween[
    of=one and three,split
  ] {pattern=north west lines};
\tikzfillbetween[
    of=two and four,split
  ] {pattern=north west lines};
\fill[pattern=north west lines] (-8,12) rectangle (8,9);
\fill[pattern=north west lines] (-8,3) rectangle (8,1);
\draw[fill=red] plot [smooth] (0,6) ellipse (0.2cm and 0.2cm);
\node[right] at (0,6) {$\lambda_{zig,e}$};
\end{tikzpicture}
}
\caption{(a) Inside the bulk spectral gap, $\mathcal{H}_{zig}\big|_{\mathcal{X}_{L,\sharp}}$ has only discrete spectrum, which is exactly the interface eigenvalue curve $\lambda_{zig,n}(\kappa_{\parallel})$ (obtained in Corollary \ref{corol_interface_eigenvalue_curve}) sampled at $\kappa_{\parallel}\in\mathcal{A}_L$. (b) The only in-gap eigenvalue in (a) that is invariant under reflection is the one at $\kappa_\parallel=0,\pm \pi$, which is further seen to be $\lambda_{zig,e}$ because we have imposed the even parity on the considered space and assumed \eqref{eq_absence_eigenvalue_pi}.}
\label{fig_disctre_spectrum_strip_space}
\end{figure}
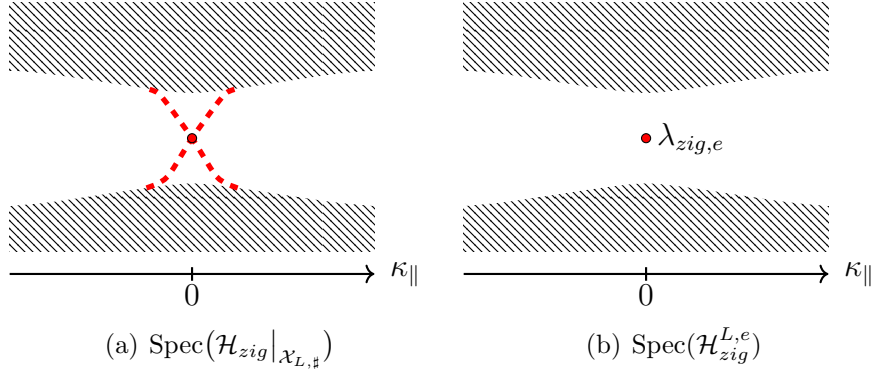

\end{proof}

\begin{proof}[Proof of Proposition \ref{prop_strip_interface_modes}]
Thanks to Lemma \ref{lem_isolation_strip_spectrum}, the existence of unique in-gap eigenvalue $\lambda_{zig,\mathcal{W}}^{L,e}$ of the perturbed strip Hamiltonian $\mathcal{H}_{zig,\mathcal{W}}^{L,e}$ follows directly by the perturbation theory (cf. \cite[Chapter VI]{kato2013perturbation}). As a consequence, the explicit construction of perturbed eigenmodes \eqref{eq_strip_interface_mode_ansatz} follows by the standard Lyapunov-Schmidt reduction argument, similar to \eqref{eq_sec3_13} in Section \ref{sec_gap_open}. Now, we prove that the perturbation part in \eqref{eq_strip_interface_mode_ansatz} is uniformly bounded. Since the projection $\mathbb{Q}_{zig}^{L,e}$ and its inverse is uniformly bounded by Lemma \ref{lem_isolation_strip_spectrum}, it is sufficient to prove that there exist $C_1,C_2>0$, which are independent of $L$, such that
\begin{equation} \label{eq_perturbed_strip_interface_modes_proof_1}
\|\mathcal{W}^{L}-\lambda_{zig,\mathcal{W}}^{L,e}+\lambda_{zig,e}\|_{\mathcal{X}_{L,\sharp}^{e}\to \mathcal{X}_{L,\sharp}^{e}}\leq C_1,
\end{equation}
and 
\begin{equation} \label{eq_perturbed_strip_interface_modes_proof_2}
\|\mathcal{W}^{L} u_{zig,e}\|_{\mathcal{X}_{L,\sharp}}\leq C_2.
\end{equation}
For \eqref{eq_perturbed_strip_interface_modes_proof_1}, we note the following inequality
\begin{equation*}
\begin{aligned}
\|\mathcal{W}^{L}\|_{\mathcal{X}_{L,\sharp}^{e}\to \mathcal{X}_{L,\sharp}^{e}}
&\overset{(i)}{\leq} \sup_{\bm{n}\in \Lambda_L}\sum_{\bm{m}\in \Lambda_L}\|\mathcal{W}(\bm{n},\bm{m})\| \leq \sup_{\bm{n}\in \Lambda}\sum_{\bm{m}\in \Lambda}\|\mathcal{W}(\bm{n},\bm{m})\| \\
&\leq \sup_{n_1\in\mathbb{Z}}\sum_{\substack{n_2\in\mathbb{Z} ,\bm{m}\in \Lambda}} \|\mathcal{W}(\bm{n},\bm{m})\| \overset{(ii)}{\leq} M_{\mathcal{W}},
\end{aligned}
\end{equation*}
where (i) is a consequence of Schur's estimate and the fact that $\mathcal{W}^{L}$ is self-adjoint, and (ii) follows by assumption \eqref{eq_longi_loc_cond}. Then \eqref{eq_perturbed_strip_interface_modes_proof_1} follows by noticing that $|\lambda_{zig,\mathcal{W}}^{L,e}-\lambda_{zig,e}|<\frac{1}{2}d_{zig,e}$ with the right side being independent of $L$. For \eqref{eq_perturbed_strip_interface_modes_proof_2}, the key is that the unperturbed interface mode $u_{zig,e}$ is $\mathbb{Z}\bm{\ell}_2$-periodic and transversely localized
\begin{equation} \label{eq_perturbed_strip_interface_modes_proof_3}
\sup_{n_2\in\mathbb{Z}}\sum_{n_1\in\mathbb{Z}}\|u_{zig,e}(n_1\bm{\ell}_1+n_2\bm{\ell}_2)\|_{\mathbb{C}^{6}}=\sum_{n_1\in\mathbb{Z}}\|u_{zig,e}(n_1\bm{\ell}_1)\|_{\mathbb{C}^{6}}<\infty,
\end{equation}
where the last inequality follows from the exponential decay of $u_{zig,e}$, proved by the standard Combes-Thomas estimate for eigenmodes associated with isolated eigenvalues (cf. \cite[Theorem 11.2]{kirsch2007invitation}). With this fact in mind, one can calculate that
\begin{equation*}
\begin{aligned}
\|\mathcal{W}^{L} u_{zig,e}\|_{\mathcal{X}_{L,\sharp}}^2
&=\sum_{\substack{n_1\in\mathbb{Z} \\ 0\leq n_2\leq L-1}} \big\|  (\mathfrak{S}^L \mathcal{W} \mathfrak{R}^Lu_{zig,e})(\bm{n})  \big\|^2
\overset{(i)}{\leq} \sum_{\bm{n}\in\Lambda} \big\|  ( \mathcal{W} \mathfrak{R}^Lu_{zig,e})(\bm{n})  \big\|^2 \\
&\leq \sum_{\bm{n}\in\Lambda} \big[ \sum_{\bm{m}\in \Lambda_L} \|\mathcal{W}(\bm{n},\bm{m})\|\|u_{zig,e}(\bm{m})\|  \big]^2 \\
&\leq \sum_{\bm{n}\in\Lambda} \big[ \sum_{m_2\in \mathbb{Z}}\sum_{m_1\in \mathbb{Z}} \|\mathcal{W}(\bm{n},\bm{m})\|\|u_{zig,e}(\bm{m})\|  \big]^2 \\
&\leq \sum_{\bm{n}\in\Lambda} \big[ \sum_{m_2\in \mathbb{Z}} \sup_{m_1\in\mathbb{Z}}\|\mathcal{W}(\bm{n},\bm{m})\|\sum_{m_1\in \mathbb{Z}}\|u_{zig,e}(\bm{m})\|  \big]^2 \\
&\leq \sum_{\bm{n}\in\Lambda} \Big\{ \big(\sum_{m_2\in \mathbb{Z}} \sup_{m_1\in\mathbb{Z}}\|\mathcal{W}(\bm{n},\bm{m})\| \big) \big( \sup_{m_2\in \mathbb{Z}}\sum_{m_1\in \mathbb{Z}}\|u_{zig,e}(\bm{m})\| \big) \Big\}^2 \\
&\overset{(ii)}{\lesssim} \sum_{\bm{n}\in\Lambda} \big(\sum_{m_2\in \mathbb{Z}} \sup_{m_1\in\mathbb{Z}}\|\mathcal{W}(\bm{n},\bm{m})\| \big)^2 \\
&\leq \big( \sum_{\bm{n}\in\Lambda} \sum_{m_2\in \mathbb{Z}} \sup_{m_1\in\mathbb{Z}}\|\mathcal{W}(\bm{n},\bm{m})\| \big)^2
\end{aligned}
\end{equation*}
where (i) is derived by the definition of periodization operator $\mathfrak{S}^L$, and (ii) follows from \eqref{eq_perturbed_strip_interface_modes_proof_3}. This estimate, together with the longitudinal localization of perturbation $\mathcal{W}$ as assumed in \eqref{eq_longi_loc_cond}, completes the proof of \eqref{eq_perturbed_strip_interface_modes_proof_2}.

\end{proof}

\subsection{Proof of Proposition \ref{prop_convergence_scattering_modes}} \label{sec_weak_compact}

The proof of Proposition \ref{prop_convergence_scattering_modes} is based on the following lemma on the weak compactness of uniformly bounded $L$-periodic sequences.
\begin{lemma} \label{lem_weak_compact_L_sequences}
Suppose that $u^{(L)}\in \mathcal{X}_{L,\sharp}$ is uniformly bounded, i.e., $\sup_{L}\|u^{(L)}\|_{\mathcal{X}_{L,\sharp}}<\infty$. Then, $u^{(L)}$ admits a subsequence $u^{(L_i)}$ that converges weakly to $u\in \mathcal{X}$ in the sense that
\begin{equation} \label{eq_weak_compact_L_sequences}
(u^{(L_i)},v)_{\mathcal{X}}\to (u,v)_{\mathcal{X}} \quad \text{for any compact supported $v\in\mathcal{X}$.}
\end{equation}
\end{lemma}
Before working on the proof on Lemma \ref{lem_weak_compact_L_sequences}, which is postponed to the end of this section, we now apply Lemma \ref{lem_weak_compact_L_sequences} to prove Proposition \ref{prop_convergence_scattering_modes}. Recall that we have found a sequence of perturbed eigenvalues $\lambda_{zig,\mathcal{W}}^{L,e}$ and ($L\mathbb{Z}\bm{\ell}_2$-periodic) eigenmodes $u_{zig,\mathcal{W}}^{L,e}=u_{zig,e}+u_{zig,\mathcal{W}}^{L,e,(1)}$ for $L\in\mathbb{N}$, which satisfy
\begin{equation} \label{eq_convergence_scattering_modes_proof_1}
\Big[\big(\mathcal{H}_{zig}+\mathcal{S}^{L}\mathcal{W}\mathcal{R}^{L}-\lambda_{zig,\mathcal{W}}^{L,e}\big)u_{zig,\mathcal{W}}^{L,e}\Big](\bm{n})=0,\quad \forall \bm{n}\in \Lambda.
\end{equation}
Now, by the uniform bound of perturbed eigenpairs (\eqref{eq_perturbed_interface_eigenvalue_bound} and \eqref{eq_strip_interface_mode_uniform_bound}), and thanks to Lemma \ref{lem_weak_compact_L_sequences}, we can extract a subsequence of eigenpairs, still indexed by $L$ for ease of notation, such that
\begin{equation} \label{eq_convergence_scattering_modes_proof_2}
\lambda_{zig,\mathcal{W}}^{L,e}\to \lambda_{zig,\mathcal{W}}^{e}\in \mathcal{I}_{\delta},
\end{equation}
and
\begin{equation} \label{eq_convergence_scattering_modes_proof_3}
u_{zig,\mathcal{W}}^{L,e} \rightharpoonup u_{zig,\mathcal{W}}^{e}=u_{zig,e}+u_{zig,\mathcal{W}}^{e,(1)}\quad \text{weakly in the sense of \eqref{eq_weak_compact_L_sequences}}.
\end{equation}
We claim that the left side of \eqref{eq_convergence_scattering_modes_proof_1} converges weakly, for which we will show the following
\begin{equation} \label{eq_convergence_scattering_modes_proof_4}
(\mathcal{H}_{zig}u_{zig,\mathcal{W}}^{L,e},v)_{\mathcal{X}}\to (\mathcal{H}_{zig}u_{zig,\mathcal{W}}^{e},v)_{\mathcal{X}},
\end{equation}
\begin{equation} \label{eq_convergence_scattering_modes_proof_5}
(\mathcal{S}^{L}\mathcal{W}\mathcal{R}^{L}u_{zig,\mathcal{W}}^{L,e},v)_{\mathcal{X}}\to (\mathcal{W}u_{zig,\mathcal{W}}^{e},v)_{\mathcal{X}},
\end{equation}
\begin{equation} \label{eq_convergence_scattering_modes_proof_6}
\lambda_{zig,\mathcal{W}}^{L,e}(u_{zig,\mathcal{W}}^{L,e},v)_{\mathcal{X}}\to \lambda_{zig,\mathcal{W}}^{e}(u_{zig,\mathcal{W}}^{e},v)_{\mathcal{X}} ,
\end{equation}
for any compactly supported $v\in \mathcal{X}$. As a consequence, we have
\begin{equation*}
\Big(\big[\mathcal{H}_{zig}+\mathcal{W}-\lambda_{zig,\mathcal{W}}^{e}\big]u_{zig,\mathcal{W}}^{e},v\Big)_{\mathcal{X}}=0.
\end{equation*}
Taking the test function as any indicator function, we conclude the proof of \eqref{eq_limiting_full_wave_eigenmode}.

Now we prove \eqref{eq_convergence_scattering_modes_proof_4}-\eqref{eq_convergence_scattering_modes_proof_6}. Note that the last convergence \eqref{eq_convergence_scattering_modes_proof_6} follows directly from \eqref{eq_convergence_scattering_modes_proof_2} and \eqref{eq_convergence_scattering_modes_proof_3}. For \eqref{eq_convergence_scattering_modes_proof_4}, we first note that\footnote{\eqref{eq_convergence_scattering_modes_proof_7} justifies nothing but the `self-adjointness' of $\mathcal{H}_{zig}$ when applying on $\ell^1-\ell^{\infty}$ pairs.}
\begin{equation} \label{eq_convergence_scattering_modes_proof_7}
\begin{aligned}
(\mathcal{H}_{zig}u_{zig,\mathcal{W}}^{L,e},v)_{\mathcal{X}}
&=\sum_{\bm{n}\in\Lambda}\sum_{\bm{m}\in\Lambda}\big(\mathcal{H}_{zig}(\bm{n},\bm{m})u_{zig,\mathcal{W}}^{L,e}(\bm{m}),v(\bm{m}) \big)_{\mathbb{C}^6} \\
&\overset{(i)}{=}\sum_{\bm{n}\in\Lambda}\sum_{\bm{m}\in\Lambda}\big(u_{zig,\mathcal{W}}^{L,e}(\bm{m}),\mathcal{H}_{zig}(\bm{m},\bm{n})v(\bm{m}) \big)_{\mathbb{C}^6} \\
&\overset{(ii)}{=}\sum_{\bm{m}\in\Lambda}\sum_{\bm{n}\in\Lambda}\big(u_{zig,\mathcal{W}}^{L,e}(\bm{m}),\mathcal{H}_{zig}(\bm{m},\bm{n})v(\bm{m}) \big)_{\mathbb{C}^6} \\
&=(u_{zig,\mathcal{W}}^{L,e},\mathcal{H}_{zig}v)_{\mathcal{X}},
\end{aligned}
\end{equation}
where the Hermiticity property $\mathcal{H}_{zig}(\bm{n},\bm{m})=\overline{\mathcal{H}_{zig}(\bm{m},\bm{n})}^{\top}$ is applied to derive (i), and the interchanging of summation in (ii) is justified by the fact that $\mathcal{H}_{zig}$ is finite-range (hence the double sum in fact only involves a single infinite sum). Since the test function $v$ is compactly supported and the Hamiltonian $\mathcal{H}_{zig}$ is finite-range, it follows that the function $\mathcal{H}_{zig}v$ is also compactly supported; hence the weak convergence of $u_{zig,\mathcal{W}}^{L,e}$ implies
\begin{equation*}
(\mathcal{H}_{zig}u_{zig,\mathcal{W}}^{L,e},v)_{\mathcal{X}}
=(u_{zig,\mathcal{W}}^{L,e},\mathcal{H}_{zig}v)_{\mathcal{X}}
\to (u_{zig,\mathcal{W}}^{e},\mathcal{H}_{zig}v)_{\mathcal{X}}
=(\mathcal{H}_{zig}u_{zig,\mathcal{W}}^{e},v)_{\mathcal{X}},
\end{equation*}
where the last step follows from an argument similar to that in \eqref{eq_convergence_scattering_modes_proof_7}. This completes the proof \eqref{eq_convergence_scattering_modes_proof_4}. The proof of \eqref{eq_convergence_scattering_modes_proof_5} is similar, for which one just keeps in mind that for any $u\in \ell^\infty(\Lambda)\otimes \mathbb{C}^6$ and compactly supported $v\in\mathcal{X}$,
\begin{equation*}
(\mathcal{R}^Lu,v)_{\mathcal{X}}=(u,v)_{\mathcal{X}},\quad (\mathcal{S}^Lu,v)_{\mathcal{X}}=(u,v)_{\mathcal{X}},\quad \text{for sufficiently large $L>0$}
\end{equation*}
by the definitions of the restriction and periodization operator. The rest of the details are left to the interested reader.

\begin{proof}[Proof of Lemma \ref{lem_weak_compact_L_sequences}]
We assume that $u^{(L)}$ is real-valued and non-negative without loss of generality, i.e. $u^{(L)}(\bm{n})\in \mathbb{R}_{+}^6$ for all $\bm{n}\in\Lambda$ with $\mathbb{R}_{+}:=[0,\infty)$

Define $w^{(L^{\prime},L)}:=\mathcal{R}^{L^{\prime}}u^{(L)}\in\mathcal{X}$ for any $(L^{\prime},L)\in \mathbb{N}\times \mathbb{N}$, that is, the restriction of $u^{(L)}$ to the $L^{\prime}$-strip. Clearly, for any fixed $L^{\prime}$, $w^{(L^{\prime},L)}$ is uniformly bounded in $L$:
\begin{equation*}
\sup_{L}\|w^{(L^{\prime},L)}\|_{\mathcal{X}}\lesssim  \sup_{L}\|u^{(L)}\|_{\mathcal{X}_{L,\sharp}}<\infty .
\end{equation*}
Hence, the weak compactness of the Hilbert space $\mathcal{X}$ indicates that there exists a subsequence $w^{(L^{\prime},L_i)}$ and $w_{L^{\prime}}\in \mathcal{X}$ such that 
\begin{equation} \label{eq_weak_compact_L_sequences_proof_1}
(w^{(L^{\prime},L_i)},v)_{\mathcal{X}}\to (w_{L^{\prime}},v)_{\mathcal{X}} \quad \text{for any $v\in\mathcal{X}$.}
\end{equation}
It is clear that the subsequential limit $w_{L^{\prime}}$ is supported in the $L^{\prime}$-strip $\Lambda_{L^{\prime}}$. Next, we select a subsequence of $\{L_i\}_{i\geq 1}$, still denoted as $\{L_i\}$ for the ease of notation, such that $w^{(L^{\prime}+1,L_i)}$ converges weakly in as $i\to\infty$ in the same sense of \eqref{eq_weak_compact_L_sequences_proof_1}, with the limit denoted by $w_{L^{\prime}+1}$. By taking a real-valued $\Lambda_{L^{\prime}}$-supported test function in \eqref{eq_weak_compact_L_sequences_proof_1}, one can verify that the subsequential limit $w_{L^{\prime}}$ is non-negative, and $w_{L^{\prime}+1}$ agrees with $w_{L^{\prime}}$ on the strip $\Lambda_{L^{\prime}}$:
\begin{equation} \label{eq_weak_compact_L_sequences_proof_2}
w_{L^{\prime}+1}(\bm{n})=w_{L^{\prime}}(\bm{n})\in \mathbb{R}_{+}^6,\quad \forall \bm{n}\in \Lambda_{L^{\prime}} .
\end{equation}
Inductively, one can construct a sequence, denoted by $\{w^{(i,L_i)}\}_{i\in\mathbb{N}}$ by a diagonal argument, whose restriction on every $j-$strip $\Lambda_j$ converges weakly to $w_{j}$. We claim that $\{w_{i}\}_{i\in \mathbb{N}}$ converges strongly in $\mathcal{X}$, with its limit denoted by $u\in \mathcal{X}$. In fact, by \eqref{eq_weak_compact_L_sequences_proof_2} and the monotone convergence theorem, it is clear that $w_{i}$ has a pointwise limit $u$; hence it is sufficient to prove that this limit is $\ell^2$ bounded. This is a consequence of the uniform boundedness of $u^{(L)}$:
\begin{equation*}
\|w_{j}\|_{\mathcal{X}}\leq \liminf_{i\to\infty}\|\mathcal{R}^{j}w^{(i,L_i)} \|_{\mathcal{X}}
=\liminf_{i\to\infty}\|\mathcal{R}^{j}u^{(L_i)} \|_{\mathcal{X}}
\leq \sup_{L}\|u^{(L)}\|_{\mathcal{X}_L}<\infty,\quad \forall j\in\mathbb{N},
\end{equation*}
where the first inequality follows by the weak lower semi-continuity of norm. Based on the convergence of $w_{i}\to u$, it follows immediately that $u^{(L_i)}$ converges weakly to $u$. In fact, for any compact supported test function $v$, we select $j\in\mathbb{N}$ such that
\begin{equation} \label{eq_weak_compact_L_sequences_proof_3}
\text{supp}(v)\subset \Lambda_j .
\end{equation}
Then, we see that
\begin{equation*}
\begin{aligned}
\lim_{i\to\infty}(u^{(L_i)},v)_{\mathcal{X}}
&\overset{(i)}{=}\lim_{i\to\infty}(\mathcal{R}^{j}u^{(L_i)},v)_{\mathcal{X}}\overset{(ii)}{=}\lim_{i\to\infty}(\mathcal{R}^{j}\mathcal{R}^{i}u^{(L_i)},v)_{\mathcal{X}} \\
&=\lim_{i\to\infty}(\mathcal{R}^{j}w^{(i,L_i)},v)_{\mathcal{X}}=(w_j,v)_{\mathcal{X}}\overset{(iii)}{=} (u,v)_{\mathcal{X}},
\end{aligned}
\end{equation*}
where the equality (i) is a consequence of \eqref{eq_weak_compact_L_sequences_proof_3}, (ii) comes from the hierarchy of the restriction operators: $\mathcal{R}^{j}=\mathcal{R}^{j}\mathcal{R}^{i}$ for $i\geq j$, and (iii) follows from \eqref{eq_weak_compact_L_sequences_proof_2}-\eqref{eq_weak_compact_L_sequences_proof_3}. This concludes the proof.
\end{proof}

\footnotesize
\bibliographystyle{plain}
\bibliography{ref}

\end{document}